\documentclass[10pt]{article}

\usepackage{amsmath,amsfonts,bm}

\def\eqref#1{equation~(\ref{#1})}

\def\Algref#1{Algorithm~\ref{#1}}

\def\1{\bm{1}}

\DeclareMathAlphabet{\mathsfit}{\encodingdefault}{\sfdefault}{m}{sl}
\SetMathAlphabet{\mathsfit}{bold}{\encodingdefault}{\sfdefault}{bx}{n}

\newcommand{\R}{\mathbb{R}}

\usepackage{fullpage}
\usepackage[utf8]{inputenc} 
\usepackage[T1]{fontenc}    
\usepackage[hidelinks]{hyperref}       
\usepackage{url}            
\usepackage{booktabs}       
\usepackage{amsfonts}       
\usepackage{nicefrac}       
\usepackage{float}
\usepackage{multirow}
\usepackage{microtype}      
\usepackage{xcolor}         
\usepackage[linesnumbered,ruled,vlined]{algorithm2e}
\usepackage{lipsum}
\usepackage{graphicx}
\usepackage{amsmath}
\usepackage{amssymb}
\usepackage{bm}
\usepackage{caption}
\usepackage{subcaption}
\usepackage{arydshln}
\usepackage{soul}
\usepackage{lipsum}
\usepackage{mathtools}

\SetKwInOut{Input}{Input}
\SetKwInOut{Output}{Output}
\SetKwComment{Comment}{// }{}

\usepackage[most]{tcolorbox}
\usepackage{float}
\usepackage{xspace}
\tcbset{
  aibox/.style={
    width=\textwidth,
    top=10pt,
    colback=white,
    colframe=black,
    colbacktitle=black,
    enhanced,
    center,
    attach boxed title to top left={yshift=-0.1in,xshift=0.15in},
    boxed title style={boxrule=0pt,colframe=white,},
    breakable,
  }
}
\newtcolorbox{AIbox}[2][]{aibox,title=#2,#1}

\hypersetup{
    colorlinks=true,
    linkcolor=purple,
    filecolor=magenta,
    urlcolor=cyan,
    citecolor=blue,
}

\renewcommand{\SetKwInOut}[2]{%
  \sbox\algocf@inoutbox{\KwSty{#2}\algocf@typo:}%
  \expandafter\ifx\csname InOutSizeDefined\endcsname\relax
    \newcommand\InOutSizeDefined{}\setlength{\inoutsize}{\wd\algocf@inoutbox}%
    \sbox\algocf@inoutbox{\parbox[t]{\inoutsize}{\KwSty{#2}\algocf@typo:\hfill}~}\setlength{\inoutindent}{\wd\algocf@inoutbox}%
  \else
    \ifdim\wd\algocf@inoutbox>\inoutsize%
    \setlength{\inoutsize}{\wd\algocf@inoutbox}%
    \sbox\algocf@inoutbox{\parbox[t]{\inoutsize}{\KwSty{#2}\algocf@typo:\hfill}~}\setlength{\inoutindent}{\wd\algocf@inoutbox}%
    \fi%
  \fi
  \algocf@newcommand{#1}[1]{%
    \ifthenelse{\boolean{algocf@inoutnumbered}}{\relax}{\everypar={\relax}}%
    {\let\\\algocf@newinout\hangindent=\inoutindent\hangafter=1\parbox[t]{\inoutsize}{\KwSty{#2}\algocf@typo:\hfill}~##1\par}%
    \algocf@linesnumbered
  }}%
\makeatother

\usepackage{amsthm}
\newtheorem{theorem}{Theorem}
\newtheorem{assumption}{Assumption}

\newtheorem{definition}{Definition}
\newtheorem{remark}{Remark}

\usepackage{tikz}
\usetikzlibrary{shapes.geometric, arrows.meta, positioning, calc, fit, backgrounds}

\newcommand{\bvec}[1]{\mathbf{#1}}
\newcommand{\calL}{\mathcal{L}}
\newcommand{\calF}{{F}}
\newcommand{\calG}{\mathcal{G}}
\newcommand{\grad}{\nabla_{\theta}}

\newcommand{\gradg}{\bvec{g}_{\theta_0}}
\newcommand{\calS}{\mathcal{S}}

\newcommand{\commenton}{1}
\ifx\commenton\undefined
\newcommand{\revision}[1]{{\color{blue}#1}}
\else
\newcommand{\revision}[1]{#1}
\fi

\title{Guiding the Recommender: Information-Aware Auto-Bidding for Content Promotion}

\author{
Yumou Liu\thanks{Shanghai Jiao Tong University}
\and
Zhenzhe Zheng\footnotemark[1] \thanks{Zhenzhe Zheng is the corresponding author.}
\and
Jiang Rong\thanks{Xiaohongshu}
\and
Yao Hu\footnotemark[3]
\and
Fan Wu\footnotemark[1]
\and
Guihai Chen\footnotemark[1]
}

\begin{document}

\maketitle

\begin{abstract} 
Modern content platforms offer paid promotion to mitigate cold start by allocating exposure via auctions. Our empirical analysis reveals a counterintuitive flaw in this paradigm: while promotion rescues low-to-medium quality content, it can harm high-quality content by forcing exposure to suboptimal audiences, polluting engagement signals and downgrading future recommendation. We recast content promotion as a dual-objective optimization that balances short-term value acquisition with long-term model improvement. 
To make this tractable at bid time in content promotion, we introduce a decomposable surrogate objective, gradient coverage, and establish its formal connection to Fisher Information and optimal experimental design. We design a two-stage auto-bidding algorithm based on Lagrange duality that dynamically paces budget through a shadow price and optimizes impression-level bids using per-impression marginal utilities. To address missing labels at bid time, we propose a confidence-gated gradient heuristic, paired with a zeroth-order variant for black-box models that reliably estimates learning signals in real time. We provide theoretical guarantees, proving monotone submodularity of the composite objective, sublinear regret in online auction, and budget feasibility. Extensive offline experiments on synthetic and real-world datasets validate the framework: it outperforms baselines, achieves superior final AUC/LogLoss, adheres closely to budget targets, and remains effective when gradients are approximated zeroth-order. These results show that strategic, information-aware promotion can improve long-term model performance and organic outcomes beyond naive impression-maximization strategies. 
\end{abstract}

\section{Introduction}

Content creation platforms, such as TikTok \cite{tt} and Xiaohongshu \cite{xhs}, have become central to modern digital landscape, with recommendation systems serving as the primary arbiters of visibility for millions of new posts daily \cite{mediastats}. A fundamental challenge within this ecosystem is ``cold start'' problem~\cite{panda2022approaches,yuan2023user}: new content with limited interaction data struggles to be accurately assessed by the platform's recommendation algorithms. Therefore, the initial exposure is often limited and stochastic, meaning that the potential high-quality content can be prematurely pruned if it fails, by chance, to resonate with its initial, algorithmically-assigned audience. This dynamic is particularly acute for content platforms, where the lifecycle of a post is ephemeral, often lasting less than 24 hours, in stark contrast to long-lived items, such as products in e-commerce platforms.

To empower content creators and mitigate the algorithmic lottery of the cold start, platforms offer paid promotion services (such as ``Shutiao'' in Xiaohongshu \cite{Xiaohongshu} and TikTok content promotion \cite{TikTok}). These services allow creators to financially secure initial impressions, transforming them from passive recipients of algorithmic judgment to active agents who can influence their content's distribution. This paid intervention mirrors the mechanics of online advertising \cite{aggarwal2024auto}, where an auction is employed to efficiently allocate the scarce resource of user attention. 
However, our empirical analysis reveals a critical and counterintuitive flaw in this paradigm: while paid promotion can rescue low-to-medium quality content, it often has a negative long-term impact on the high-quality content. By forcing distribution to a broader, less optimal audience, the paid campaign can ``pollute'' the content's performance metrics with low-engagement signals, causing the recommendation algorithm to downgrade its initially high assessment.
This finding motivates our work. We posit that a naive bidding strategy borrowed from online advertising focused solely on maximizing immediate impressions is suboptimal and potentially harmful. A truly effective strategy must instead adopt a more sophisticated, dual-objective approach: balancing the short-term goal of acquiring high-value engagement with the long-term goal of strategically reducing the recommendation model's uncertainty. By providing the model with informative training data to reduce uncertainty, a creator can improve its ability to recognize their high-quality content organically in the future. 

However, designing and implementing such a strategy presents several non-trivial technical challenges.
The first challenge arises from {formulating a tractable long-term objective}. 
The goal of reducing model uncertainty is formally quantified by information-theoretic approaches like A-/D-/I-optimal design with Fisher Information Matirx (FIM)-related measures \cite{huan2024optimal}.
However, directly optimizing a FIM-based objective is computationally prohibitive \cite{allen2021near} in a real-time bidding environment. 
The calculation involves the set of all impressions won so far, making it non-decomposable and requiring expensive matrix operations (e.g., inversion) for every potential bid.\footnote{Please refer to the details in Section \ref{sec:oed}.} 
This is infeasible given the millisecond-scale latency requirements~\cite{grbovic2013large} of auctions in online platforms. 

Second, the bidding strategy must address the challenge of balancing the conflicting objectives under a budget. The short-term goal of maximizing immediate value (exploitation) often favors bidding on ``safe'' impressions where the model is already confident. In contrast, the long-term goal of uncertainty reduction (exploration) favors bidding on ``risky'' or novel impressions where the model is uncertain, which may not yield immediate clicks. This inherent conflict must be managed dynamically during the bidding process. The difficulty is compounded by the online nature of auctions, where impression opportunities arrive sequentially. The bidder must make irrevocable decisions with a finite budget and no foresight into the quality or cost of future impressions, making it exceptionally challenging to design a pacing mechanism that allocates budget intelligently over the entire campaign.

Finally, a core difficulty lies in real-time gradient estimation with missing labels. Objectives aimed at reducing model uncertainty should quantify the informativeness of a potential impression, which is represented by its loss gradient. However, computing this gradient requires the ground-truth label (i.e., whether a user will click), which is unknown at the moment of bidding. This fundamental ``missing label'' problem makes direct calculation impossible. Simple approximations, such as taking an expectation over the predicted probability, are brittle and can be highly inaccurate, especially for the most informative samples where the model is confident but incorrect. The challenge is to devise a robust heuristic that can accurately estimate this essential gradient in real-time, without the true label, to guide the bidding decision.

To address these challenges, we propose a principled and integrated bidding framework. To create a tractable objective, we introduce a novel surrogate called ``gradient coverage,'' which maximizes the similarity between the gradients of the acquired impressions and a representative set of validation data. To balance this long-term goal with short-term value acquisition under a budget, we develop a two-stage bidding algorithm based on Lagrange duality \cite{karlsson2021adaptive}. This framework uses a dual variable, acting as a dynamic shadow price for the budget, to control spending and optimally solve the composite objective. Finally, to overcome the missing label issue, we design a practical confidence-gated heuristic that uses the model's own prediction entropy to either assign a high exploration value to uncertain impressions or, for confident ones, to approximate the true gradient by selecting the hypothetical gradient (click vs. no-click) with the smaller L2-norm.

Our main technical contributions in this work are as follows:
\begin{itemize}
    \item \textbf{Modeling:} We formulate the bidding problem for content promotion as a dual-objective optimization that balances short-term predicted Click-through Rate (pCTR) maximization with long-term model improvement. We introduce a novel and computationally tractable ``gradient coverage'' function as a surrogate for reducing model uncertainty.
    
    \item \textbf{Algorithm Design:} We propose a two-stage bidding framework that uses Lagrange duality for campaign-level budget pacing and impression-level bid optimization. A key component is our confidence-gated heuristic for real-time gradient estimation without labels.
    
    \item \textbf{Theoretical Analysis:} We provide a rigorous analysis of our framework. We prove the formal relationship between our surrogate objective and the I-optimal experimental design, establish the submodularity of our composite objective function, and derive sublinear regret and budget feasibility guarantees for our online algorithm.
    
    \item \textbf{Evaluation:} We conduct extensive experiments that first validate each component of our method individually and then demonstrate the superior performance of the end-to-end framework on both synthetic and large-scale real-world datasets, confirming its ability to improve long-term model performance more effectively than standard bidding strategies.
\end{itemize}

\section{Preliminaries}
\subsection{Content Promotion Paradigm}
In the creator economy of online platforms, creators earn profit from user engagement with their content notes. Typically, a creator uploads content note, and the platform's recommendation system matches it with users. Creators face a strategic decision: either rely solely on the platform's organic recommendation, which is based on content quality and relevance, or pay to sponsor their content notes for additional impressions, a service known as \textit{Content Promotion}.\footnote{This work is based on Shutiao, which is the content promotion service in Xiaohongshu. ``Shutiao'' and ``content promotion'' will be used interchangeably.} This introduces a monetary dimension to a creator's strategy, allowing them to actively purchase impressions to improve their key performance indicators (KPIs) of a specific content.

The motivation for creators to participate in content promotion fundamentally differs from that of traditional advertisers. While advertisers typically seek direct, short-term returns, creators using content promotion are focused on the performance of the content throughout its full life cycle. They are less concerned with the immediate return on spending and more with the long-term effect of the promotion, specifically, how it can improve a single content's KPIs and subsequent organic reach after the content promotion campaign has concluded.
This forward-looking behavior is crucial: conventional advertisers in Online Advertising value immediate, campaign-level returns, whereas paid creators in content promotion value the total-lifecycle success of an individual piece of content note. 
To allocate the scarce resource of promotional impression, online platforms employ auctions, a well-established method for efficient allocation in a competitive environment \cite{myerson1981optimal}. 
This transforms promotional opportunities into a marketplace where creators bid for impressions. 
Our work is situated in this context, focusing on designing an optimal bidding strategy for creators, maximizing the KPIs of content throughout the life cycle.

\paragraph{System Model.}
We now formalize the environment in content promotion. The environment consists of a set of creators $\mathcal{S}$ participating in the content promotion program. Each creator $i \in \mathcal{S}$ produces a piece of content note, represented by a feature vector $\mathbf{x}_i$. The platform employs a pCTR (\textit{i.e.}, predicted Click-Through Rate) model, $\mathcal{M}$, to estimate the relevance of content to a user. For content $i$, the model predicts its click probability $\hat{\sigma}_i = \mathcal{M}(\mathbf{x}_i)$. This prediction is the platform's estimate of the true, unknown CTR, $\sigma_i$, and the model's predictive error is a primary source of uncertainty.

When a promotional impression slot becomes available, the platform conducts an auction among the creators $\mathcal{S}$. Each creator $i$ submits a per-click bid $b_i$. The ranking is based on an eCPM-like (expected Cost Per Mille) score that combines the content's quality with the bid:
\( \text{score}(i) = \hat{\sigma}_i \cdot b_i. \)
The slot is allocated to the creator $i^*$ with the highest score, $i^* = \arg\max_{i \in \mathcal{S}} \text{score}(i)$, with ties broken randomly. 
\revision{The platform employs a first-price, pay-per-impression payment rule, reflecting the dominant standard in modern online platforms \cite{Googlefirstprice}. A winning creator $i^*$ pays her submitted bid, $p_{i^*} = \hat{\sigma}_{i^*}\cdot b_{i^*}$. Otherwise, the payment is zero.
However, we emphasize that our core contribution, measuring the information of an impression, is mechanism-agnostic.\footnote{While we derive our bidding strategy for the more challenging first-price setting (which requires bid shading \cite{zhou2021efficient}), our method for estimating the ``uncertainty reduction value'' ($\Delta_t$) can be directly applied to second-price (VCG) auction mechanisms with only minor modifications to the final bid calculation (Please refer to the details in Remark \ref{rem:generalize-second}).}}

\paragraph{Problem Formulation.}
Each creator $i$ has a private value $v_i$ for a click, representing the benefit they can derive from that engagement. A creator's objective is to maximize her utility. When creator $i \in \mathcal{S}$ wins an auction and receives a click, her short-term utility from that impression is her value net of cost, $u_i = v_i - b_i$. 

Beyond this immediate utility, a creator has a long-term objective. 
The inherent uncertainty in the platform's pCTR model leads to noisy and unreliable rewards, which can hinder a creator's ability to systematically improve the performance of their content throughout the life cycle. 
Therefore, creators can use the paid promotion service not just for immediate reach, but also to provide the platform with valuable training data. 
By strategically winning impressions, a creator can help reduce model uncertainty. This, in turn, improves the platform's ability to accurately assess the creator's content, leading to better performance in the future.\footnote{Please refer to a toy model illustrating this in Appendix \ref{sec:insight-toy-model}.}
The central problem for a creator is to design a bidding strategy that optimally balances these dual objectives, short-term value acquisition and long-term model uncertainty reduction, within a given campaign budget. Our work focuses on developing such a strategy.

\subsection{The Dilemma of Content Promotion}
To ground our work in real-world observations and motivate our approach, we first analyze the performance of a typical content promotion service, Shutiao in Xiaohongshu~\cite{Xiaohongshu}. Our findings reveal a critical, counterintuitive flaw that forms the central problem this work addresses.

\subsubsection{An Empirical Performance Dichotomy in Content Promotion}\label{sec:measurement}
\begin{figure}[t]
     \centering
     \begin{subfigure}[t]{0.495\textwidth}
         \centering
         \includegraphics[width=\textwidth]{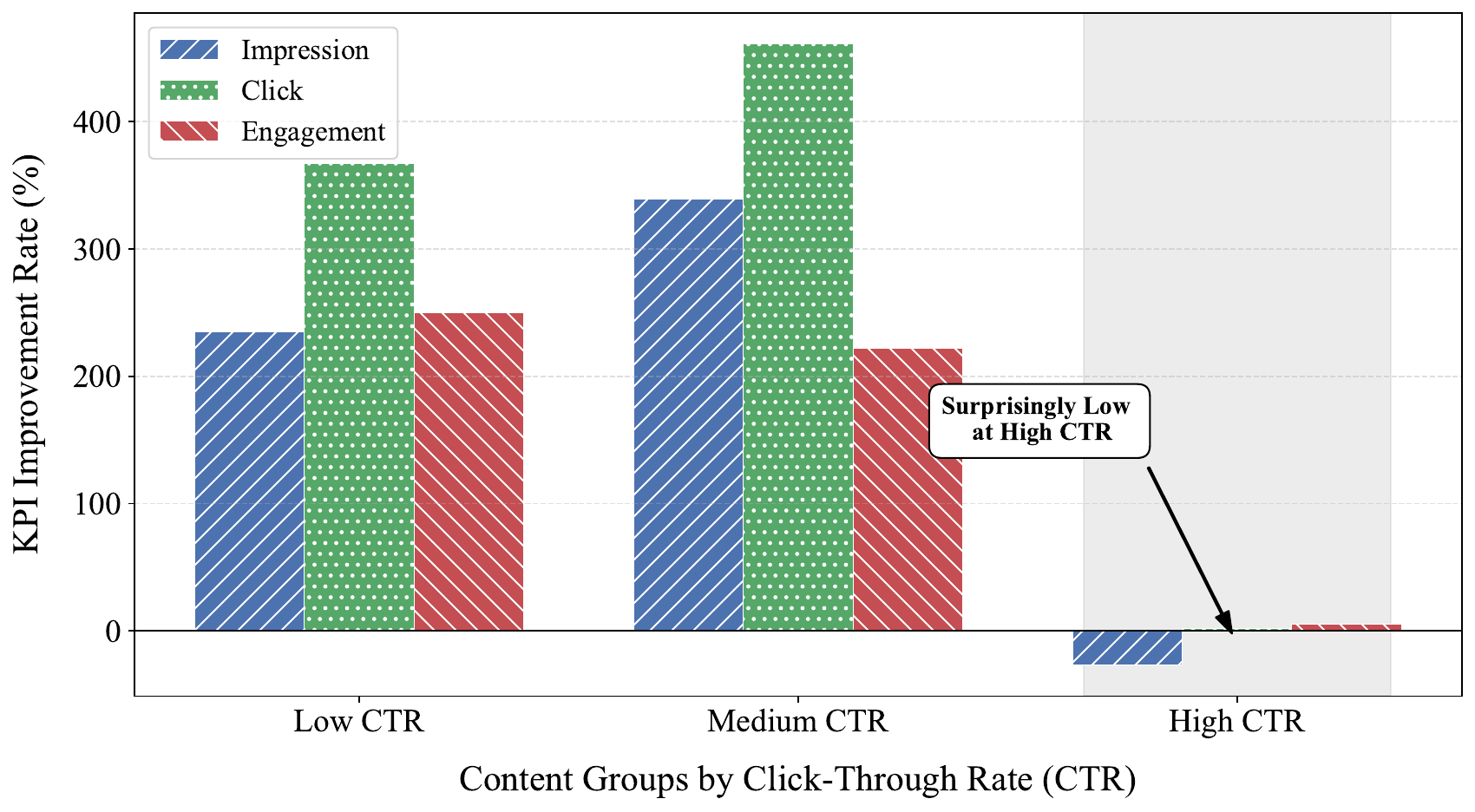}
         \caption{7 days after Shutiao.}
         \label{fig:improv-rate-7}
     \end{subfigure}
     \hfill
     \begin{subfigure}[t]{0.495\textwidth}
         \centering
         \includegraphics[width=\textwidth]{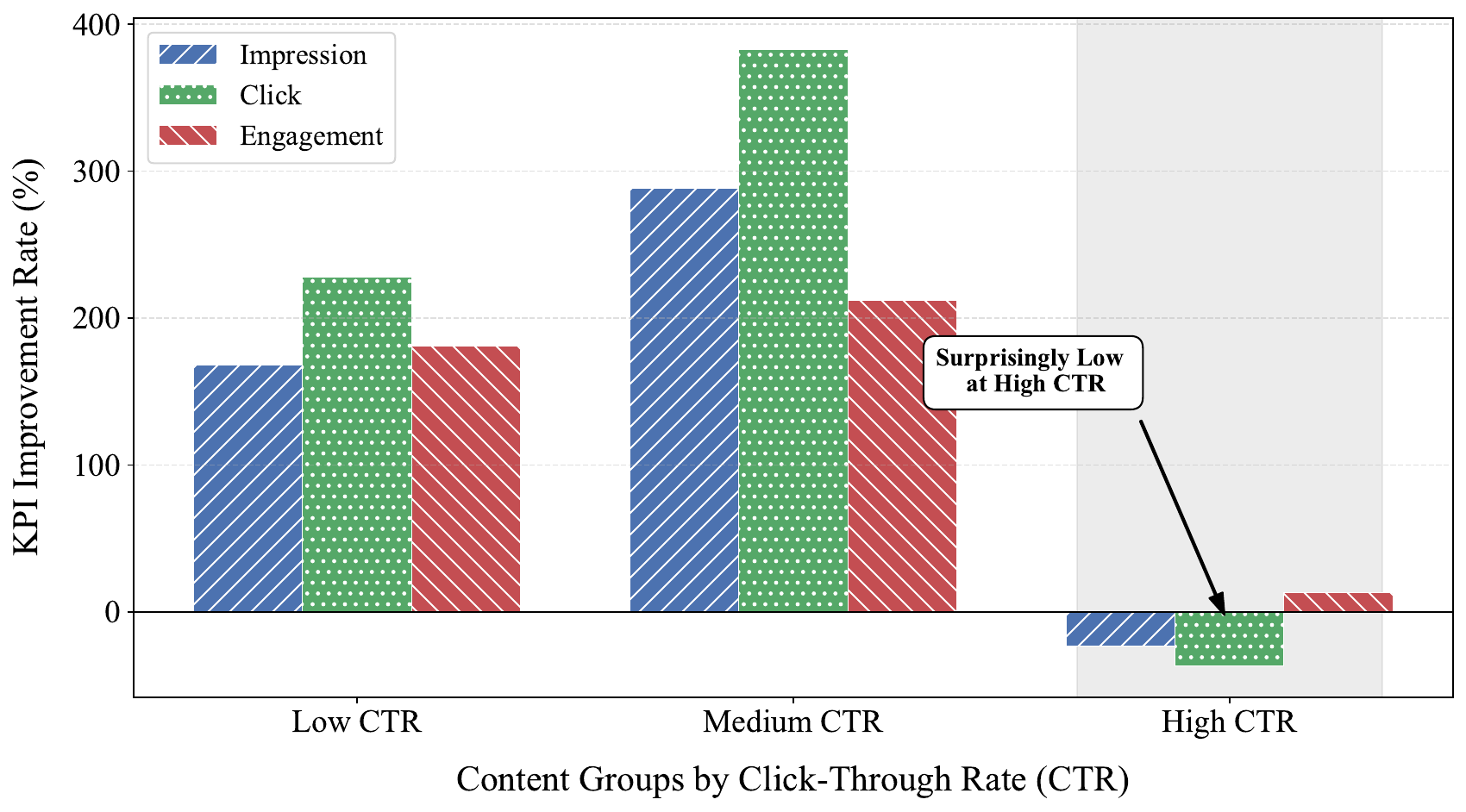}
         \caption{14 days after Shutiao.}
         \label{fig:improv-rate-14}
     \end{subfigure}
        \caption{
        KPI Improvement ratio by Shutiao stratified by content CTR, compared with organic content.
        }
        \label{fig:improv-rate}
\end{figure}
We begin by measuring the performance of the existing Shutiao promotion service to understand its true impact on a content note's KPIs throughout its lifecycle. We investigate a critical question: \emph{Does a short-term content promotion campaign consistently lead to long-term gains in visibility and engagement?}

We conduct a comparative study between a large group of promoted content notes and a control group of organic content notes. For the promoted group, we select content notes that participated in a single Shutiao campaign with a budget 30 units. For the control group, we randomly sample content notes that not utilized the promotion service.\footnote{Detailed configurations of the sampling process are provided in Appendix~\ref{sec:measure-details}.} From a one-week period (March 27, 2025, to April 2, 2025), we sampled 2,560 content notes for each group. To analyze the effect of initial content quality, we stratify all content notes based on their organic CTR measured before the promotion campaign began.\footnote{The proportions of three groups content notes 13.47\%, 56.95\% and 29.57\%, respectively.} We then track the cumulative impressions, clicks, and engagements for each group at 7 and 14 days post-campaign. We compute the KPI improvement rate defined as $\frac{\text{KPI}_{promoted}-\text{KPI}_{organic}}{\text{KPI}_{organic}}\times 100\%$ for each KPI and each group of content notes, respectively.

The results from Figure~\ref{fig:improv-rate} reveal a striking dichotomy. While Shutiao provides a substantial boost for most content, it has a negative long-term impact on high-quality posts (those with a high initial CTR). 
Specifically, as shown in Figure \ref{fig:improv-rate-14}, while content in the low-CTR bucket sees a post-promotion engagement lift of over 200\%, the high-CTR group actually suffers a decline in organic metrics, with click and impression improvement rates dropping to approximately -25\% compared to the control group.
We posit this phenomenon arises because different types of content have different sensitivities to the \emph{quality} of impression obtained from content promotion. During the initial ``cold-start'' phase, all content receives limited organic impressions.
\begin{itemize}
    \item \textbf{For Low-to-Medium Quality Content:} The performance of this content is relatively \textbf{insensitive} to the specifics of the bidding strategy; its primary bottleneck is the lack of sufficient impressions. Paid promotion serves as a brute-force solution to acquire impressions. This new data helps the model revise its initial low pCTR estimate, giving the content a chance to find an audience and thus improving its long-term organic reach.
    
    \item \textbf{For High-Quality Content:} This content is highly \textbf{sensitive} to impression quality. A naive paid promotion campaign, focused on spending a budget, often forces distribution to a broad, less-optimal audience. This poorly-targeted impression yields a lower engagement rate than the organic system would have achieved. The influx of low-engagement signals ``pollutes'' the realized performance, causing the pCTR model to question its initial high-quality assessment and reduce future organic distribution. Here, paid promotion interferes with an already effective organic matching process.
\end{itemize}
This empirical result shows that naive promotion is a double-edged sword. A creator's investment can either rescue their content or inadvertently sabotage it, underscoring the need for an intelligent bidding strategy that looks beyond simply buying impressions.

\subsubsection{Redefining Impression Quality as Informative Data}

The empirical analysis reveals that the goal of paid promotion should not be to maximize raw impression volume, but to acquire \emph{high-quality impression}. We propose a long-term, model-centric definition of quality.

\revision{We define ``long-term success'' as maximizing \textit{Total Lifetime Value} (TLV) of a piece of content. In recommender systems, new content faces a ``cold start'' period where the platform gathers data to estimate its quality. Relying solely on organic impression to exit this high-uncertainty phase is risky due to two factors: (1) \textbf{Time-Discounted Utility:} Organic data accumulation is slow. Since future rewards are discounted, a delayed discovery of quality significantly reduces the creator's total utility. (2) \textbf{Algorithmic Pruning:} Modern recommenders (often contextual bandits) act as ``filters.'' If content fails to accumulate positive signals quickly enough, its predicted Upper Confidence Bound (UCB) drops below the system's selection threshold. Once this happens, the system stops recommending the content entirely—effectively ``killing'' it before it can prove its true worth.

From this viewpoint, high-quality impression consists of impressions that are most useful for training the platform’s pCTR model to quickly reduce variance. By strategically winning impressions in uncertain regions, a creator accelerates the model's convergence, shortening the cold-start duration. This ensures the content survives the initial ``filter'' and unlocks high-volume organic impression earlier in its lifecycle, thereby maximizing discounted total utility.
}
\subsection{Optimal Experimental Design}\label{sec:oed}
The proposed bidding framework is grounded in the principles of Optimal Experimental Design (OED) \cite{pukelsheim2006optimal}, a field of statistics concerned with selecting the most informative data points to minimize the uncertainty of model parameter estimates. Given a model parameterized by $\theta$ and a set of candidate observations $O = \{x_z\}_{z \in \mathcal{S}}$, the informativeness of these observations is typically quantified by the FIM \cite{wittman2025fisher}:
$$
I(O) = \sum_{z \in \mathcal{S}} g_\theta(x_z, y_z; \theta) g_\theta(x_z, y_z; \theta)^\top,
$$
where $g_\theta (\cdot)$ denotes the gradient of the loss function. Classical OED criteria aim to optimize different scalar properties of the FIM to achieve specific variance reduction goals:
\begin{itemize}
    \item \textbf{D-optimality}: Maximizes $\det(I(O))$, which minimizes the volume of the confidence ellipsoid for the parameters $\theta$.
    \item \textbf{A-optimality}: Minimizes $\text{tr}(I(O)^{-1})$, effectively minimizing the average variance of the parameter estimates.
    \item \textbf{I-optimality} (or V-optimality): Minimizes the integrated prediction variance over a region of interest, defined as $\int w(x) g_\theta(x)^\top I(O)^{-1} g_\theta(x) dx$.
\end{itemize}
While these criteria provide rigorous measures for uncertainty reduction, directly optimizing them in a real-time bidding environment is computationally prohibitive due to the need for frequent matrix inversions and non-decomposable updates. This motivates our design of a tractable surrogate objective in Section \ref{sec:method}.
\section{Methods}\label{sec:method}



\begin{figure}[t]
    \centering
    \resizebox{\columnwidth}{!}{\begin{tikzpicture}[
    node distance=1.5cm and 0.8cm,
    block/.style={rectangle, draw, thick, text centered, minimum height=2.5em, text width=2.4cm, fill=white},
    input/.style={rectangle, draw, dashed, text centered, text width=2.0cm, font=\small},
    output/.style={rectangle, draw, thick, text centered, rounded corners, minimum height=3em, text width=3.2cm, fill=green!5},
    arrow/.style={-Stealth, thick},
    labelnode/.style={font=\bfseries\small, text=gray}
]

\node [input] (in1) {Budget \& Expenditure};
\node [input, below=0.5cm of in1] (in2) {Model $\mathcal{M}$ \& Feature $x_t$};

\node [block, right=1.2cm of in1] (pacing) {\textbf{Pacing}\\(Sec \ref{sec:two-stage-bidding})};
\node [right=0.4cm of pacing, font=\large] (lambda) {$\lambda_t$};

\node [block, right=1.2cm of in2] (est) {\textbf{Estimation}\\(Sec \ref{sec:gradient-estimation})};
\node [right=0.4cm of est, font=\large] (gt) {$\hat{g}_t$};

\node [block, right=1cm of gt] (utility) {\textbf{Utility}\\(Sec \ref{sec:objective})};
\node [right=0.4cm of utility, font=\large] (delta) {$\Delta_t$};

\node [output, right=2.61cm of lambda] (bid) {Bid $b_t^*$: \\ $\arg \max [W_a(b)(\Delta_t - \lambda_t b)]$};

\begin{scope}[on background layer]
    \node[fill=blue!5, rounded corners, fit=(pacing) (lambda), label={[labelnode]above:Shadow Price}] {};
    
    \node[fill=yellow!15, rounded corners, fit=(est) (gt) (utility) (delta) (bid.east |- delta), label={[labelnode]below:Surrogate Calculation}] {};
\end{scope}

\draw [arrow] (in1) -- (pacing);
\draw [arrow] (in2) -- (est);
\draw [arrow] (pacing) -- (lambda);
\draw [arrow] (est) -- (gt);
\draw [arrow] (gt) -- (utility);
\draw [arrow] (utility) -- (delta);

\draw [arrow] (lambda.east) -| ([xshift=-0.5cm]bid.west) -- (bid.west);
\draw [arrow] (delta.north)  -| (bid.south);

\end{tikzpicture}} 
    
    
    \caption{System Diagram of Information-Aware Auto-Bidding.}
    \label{fig:diagram}
\end{figure}
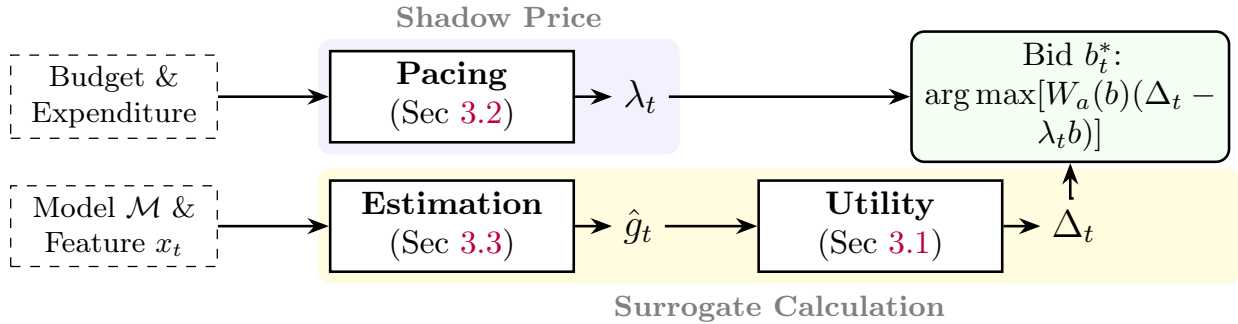

In this section, we present our proposed bidding methodology designed to optimize content promotion for creators. Our approach is structured into three main parts, as shown in Figure \ref{fig:diagram}. First, we formally define our bidding objective, which innovatively combines the immediate value of an impression with a surrogate for long-term model uncertainty reduction, and establish its crucial submodularity property (Sec.~\ref{sec:objective}). Next, we detail a two-stage bidding framework that leverages this submodular objective to efficiently solve the budget-constrained optimization problem via Lagrange duality (Sec.~\ref{sec:two-stage-bidding}). Finally, we address a key practical challenge by introducing a confidence-gated heuristic for estimating the necessary loss gradients in real-time, even when the true outcome label is not yet available (Sec.~\ref{sec:gradient-estimation}).

\subsection{Surrogate Uncertainty}
\label{sec:objective}
The core of our bidding strategy is to select a set of impressions, $\mathcal{S}$, that optimally balances two competing goals: the short-term goal of maximizing immediate campaign returns and the long-term goal of improving the model by reducing its uncertainty. To this end, we formulate a composite objective function that captures this trade-off.

The first component of our objective, $V(\mathcal{S})$, quantifies the immediate value accrued from the campaign. We define this as the total expected value from the winning impressions, which can be modeled as the sum of their pCTRs:
\begin{equation*}
    V(\mathcal{S}) = \sum_{\mathbf{z} \in \mathcal{S}} \hat{\sigma}(\mathbf{z}),
    \label{eq:value_term}
\end{equation*}
where $\hat{\sigma}(\mathbf{z})$ is the pCTR for the impression $\mathbf{z}$. This term incentivizes the bidder to acquire impressions that are likely to yield immediate positive user engagement.

The second component, $U(\mathcal{S})$, serves as a computationally tractable surrogate for model uncertainty reduction. It is designed to encourage the selection of a diverse and representative set of training samples. Let $\mathcal{D}_{\text{val}}$ be a fixed, representative set of validation samples. We define this uncertainty surrogate as:
\begin{equation}
    U(\mathcal{S}) := \sum_{\mathbf{x} \in \mathcal{D}_{\text{val}}} \max_{\mathbf{z} \in \mathcal{S}} \exp \left(-\lambda \|\mathbf{g}_{\theta_0}(\mathbf{x}) - \mathbf{g}_{\theta_0}(\mathbf{z})\|^2\right),
    \label{eq:uncertainty_surrogate}
\end{equation}
where $\mathbf{g}_{\theta_0}(\cdot)$ is the gradient of the model's loss function. This objective, which we term ``gradient coverage,'' rewards the selection of a set $\mathcal{S}$ whose samples' gradients are collectively close to the gradients of the validation data, intuitively covering the space of necessary learning signals.
For any data point $\bvec{x}$, $\gradg(\bvec{x})$ denotes the gradient of the model's loss function with respect to the model parameters $\theta_0$\footnote{A discussion is given in Appendix~\ref{sec:parameter-alignment} on the practicality of computing the gradient in industry.}. It represents the ``learning signal'' from that sample, which indicates the direction in parameter space that could reduce the loss for that specific example.
The term $\exp(-\lambda \|\cdot\|^2)$ is a Gaussian kernel that measures the similarity between two gradient vectors. It yields a value close to 1 if the gradients are nearly identical and decays to 0 as they become more distant. The hyperparameter $\lambda > 0$ controls the sensitivity of this similarity measure.
For each sample $\bvec{x}$ in the validation set, the inner expression, $\max_{\bvec{z} \in \calS} \exp(\dots)$, finds the training sample $\bvec{z}$ in our selected set $\calS$ whose gradient is most similar to the gradient of $\bvec{x}$. This can be interpreted as the ``best coverage'' that $\calS$ provides for the learning signal required by $\bvec{x}$.

We combine these two components into a single, unified objective function, $F(\mathcal{S})$, using a hyperparameter $\beta \in [0, 1]$ to control the balance:
\begin{equation}
    F(\mathcal{S}) := (1-\beta) \cdot U(\mathcal{S}) + \beta \cdot V(\mathcal{S}).
    \label{eq:new_objective}
\end{equation}
Here, $\beta=0$ prioritizes purely long-term uncertainty reduction, while $\beta=1$ focuses exclusively on short-term pCTR maximization.

In essence, maximizing $\calF(\calS)$ makes the bidder to acquire a portfolio of impressions whose learning signals collectively blanket the space of signals needed to improve performance on the overall data distribution. This property of ``gradient coverage'' serves as our intuitive and computationally friendly proxy for the more complex goal of direct uncertainty minimization. In the subsequent analysis section, we will formally establish the relationship between this surrogate and standard information-theoretic measures of model uncertainty.

\revision{\paragraph{Advantages over classical metrics.} Objectives such as conventional A-/D-optimal designs introduced in Section \ref{sec:oed} require repeated matrix updates and inversions and are thus non-decomposable and impractical for millisecond-latency bidding \cite{kiefer1959,pukelsheim2006optimal,ChalonerVerdinelli1995}. 
In contrast, our gradient-coverage objective is a facility-location–style \cite{owen1998strategic,soland1974optimal} max-kernel in gradient space: it yields per-impression marginal gains by updating running maxima over a fixed validation set, avoiding any matrix inversion and enabling real-time computation. The construction induces monotone submodularity, providing diminishing returns that support efficient online selection \cite{guestrin2005near,sener2018active}. 
Compared to expected-gradient/model-change criteria from active learning \cite{settles2008analysis,settles2009active} and gradient-embedding methods such as BADGE \cite{Ash2020Deep} which are based on the techniques in Section \ref{sec:oed}, our formulation explicitly targets long-term variance reduction via a provable link to Fisher information (Theorem~\ref{thm:surrogate-relation}) while remaining decomposable at impression granularity, which is essential for auction-time bidding. 
Computationally, computing our surrogate requires $O(|\mathcal{D}_{\text{val}}|)$ kernel evaluations and max updates per impression, and can be further reduced via batching or cached partial maxima—without forming or inverting any Fisher matrices.}

\subsection{Two-Stage Bidding}\label{sec:two-stage-bidding}
In this subsection, we optimize our objective as a budget-constrained bidding (BCB) problem.
The goal is to select a set of impressions $S$ to win via bidding that maximizes an uncertainty-reduction utility function, $\calF(S)$, subject to a total budget $B_i$.
\begin{equation*}
    \max_{S} \calF(S) \quad \text{s.t.} \quad \sum_{t \in S} b_t \leq B_i.
\end{equation*}
As $\calF(S)$ can be shown to be a monotone submodular function (Section \ref{sec:analysis}), this problem can be effectively solved using Lagrange duality \cite{karlsson2021adaptive}. The Lagrangian is:
\begin{equation*}
    \calL(S, \lambda) = \calF(S) - \lambda \left( \sum_{t \in S} b_t - B_i \right).
\end{equation*}
This allows us to decompose the global problem into a per-impression bidding decision. For each impression $x_t$, the marginal utility gain is $\Delta_t = \calF(S_{t-1} \cup \{x_t\}) - \calF(S_{t-1})$. The expected budget-aware surplus from bidding $b_t$ is $W_a(b_t) \cdot (\Delta_t - \lambda b_t)$, where $W_a(b_t)$ is the win probability for a bid $b_t$. The optimal bid $b_t^*$ is then:
\begin{equation*}
    b_t^* = \arg\max_{b_t \geq 0} \left[ W_a(b_t) \cdot (\Delta_t - \lambda b_t) \right].
\end{equation*}
This forms a two-stage bidding process:
\begin{enumerate}
    \item \textbf{Campaign-Level Pacing:} The dual variable $\lambda$, which represents the shadow price of the budget, is controlled at the campaign level. It can be updated dynamically using a multiplicative weights update rule:
    \begin{equation}\label{eq:dual-update}
        \lambda_k = \lambda_{k-1} \cdot \exp\left(\eta \cdot \frac{\text{Cost}_{k-1} - \text{Paced\_Budget}_{k-1}}{B_i}\right),
    \end{equation}
    where $k$ indexes pacing periods and $\eta$ is a learning rate. The term $\text{Paced\_Budget}_{k-1}$ denotes the target cumulative expenditure up to period $k-1$. To ensure the budget spans the entire campaign duration, we typically employ a linear pacing schedule defined as $\text{Paced\_Budget}_{k} = B_i \cdot \frac{k}{K}$, where $K$ is the total number of pacing intervals. Equation (\ref{eq:dual-update}) allows the algorithm to ``learn'' the appropriate value she would pay for the impressions relative to its remaining budget.
    \item \textbf{Impression-Level Bidding:} At each auction, the optimal bid $b_t^*$ is computed based on the current $\lambda$ and the estimated marginal utility $\Delta_t$.
\end{enumerate}

\revision{\begin{remark}\label{rem:generalize-second}
While Eq. (\ref{eq:dual-update}) derives the bid for First-Price Auctions \cite{karlsson2021adaptive}, our framework adapts easily to Second-Price Auctions (SPA). In an SPA setting with budget constraints, the truthful bidding strategy is modified by the shadow price $\lambda$. The optimal bid simplifies to a scaled truthful bid: $b^*_t=\frac{\Delta_t}{\lambda}$ (where $\lambda\geq 1$). See details in Appendix \ref{sec:spa}.
\end{remark}}

The central challenge is to define and compute $\Delta_t$ in real-time, which we address next.

\begin{algorithm}[t]
    \caption{Confidence-Gated Marginal Utility Estimation}
    \label{alg:utility_estimation}
    \DontPrintSemicolon
    \Input{
        Current impression features $\bvec{x}_t$;
        CTR prediction model $M(\cdot; \theta_t)$ with parameters $\theta_t$;
        Set of gradients from previously won impressions $\calG_S = \{\bvec{g}_z\}_{z \in S_{t-1}}$;
        Confidence (entropy) threshold $\zeta_t$;
        High-uncertainty utility value $\bar{u}$;
        Loss function $\mathcal{L}(\cdot, \cdot)$;
        Utility function hyperparameters (e.g., $\lambda$ from $\calF$);
    }
    \Output{Estimated marginal utility $\Delta_t$;}
    
    \BlankLine
    
    \Comment{1. Assess model confidence}
    $p_t \leftarrow M(\bvec{x}_t; \theta_t)$; \Comment*[r]{Get predicted CTR}
    $H(p_t) \leftarrow -p_t \log_2(p_t) - (1-p_t)\log_2(1-p_t)$; \Comment*[r]{Compute prediction entropy}
    
    \BlankLine
    
    \uIf{$H(p_t) > \zeta_t$}{
        \Comment{2a. Low-confidence case: assign high fixed utility}
        \Return $\bar{u}$;\;
    }
    \Else{
        \Comment{2b. High-confidence case: use gradient heuristic}
        
        $\bvec{g}_0 \leftarrow \grad \mathcal{L}(M(\bvec{x}_t; \theta_t), 0)$; \Comment*[r]{Gradient assuming label is 0}
        $\bvec{g}_1 \leftarrow \grad \mathcal{L}(M(\bvec{x}_t; \theta_t), 1)$; \Comment*[r]{Gradient assuming label is 1}
        
        \uIf{$\|\bvec{g}_0\|_2 < \|\bvec{g}_1\|_2$}{
            $\hat{\bvec{g}}_t \leftarrow \bvec{g}_0$;\;
        }
        \Else{
            $\hat{\bvec{g}}_t \leftarrow \bvec{g}_1$;\;
        }
        
        \BlankLine
        \Comment{Compute marginal utility gain using the proxy gradient}
        $\calF_{\text{current}} \leftarrow \text{ComputeUtility}(\calG_S)$;\;
        $\calF_{\text{new}} \leftarrow \text{ComputeUtility}(\calG_S \cup \{\hat{\bvec{g}}_t\})$;\;
        $\Delta_t \leftarrow \calF_{\text{new}} - \calF_{\text{current}}$;\;
        
        \Return $\Delta_t$;\;
    }
    
    \BlankLine
    \SetKwFunction{FMain}{ComputeUtility}
    \SetKwProg{Fn}{Function}{}{end}
    \Fn{\FMain{$\calG_{set}$}}{
        \Comment{Helper to compute total utility for a set of gradients}
        \Comment{Implements $\calF(S) = \sum_{x \in D_{test}} \max_{z \in S} \exp(-\lambda \|g_x - g_z\|^2)$}
        \KwRet value based on the definition of $\calF$;\;
    }
\end{algorithm}

\subsection{Gradient Estimation}\label{sec:gradient-estimation}

A critical challenge in applying our framework is the real-time computation of the marginal utility, $\Delta_t = \calF(S_{t-1} \cup \{x_t\}) - \calF(S_{t-1})$, at bid time. The utility function $\calF$ in Equation~(\ref{eq:new_objective}) depends on the loss gradient of a potential training sample $x_t$. However, at the moment of bidding, the true label (i.e., whether the user will click) is unknown. This ``missing label'' problem prevents the direct computation of the gradient via standard backpropagation.

To overcome this, we propose a practical hybrid strategy, detailed in \Algref{alg:utility_estimation}, which uses the model's own confidence as a signal to switch between two estimation modes. The model's confidence in its prediction for an impression $x_t$ is measured by the entropy of its pCTR, $\hat{\sigma}_t$. A high entropy signifies high uncertainty (low confidence), while low entropy signifies high confidence.

\paragraph{Confidence-Gated Heuristic}
Our approach is governed by a dynamic confidence threshold, $\zeta_t$. For each impression, we compare its prediction entropy to this threshold.

\noindent\textbf{High-Confidence (Low-Entropy) Case:} If the entropy $H(p_t) \leq \zeta_t$, the model is relatively certain about its prediction. In this regime, we can devise a heuristic to approximate the true gradient. We compute two \textit{hypothetical} gradients: $\bvec{g}_0$, assuming the label is $y=0$, and $\bvec{g}_1$, assuming the label is $y=1$. The core assumption is that for a well-trained and confident model, the loss gradient corresponding to the \textit{correct} (and therefore more likely) label will be smaller in magnitude. The model state is already near a local minimum for that class, requiring a smaller update. We thus select the gradient with the smaller L2-norm as our proxy for the true gradient:
\begin{equation*}
    \hat{\bvec{g}}_t = \arg\min_{\bvec{g} \in \{\bvec{g}_0, \bvec{g}_1\}} \|\bvec{g}\|_2.
\end{equation*}

\revision{This heuristic assumes the model is well-calibrated. In ``confident-but-wrong'' scenarios, the heuristic might select the uninformative gradient. However, our architecture mitigates this risk via the entropy threshold $\zeta_t$. High-entropy samples—where the heuristic is least reliable—are effectively filtered out and assigned a fixed high exploration utility $\bar{u}$, ensuring we acquire the label rather than relying on a noisy gradient estimate.}
    
\noindent\textbf{Low-Confidence (High-Entropy) Case:} If the entropy $H(p_t) > \zeta_t$, the model is highly uncertain about the outcome. From an active learning perspective, such samples are intrinsically valuable because they reside in regions of the feature space where the model is uncertain. Correctly labeling and training on these samples offers the highest potential for model improvement and future uncertainty reduction. Therefore, instead of estimating a precise gradient-based utility, we assign a high, fixed utility value, $\bar{u}$, to winning this impression. This value represents the strategic importance of exploring the model's uncertain data.



This dual-mode approach is robust: it prioritizes exploration when the model is uncertain and relies on a reasonable heuristic for exploitation and refinement when the model is confident. The threshold $\zeta_t$ and the utility constant $\bar{u}$ are system hyperparameters.

Label-free utilities often rely on an expected-gradient or expected model-change under the model’s predicted label distribution \cite{settles2008analysis,settles2009active}, or on expected output change \cite{freytag2014selecting}. These estimators can fail in the confident-but-wrong regime: the posterior mass collapses on the incorrect label and the expected gradient points in an uninformative direction. Our confidence-gated heuristic is tailored to this failure mode: we explore aggressively when entropy is high, and when entropy is low we approximate the true gradient by the smaller-norm hypothetical gradient, consistent with local optimality around the more likely label. Moreover, our zeroth-order (two-point) variant enables use with black-box CTR models, leveraging established ZO/SPSA estimators \cite{duchi2015optimal,nesterov2017random,spall1992multivariate,liu2020primer,malladi2023fine}. This combination makes information-aware bidding feasible when analytical gradients and labels are unavailable at bid time.
\section{Theoretical Analysis}\label{sec:analysis}

In this section, we conduct theoretical analysis to validate the soundness of our proposed bidding framework. We structure our analysis around four key pillars. First, we establish a formal connection between our tractable surrogate objective and a standard information-theoretic measure of model uncertainty, thereby justifying our problem formulation (Sec.~\ref{sec:surrogate-relationship}). Second, we prove that our composite objective function is submodular, a fundamental property that makes the optimization problem computationally tractable (Sec.~\ref{sec:submodularity}). Building on this, we analyze the algorithm's performance in its natural online environment, providing a sublinear regret bound that demonstrates its competitiveness against an optimal offline solution (Sec.~\ref{sec:regret-analysis}). Finally, we prove a budget feasibility guarantee, ensuring that the algorithm is fiscally responsible and predictable (Sec.~\ref{sec:budget-feas}). Together, these results formally establish our method as effective and reliable.

\subsection{Surrogate Relation}\label{sec:surrogate-relationship}
In this subsection, we prove the correctness of our proposed objective function in Section \ref{sec:objective}.
Part of the objective function $U(\mathcal{S})$ is designed as as a computationally efficient surrogate for reducing model uncertainty. A crucial question, however, is whether maximizing this ``gradient coverage'' objective truly corresponds to the fundamental goal of improving the model's predictive certainty. 
To validate our approach, we establish a formal connection between our tractable surrogate $F(\mathcal{S})$ and this information-theoretic measure of uncertainty. 
The following theorem proves that maximizing our surrogate $U(\mathcal{S})$ is indeed a principled approach for minimizing the true model uncertainty.

\begin{theorem}[Regularized Fisher-Coverage Relationship]
\label{thm:surrogate-relation}
Let $\mathcal{D}_{\mathrm{val}}$ be a fixed validation set of size $k$, and for each $x\in\mathcal{D}_{\mathrm{val}}$ let $g(x)\in\mathbb{R}^d$ denote the loss gradient at a common anchor parameter $\theta_{\mathrm{anchor}}$. For a selected training set $S\subseteq\mathcal{D}_{\mathrm{train}}$ with gradients $\{g(z)\}_{z\in S}$, define the regularized empirical Fisher
\[
\mathcal{I}_\gamma(S) \;:=\; \sum_{z\in S} g(z)\,g(z)^\top \;+\; \gamma I_d,\quad \gamma>0,
\]
and the regularized total uncertainty (analogy to \cite{lu2024daved})
\[
G_\gamma(S) \;:=\; \sum_{x\in\mathcal{D}_{\mathrm{val}}} g(x)^\top \mathcal{I}_\gamma(S)^{-1} g(x).
\]
Let the gradient-coverage surrogate be
\[
U_\lambda(S) \;:=\; \sum_{x\in\mathcal{D}_{\mathrm{val}}} \max_{z\in S} \exp\!\big(-\lambda\,\|g(x)-g(z)\|^2\big),\qquad \lambda>0.
\]
Assume (i) bounded gradients: $\|g(v)\|\le L$ for all $v\in \mathcal{D}_{\mathrm{val}}\cup S$, and (ii) non-degenerate norms on $S$: $\|g(z)\|\ge m>0$ for all $z\in S$. Then for any choice of threshold $\tau\in(0,2m^2]$,
\[
G_\gamma(S) \;\le\; \frac{k\,L^2}{\gamma} 
\;-\; \frac{\big(2m^2-\tau\big)^2}{4\,\gamma^2\,\big(1+\tfrac{L^2}{\gamma}\big)}\cdot 
\frac{U_\lambda(S)-k\,e^{-\lambda\tau}}{1-e^{-\lambda\tau}}.
\]
In particular, increasing $U_\lambda(S)$ (for fixed $\lambda,\tau,\gamma$) tightens the upper bound on $G_\gamma(S)$, so $U_\lambda(S)$ is a monotone surrogate for reducing $G_\gamma(S)$.
\end{theorem}


The proof is given in Appendix~\ref{sec:proof-surrogate-relation}.
Theorem~\ref{thm:surrogate-relation} provides the core theoretical justification for our proposed bidding objective. 
As we select a set $\mathcal{S}$ that increases the value of $U(\mathcal{S})$, the upper bound on $G(\mathcal{S})$ is driven down.
The uncertainty measure $G(\mathcal{S})$ can be written as $\text{tr}(I(S)^{-1} J_\text{val})$ (or with a small ridge regularization if needed), where $I(S)=\sum_{z\in \mathcal{S}} g(z) g(z)^\top$ is the empirical Fisher and $J_\text{val}=\sum_{x\in \mathcal{D}_{\text{val}}} g(x) g(x)^\top$ is the validation gradient second-moment. This is exactly the I-optimality criterion (integrated prediction variance) evaluated on the validation set. In the special case where the validation gradient covariance is isotropic (or after whitening), i.e., $J_{\text{val}}$ is proportional to the identity, $G(\mathcal{S})$ reduces (up to a constant factor) to the A-optimal objective $\text{tr}(I(S)^{-1})$. Thus, minimizing $G(\mathcal{S})$ recovers I-optimal design in general and A-optimal design in the isotropic (or whitened) case, which further justifies our use of $G(\mathcal{S})$ as principled uncertainty.

Intuitively, Theorem~\ref{thm:surrogate-relation} holds because both functions, despite their different forms, capture the notion of ``coverage.'' A high value of $U(\mathcal{S})$ means that the gradients of the selected samples in $\mathcal{S}$ are close to the gradients of the validation data $\mathcal{D}_{\text{val}}$. Similarly, a low value of $G(\mathcal{S})$ means that the vector space spanned by the gradients in $\mathcal{S}$ effectively represents the validation data gradients, minimizing projection error and thus variance. Our proof formalizes this shared intuition.


\subsection{Submodularity of the Uncertainty-Reduction Utility Function}\label{sec:submodularity}

In this subsection, we analyze the structural properties of our objective function. 
Proving that our objective function is submodular is the key that unlocks our bidding algorithm. 

\begin{theorem}[Submorularity of the Uncertainty-Reduction Utility Function]\label{thm:submodular}
$F:2^\mathcal{S}\to \R^+$ is a submodular function.
\end{theorem}
The proof is given in Appendix~\ref{sec:proof-submodular}.

Submodularity formalizes the intuition that the marginal gain of adding a new impression to our selected set $\mathcal{S}$ decreases as the set grows. 
\begin{itemize}
    \item For the uncertainty component $U(\mathcal{S})$, adding an impression with a novel gradient to a small set $\mathcal{S}$ yields a large increase in ``gradient coverage.'' However, adding that same impression to a large, diverse set offers a smaller marginal benefit, as its learning signal is likely already well-represented by other samples in the set.
    \item For the value component $V(\mathcal{S})$, the marginal gain is constant, which is a special case of submodularity.
\end{itemize}

The budget-constrained maximization of a monotone submodular function can be efficiently solved with approximation guarantees. 
This property transforms an otherwise intractable optimization problem into one that can be realistically solved within the stringent time constraints of a real-time bidding auction, thereby improving the efficiency and reliability of the model optimization process.

\subsection{Regret Analysis}\label{sec:regret-analysis}
Our two-stage bidding algorithm operates in an online setting: impression opportunities arrive sequentially, and an irrevocable decision to bid (and how much) must be made for each one without knowledge of future opportunities.
The difference between the expected utility of this offline optimum and our online algorithm's utility is termed \textit{regret}. 

\begin{theorem}[Regret Bound for First-Price CPM Dual Pacing]
\label{thm:regret}
Consider a sequence of $T$ auctions. At round $t$, the bidder observes a win-probability curve $W_t(b)\in[0,1]$ for bids $b\in[0,b_{\max}]$ and a marginal utility gain $\Delta_t$ (from acquiring the impression), with $|\Delta_t|\le \Delta_{\max}$. 
Under a first-price CPM mechanism (winner pays the eCPM), the expected spend at round $t$ is $h_t := W_t(b_t)\,b_t$, where $b_t$ is the bid placed at round $t$.

Define the per-round dual objective
\(
f_t(\lambda) := \max_{b\in[0,b_{\max}]}\, W_t(b)\,\big(\Delta_t - \lambda\,b\big),
\)
and let the algorithm choose $b_t\in\arg\max_{b} W_t(b)\,(\Delta_t-\lambda_{t-1}b)$, with the dual (shadow-price) update given by multiplicative weights (mirror descent with negative entropy):
\[
\lambda_t \;=\; \lambda_{t-1}\,\exp\!\Big(\eta\,\frac{h_t}{C}\Big),\qquad C\ge b_{\max},\;\;\eta>0,\;\;\lambda_0\in(0,\lambda_{\max}].
\]
Let the algorithm's total expected utility be 
\(
\mathrm{ALG} \;:=\; \sum_{t=1}^T \mathbb{E}\!\left[W_t(b_t)\,\Delta_t\right].
\)
Let $B>0$ be the budget and define the offline optimal value
\(
\mathrm{OPT} 
\;:=\; \max_{S^*:\,\mathrm{Spend}(S^*)\le B}\; \sum_{t=1}^T \mathbb{E}\!\left[\Delta_t^{(S^*)}\right],
\)
where $\mathrm{Spend}(S^*)$ is the (CPM) spend of the offline solution. 
Let $\lambda^*\in[0,\lambda_{\max}]$ be a dual optimal solution to 
$\min_{\lambda\ge 0} \sum_{t=1}^T f_t(\lambda) + \lambda B$.
Then, choosing
\(
\eta \;=\; \sqrt{\frac{\log(\lambda_{\max}/\lambda_0)}{T}}\,\frac{1}{C},
\)
we have the regret bound
\[
\mathbb{E}[\mathrm{ALG}]
\;\ge\;
\mathrm{OPT}
\;-\;
2\,C\,\sqrt{T\,\log(\lambda_{\max}/\lambda_0)}
\;-\;
\lambda^*\,\Big(B - \mathbb{E}\big[\textstyle\sum_{t=1}^T h_t\big]\Big).
\]
In particular, if the pacing ensures $\mathbb{E}[\sum_{t=1}^T h_t]\approx B$, the last term is negligible and the regret is $O\!\big(C\,\sqrt{T\,\log(\lambda_{\max}/\lambda_0)}\big)$.
\end{theorem}

The proof is given in Appendix~\ref{sec:proof-regret}.

Theorem~\ref{thm:regret} provides the theoretical foundation for our dual-variable-based pacing. 
The sublinear regret bound proves that this online learning process is effective, ensuring that the algorithm intelligently allocates the budget over the entire campaign horizon. This avoids common pitfalls such as prematurely exhausting the budget on mediocre impressions or being overly conservative and failing to spend the budget on high-value opportunities at the end.

The result confirms that our method is not merely a heuristic but a principled online optimization algorithm with provable near-optimal performance. It guarantees that a creator's budget will be utilized efficiently and effectively over time, adapting to the dynamic conditions of the auction marketplace.

\subsection{Budget Feasibility}\label{sec:budget-feas}

A crucial property of our algorithm is \textit{budget feasibility}, a formal guarantee that its total expenditure will remain close to the allocated budget. 
This property is fundamental to establishing trust and making the promotion tool reliable and predictable for its users. The following theorem proves that our algorithm satisfies this critical requirement in expectation.

\begin{theorem}[Budget Feasibility Guarantee]\label{thm:budget-feas}
The expected total expenditure of the algorithm over $T$ auctions is bounded. Using the same learning rate $\eta$ as in the regret analysis, the expected cost satisfies:
\begin{equation*}
    \mathbb{E}\left[\sum_{t=1}^{T} b_t \cdot \mathbf{1}_{\text{win}_t}\right] \le B + \frac{\log(\lambda_{\max}/\lambda_0)}{\eta},
\end{equation*}
where $\mathbf{1}_{\text{win}_t}$ is an indicator that the bid $b_t$ wins the auction at time $t$.
\end{theorem}
The proof is given in Appendix~\ref{sec:proof-budget-feas}.
Theorem~\ref{thm:budget-feas} demonstrates that the expected expenditure is bounded by the budget $B$ plus an additional term that is controlled by the algorithm's learning parameters. This second term, $\frac{\log(\lambda_{\max}/\lambda_0)}{\eta}$, can be understood as the ``cost of online learning.'' Because the algorithm cannot see the future, it must dynamically adjust its spending, and this term bounds the potential overspend that arises from this adaptive process.

This result is paramount for user trust. It assures a creator that the system will not behave erratically and deplete their funds uncontrollably. By providing a formal upper bound on expected spending, our method transforms the promotion tool from a black box into a reliable and auditable system. This financial control is essential for creators to confidently invest in content promotion, knowing their budget will be managed responsibly throughout the campaign's lifecycle.

\section{Evaluation}


We conduct experiments to evaluate our methods in Section~\ref{sec:method}. Firstly, we verify the three parts of our methods separately. Then, we combine the methods together and carry out offline experiments. In the subsequent content, we introduce the setup and results of our experiments.



\subsection{Experiment 1: Surrogate Relationship}\label{sec:exp-surrogate-relationship}

\subsubsection{Experimental Setup}

\textbf{Dataset and Partitioning.}
To ensure a controlled and reproducible environment, we generate a synthetic binary classification dataset using the \texttt{scikit-learn}. The dataset comprises 1,200 samples, each with 20 features. This dataset is then partitioned into three disjoint sets: 
An initial, small labeled training set, $\mathcal{D}_{\text{initial}}$, containing 200 samples used to train the base model; A large, unlabeled candidate pool, $\mathcal{D}_{\text{candidate}}$, containing 500 samples from which the active learning strategies will select data; A held-out test set, $\mathcal{D}_{\text{test}}$, containing 500 samples, used exclusively for the final performance evaluation of the retrained models.

\textbf{Model and Protocol.}
The underlying pCTR model is a standard Logistic Regression classifier, implemented without an intercept term to align with our gradient formulation. The experimental protocol follows a single-batch active learning cycle: Firstly, we train a base model, $\theta_0$ on $\mathcal{D}_{\text{initial}}$. Then, each selection strategy is used to select a batch of $B=50$ samples from $\mathcal{D}_{\text{candidate}}$. For each strategy, a new, augmented training set is formed by combining $\mathcal{D}_{\text{initial}}$ with the 50 selected samples. A new model is trained from scratch on this augmented dataset. The performance of the newly trained model is evaluated on $\mathcal{D}_{\text{test}}$.

\textbf{Compared Methods.}
We evaluate the performance of three distinct selection strategies:
\begin{itemize}
    \item \textbf{\texttt{Greedy-Surrogate}}: Our proposed method, which iteratively selects samples that yield the maximum marginal gain in the surrogate objective $U(S)$ (see Equation~\ref{eq:uncertainty_surrogate}). The kernel bandwidth hyperparameter is set to $\lambda=0.1$.
    \item \textbf{\texttt{Greedy-FIM (Oracle)}}: A strong but computationally expensive baseline that greedily selects samples to maximize the reduction in the true model uncertainty $G(S)$. This serves as a practical upper bound on performance.
    \item \textbf{\texttt{Random}}: A naive baseline that selects 50 samples uniformly at random from $\mathcal{D}_{\text{candidate}}$.
\end{itemize}

\textbf{Evaluation Metrics.}
The effectiveness of each selection strategy is quantified by the performance of the corresponding retrained model on the unseen test set $\mathcal{D}_{\text{test}}$. We use two standard metrics:
\begin{itemize}
    \item \textbf{Test Log Loss}: Measures the model's goodness-of-fit. Lower values are better.
    \item \textbf{Area Under the ROC Curve (AUC)}: Measures the model's ability to discriminate between the positive and negative classes. Higher values are better.
\end{itemize}
Detailed configurations are given in Appendix~\ref{sec:exp-set-app-surrogate-relationship}.

\subsubsection{Result Analysis}
\begin{figure}[t]
     \centering
     \begin{subfigure}[t]{0.3\textwidth}
         \centering
         \includegraphics[width=\textwidth]{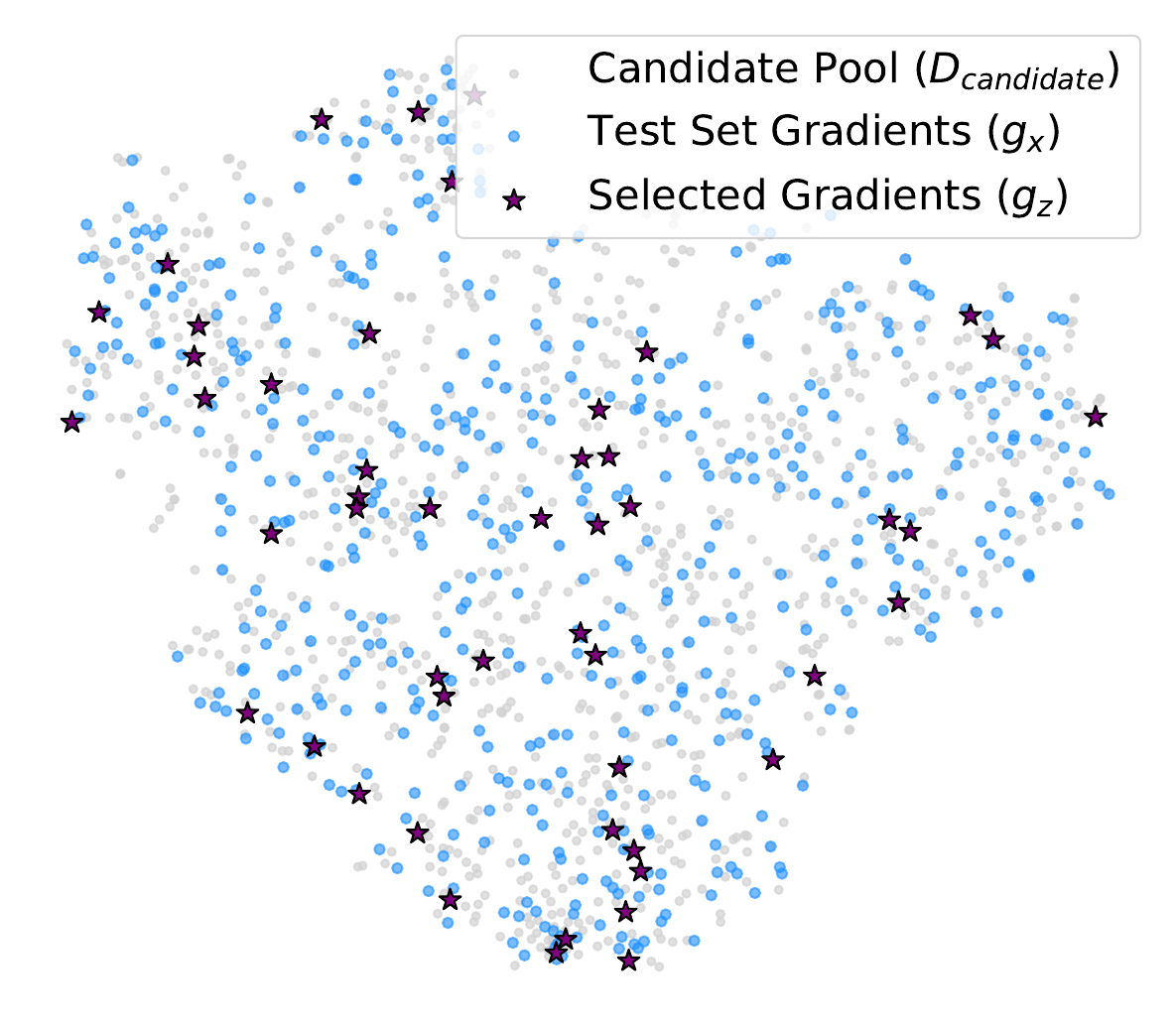}
         \caption{Random.}
         \label{fig:surro-tsne-random}
     \end{subfigure}
     \hfill
     \begin{subfigure}[t]{0.3\textwidth}
         \centering
         \includegraphics[width=\textwidth]{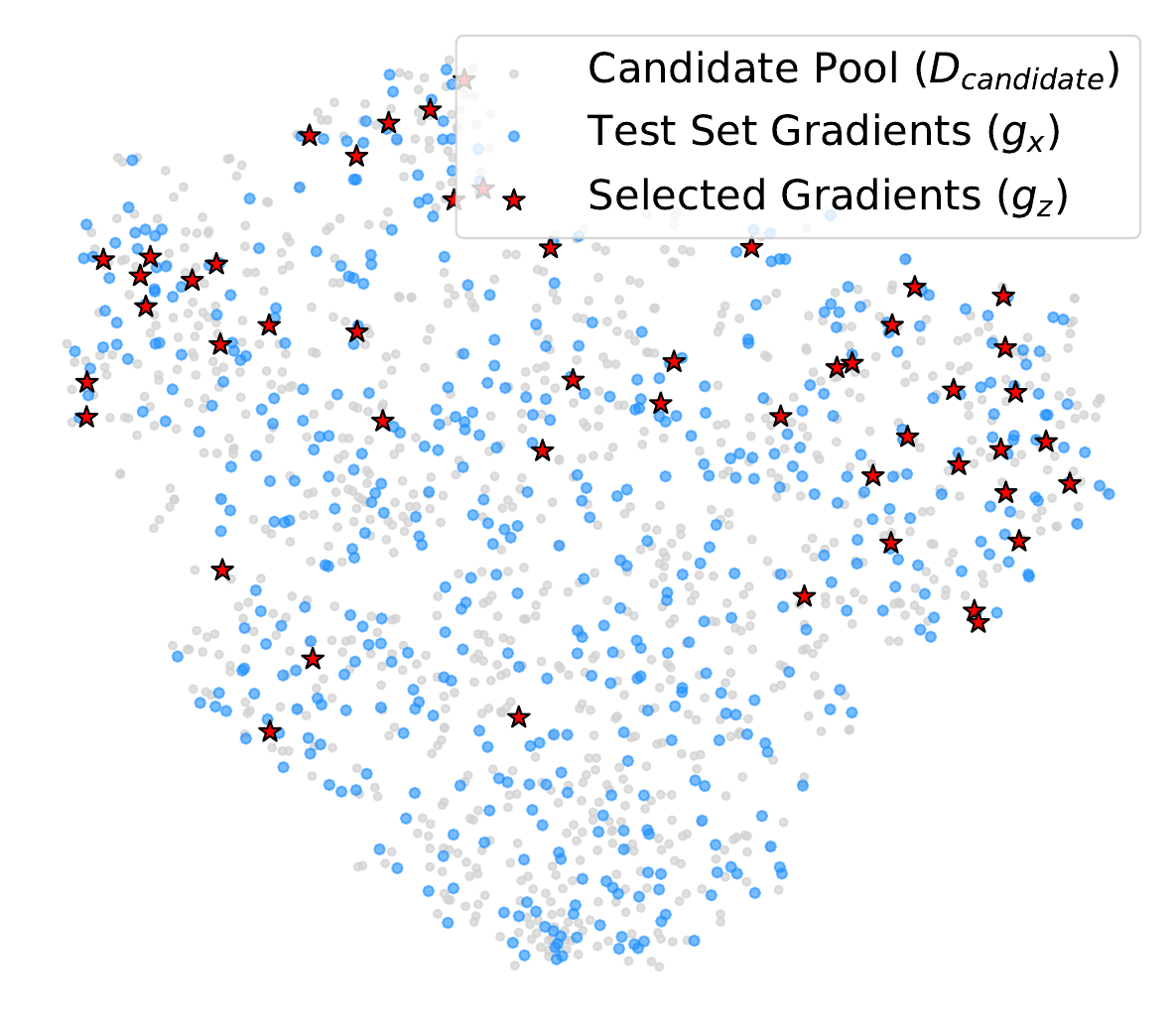}
         \caption{Greedy-Surrogate (Proposed).}
         \label{fig:surro-tsne-surro}
     \end{subfigure}
     \hfill
     \begin{subfigure}[t]{0.3\textwidth}
         \centering
         \includegraphics[width=\textwidth]{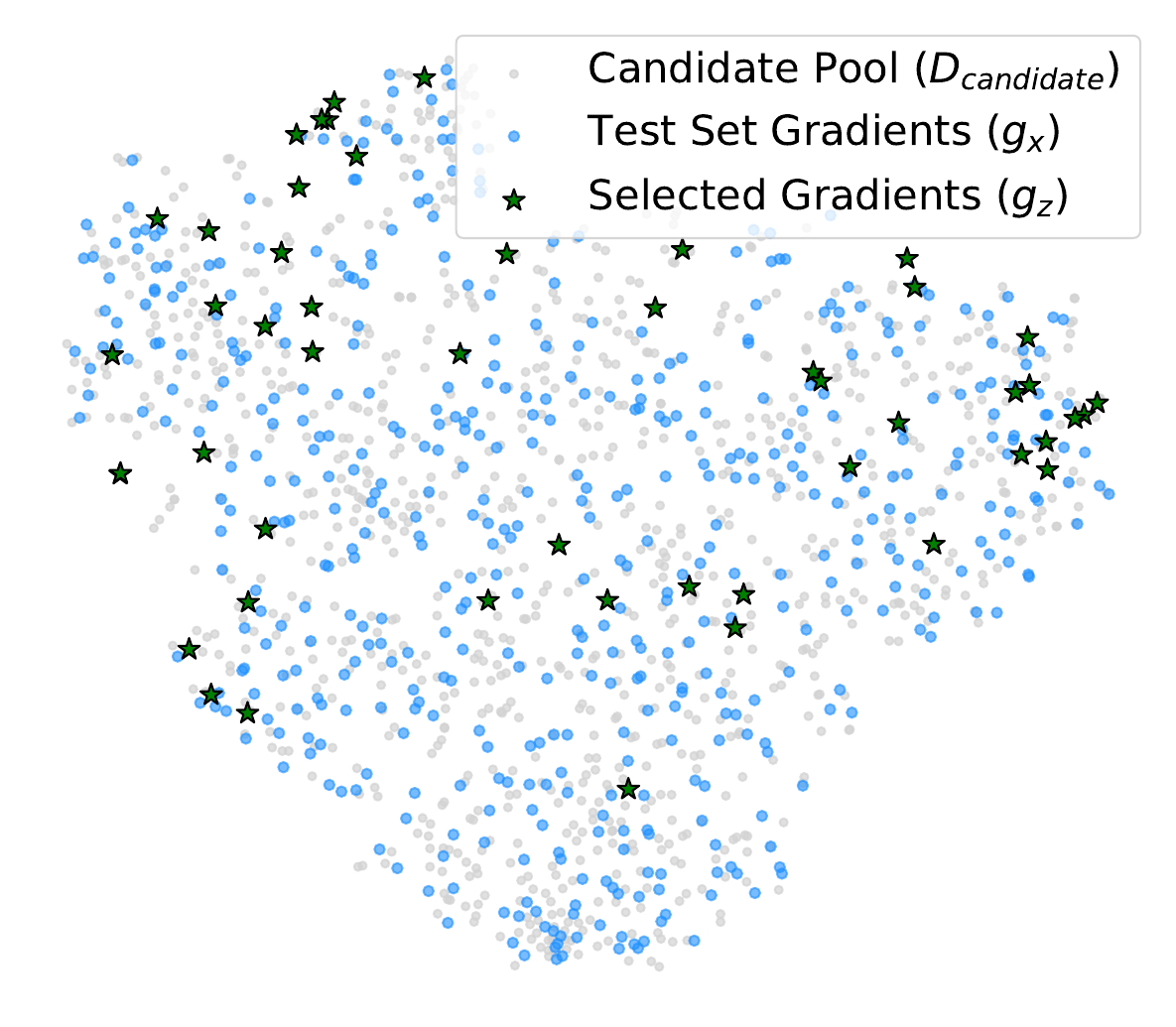}
         \caption{Greedy-FIM (Oracle).}
         \label{fig:surro-tsne-fim}
     \end{subfigure}
        \caption{
        Visualization of selected gradients.
        }
        \label{fig:surro-tsne}
\end{figure}
\begin{figure}[t]
     \centering
     \begin{subfigure}[t]{0.3\textwidth}
         \centering
         \includegraphics[width=\textwidth]{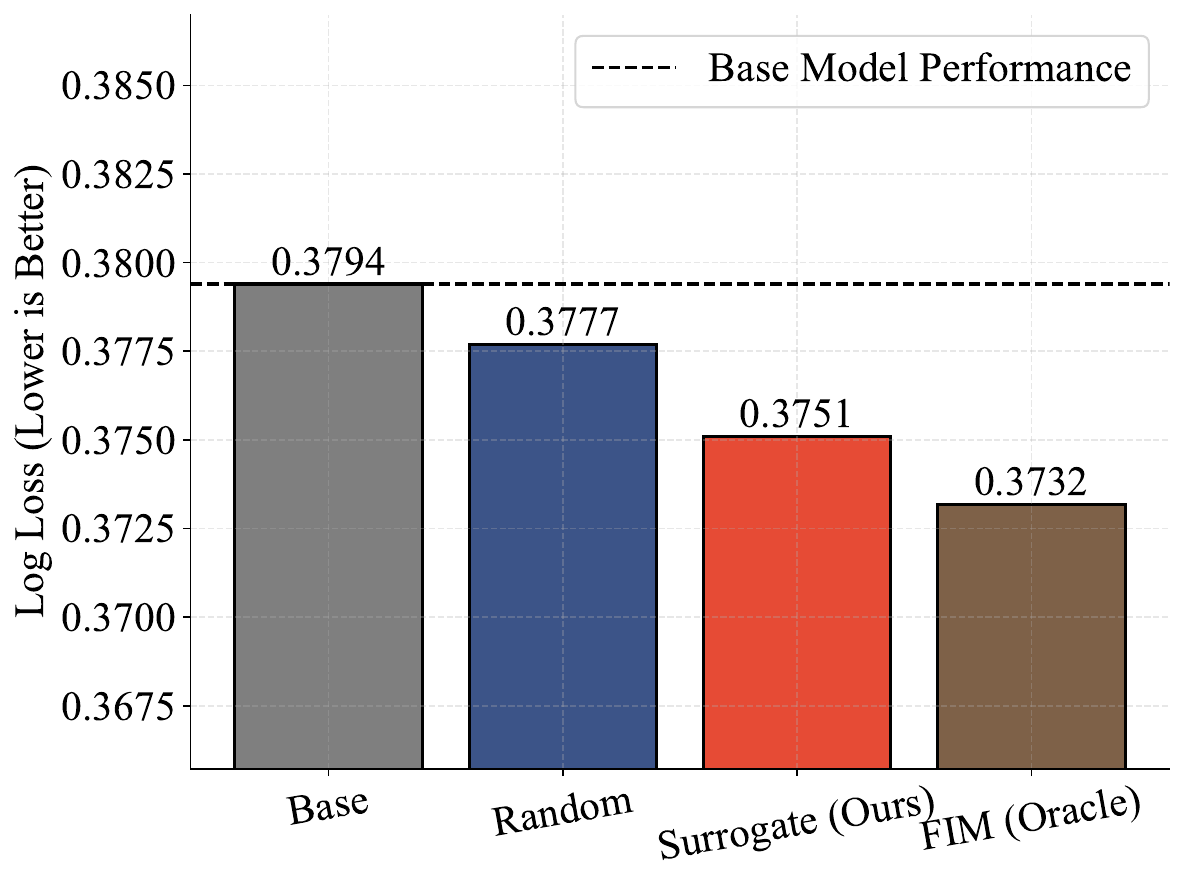}
         \caption{Log Loss.}
         \label{fig:surro-log-loss}
     \end{subfigure}
     \hspace{2cm}
     \begin{subfigure}[t]{0.3\textwidth}
         \centering
         \includegraphics[width=\textwidth]{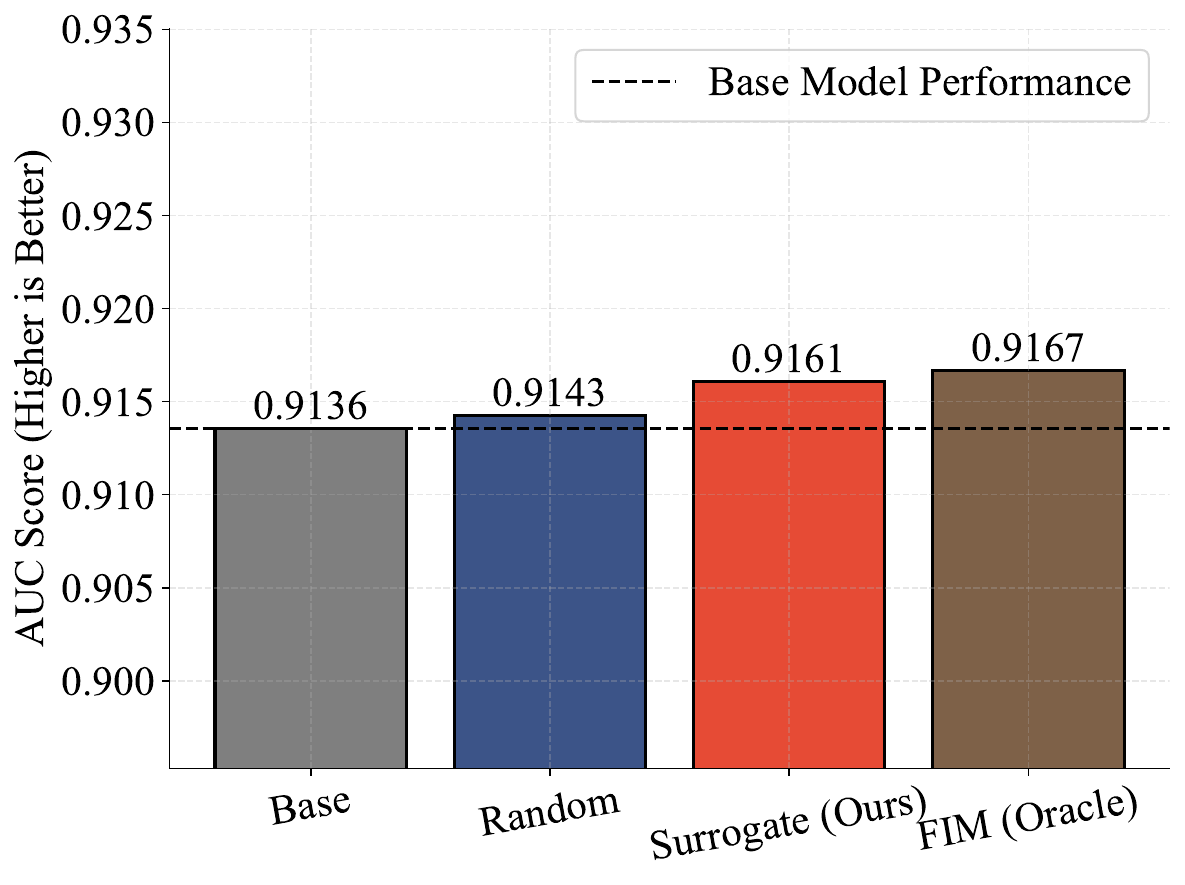}
         \caption{AUC.}
         \label{fig:surro-auc}
     \end{subfigure}
        \caption{
        Performance comparison of selected gradients on continually training.
        Our surrogate-based method achieves a significant reduction in Log Loss and an increase in AUC, closely approaching the performance of the FIM oracle and substantially outperforming the random baseline.
        }
        \label{fig:surro-performance}
\end{figure}
The analysis is presented in two parts: a qualitative visualization of the selected sample gradients and a quantitative evaluation of the final model performance.

To build an intuition for why our selection strategy is effective, we first visualize the high-dimensional gradients of the selected samples in a 2D space using t-SNE \cite{maaten2008visualizing}. Figure~\ref{fig:surro-tsne} presents the results for all three methods. The goal is to select a set of candidate gradients (stars) that best represents the distribution of the test set gradients (blue dots), which symbolize the space of model uncertainty we aim to reduce.
As illustrated in Figure~\ref{fig:surro-tsne-surro}, our \texttt{Greedy-Surrogate} method selects a diverse set of samples whose gradients are well-distributed across the embedding space. These selected points provide representative coverage of the different clusters formed by the test set gradients. In stark contrast, the \texttt{Random} selection strategy, shown in Figure~\ref{fig:surro-tsne-random}, results in a clustered selection that is concentrated in a dense region of the candidate pool, failing to capture the full diversity of the test set gradients. Crucially, the selections made by our surrogate method bear a strong qualitative resemblance to those made by the FIM oracle (Figure~\ref{fig:surro-tsne-fim}). This visual evidence strongly supports our hypothesis that maximizing the surrogate objective $U(S)$ is an effective proxy for selecting a diverse and informative set of samples, much like directly optimizing the Fisher Information Matrix.

Beyond qualitative assessment, we quantitatively evaluate the impact of the selected data by training new models on the augmented datasets and measuring their performance on a held-out test set. Figure~\ref{fig:surro-performance} displays the final Test Log Loss and AUC for each strategy. 
The model trained with data selected by our \texttt{Greedy-Surrogate} strategy demonstrates a marked reduction in Log Loss and a substantial increase in AUC. This confirms that the diverse samples it identified are highly valuable for improving model generalization. Most importantly, our method's performance is nearly on par with the \texttt{Greedy-FIM} oracle, which represents a practical upper bound for performance in this setting. Both strategies significantly outperform the naive random selection baseline, which, while beneficial, proves to be a suboptimal approach for acquiring the most informative labels. This quantitative result validates that our computationally efficient surrogate objective successfully guides the selection process towards a near-optimal state, leading to a measurably superior model.

\subsection{Experiment 2: Budget Feasibility}\label{sec:exp-budget-feas}
\subsubsection{Setup}


The objective of this experiment is to empirically verify the budget feasibility guarantee of our two-stage bidding framework (Theorem~\ref{thm:budget-feas}) and to analyze the effectiveness of the pacing controlled by the dual variable $\lambda$. We specifically investigate the algorithm's sensitivity to its key hyperparameters: the learning rate $\eta$ and the total budget $B$.

\textbf{Methodology.}
To isolate the behavior of the pacing controller, we conduct a series of controlled simulations. We simulate a stream of $T$ auctions where the per-impression marginal utility, $\Delta_t$, and the market-clearing price (i.e., the highest competitor bid) are drawn from random distributions at each timestep. This allows us to focus exclusively on the dynamics of the budget allocation algorithm from Section~\ref{sec:two-stage-bidding}.

The core of the experiment involves running two sets of simulations:
\begin{enumerate}
    \item \textbf{Sensitivity to Learning Rate ($\eta$):} We fix the total budget $B$ and the campaign length $T$. We then run the full simulation multiple times for a range of $\eta$ values. To ensure statistical robustness, we execute $N=30$ independent trials for each $\eta$ value and aggregate the results.
    \item \textbf{Sensitivity to Total Budget ($B$):} We fix the learning rate $\eta$ to a well-performing value identified in the first experiment. We then run the simulation for a wide range of total budget values, from small to large, to assess the algorithm's scalability and reliability.
\end{enumerate}

\textbf{Evaluation Metrics.}
The performance of the pacing is evaluated using the following metrics:
\begin{itemize}
    \item \textbf{Mean Absolute Error (MAE) vs. $\eta$:} For each $\eta$, we compute the mean of the absolute relative spending error over all trials: $\text{MAE} = \mathbb{E}\left[ \left| \frac{\text{Cost}_{\text{final}} - B}{B} \right| \right]$. This metric quantifies the accuracy and stability of the controller as a function of its learning rate.
    \item \textbf{Final Spend vs. Target Budget:} For the second experiment, we plot the final expenditure against the target budget $B$. This visualizes the algorithm's absolute adherence to the budget constraint across different scales.
    \item \textbf{Relative Spending Error vs. Target Budget:} To complement the absolute plot, we also show the final relative error, $\frac{\text{Cost}_{\text{final}} - B}{B}$, for each budget $B$. This normalizes the error and shows if the algorithm's precision is maintained as the budget grows.
\end{itemize}

Detailed configurations are given in Appendix~\ref{sec:exp-set-app-budget-feas}.

\subsubsection{Result Analysis}
\begin{figure}[t]
     \centering
     \begin{subfigure}[t]{0.37\textwidth}
         \centering
         \includegraphics[width=\textwidth]{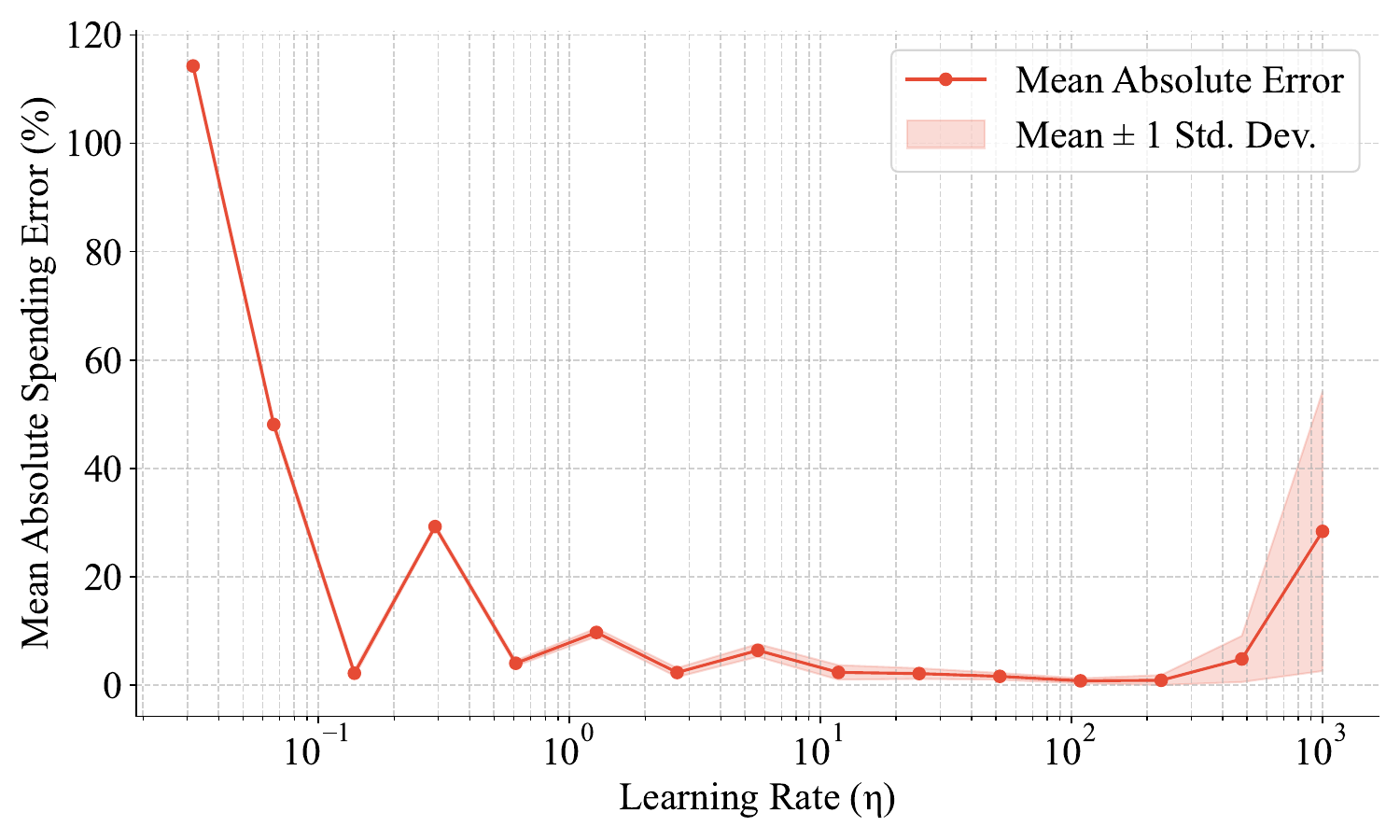}
         \caption{Mean-absolute error v.s. $\eta$.}
         \label{fig:budget-eta-impact}
     \end{subfigure}
     \hfill
     \begin{subfigure}[t]{0.30\textwidth}
         \centering
         \includegraphics[width=\textwidth]{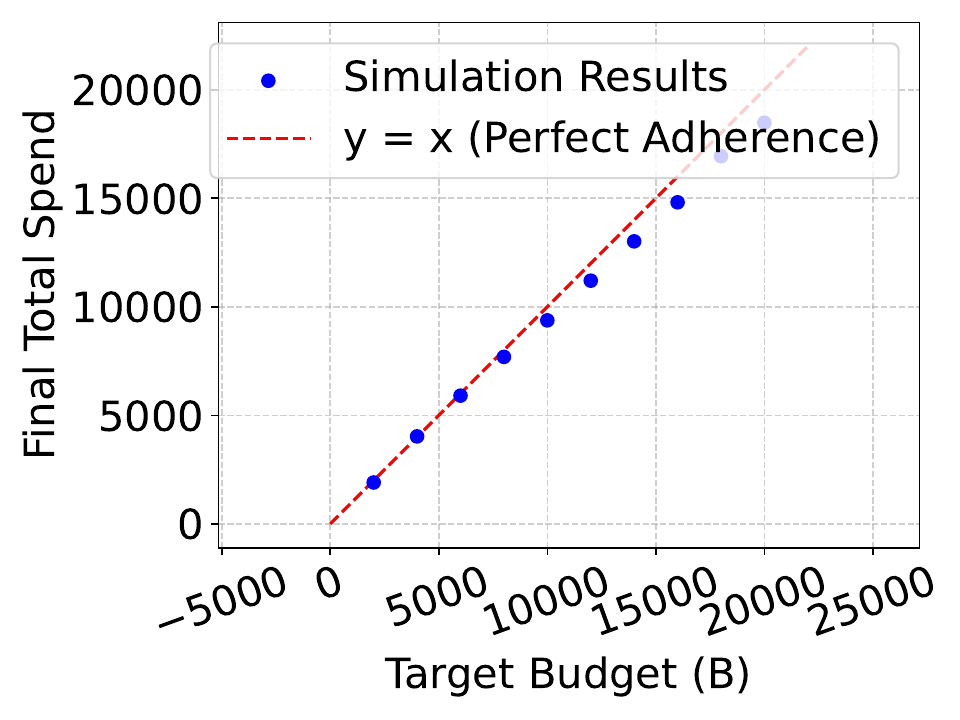}
         \caption{Budget spend wrt. budget.}
         \label{fig:budget-spend}
     \end{subfigure}
     \hfill
     \begin{subfigure}[t]{0.31\textwidth}
         \centering
         \includegraphics[width=\textwidth]{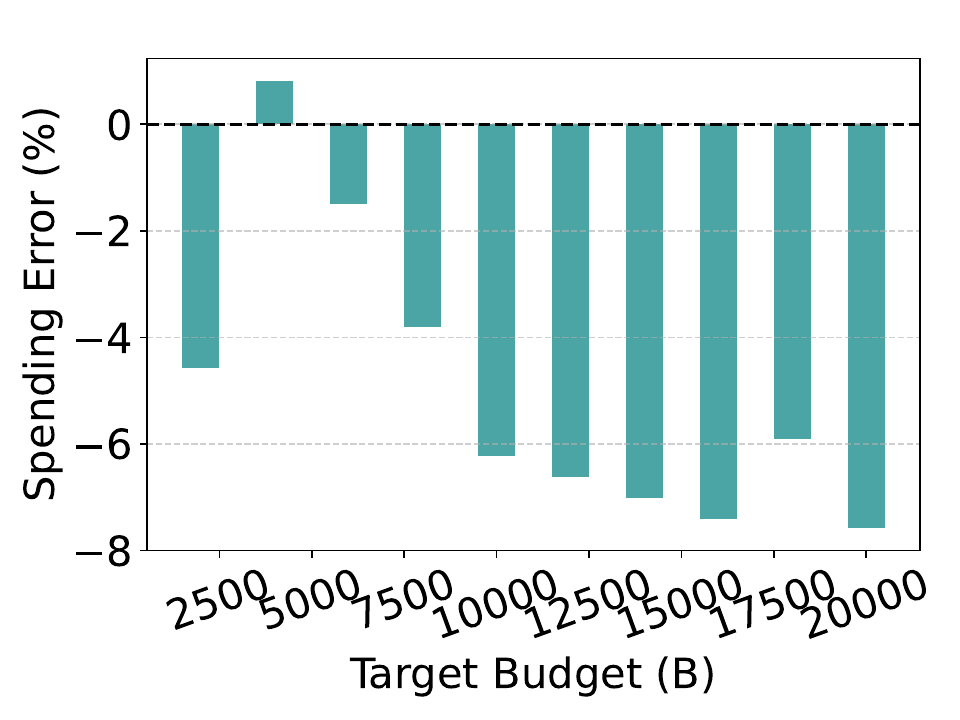}
         \caption{Relative spending error wrt. budget.}
         \label{fig:budget-error}
     \end{subfigure}
        \caption{Evaluation of Budget Feasibility and Pacing Dynamics.}
        \label{fig:budget}
\end{figure}

~\\
\textbf{Impact of Learning Rate $\eta$.}
Figure~\ref{fig:budget-eta-impact} illustrates the critical role of the learning rate $\eta$ in balancing the controller's responsiveness and stability. The plot of Mean Absolute Error versus $\eta$ exhibits a distinct U-shape, which is characteristic of a well-behaved control system. For very low values of $\eta$, the controller is too sluggish; it reacts too slowly to deviations from the ideal spending pace, resulting in a significant final error. As $\eta$ increases, the controller becomes more effective at correcting its course, and the mean error converges towards a minimum, indicating an optimal operational range. 
However, if $\eta$ becomes too large, the controller becomes overly aggressive and unstable. It overcorrects in response to small deviations, leading to oscillations in the shadow price $\lambda$ and a subsequent increase in the final spending error. This result confirms that $\eta$ is a crucial, tunable hyperparameter that allows us to achieve precise budget control, with a clear optimal region that avoids both sluggishness and instability.

\textbf{Robustness to Total Budget $B$.}
Figure~\ref{fig:budget-spend} provides a comprehensive validation of the algorithm's budget feasibility guarantee (Theorem~\ref{thm:budget-feas}), which plots the final expenditure against the target budget for multiple campaigns. The data points lie almost perfectly along the $y=x$ diagonal, demonstrating that the algorithm adheres to the specified budget with remarkable accuracy, regardless of whether the budget is small or large. 
Figure~\ref{fig:budget-error} shows the relative spending error across different budget levels. The errors are consistently contained within a very narrow band around 0\%. It shows that the algorithm's precision is not merely absolute but also relative; the percentage error does not increase as the budget grows. Together, these results provide strong empirical evidence that our two-stage bidding framework with its dynamic pacing is a fiscally responsible and predictable tool, capable of managing creator budgets reliably across a wide range of scales. 

\subsection{Experiment 3: Gradient Estimation}\label{sec:exp-grad-est}
\subsubsection{Setup}

    

The objective of this experiment is to evaluate the accuracy of our confidence-gated heuristic for estimating loss gradients in the absence of a true label (as detailed in Section~\ref{sec:gradient-estimation}). Furthermore, we aim to quantify its performance when analytical gradients are unavailable, necessitating the use of a Zeroth-Order (ZO) estimator, thereby simulating a black-box model environment.

\textbf{Methodology.}
We use a pre-trained pCTR model (Logistic Regression) and a held-out labeled test set. The protocol is as follows: for each sample $\mathbf{x}_t$ in the test set, we first generate a proxy gradient $\hat{\mathbf{g}}_t$ using several competing methods \textit{without} using the true label $y_t$. We then compute the true analytical gradient, $\mathbf{g}_{\text{true}}$, using the known label $y_t$. The accuracy of each proxy gradient is measured against this ground truth.

\textbf{Compared Methods.}
We compare four distinct strategies for generating the proxy gradient $\hat{\mathbf{g}}_t$:
\begin{itemize}
    \item \textbf{Our Heuristic (Analytical):} This is our primary proposed method from Section~\ref{sec:gradient-estimation}. It computes two hypothetical \textit{analytical} gradients, $\mathbf{g}_0$ (for label $y=0$) and $\mathbf{g}_1$ (for label $y=1$), and selects the one with the smaller L2-norm as the proxy: $\hat{\mathbf{g}} = \arg\min_{\mathbf{g} \in \{\mathbf{g}_0, \mathbf{g}_1\}} \|\mathbf{g}\|_2$.
    \item \textbf{Our Heuristic (ZO):} This method evaluates our heuristic in a black-box setting. It uses the same L2-norm selection rule but applies it to gradients approximated via a two-point Zeroth-Order (ZO) estimator. For each hypothetical label, the gradient is estimated as $\hat{\mathbf{g}}_{\text{ZO}} \approx \frac{1}{2\mu} [L(\theta+\mu\mathbf{u}) - L(\theta-\mu\mathbf{u})] \mathbf{u}$, averaged over multiple random directions $\mathbf{u}$.
    \item \textbf{pCTR-Weighted (Analytical):} A common baseline that computes an expected gradient based on the model's prediction $\hat{p}_t$: $\hat{\mathbf{g}}_{\text{est}} = \hat{p}_t \cdot \mathbf{g}_1 + (1-\hat{p}_t) \cdot \mathbf{g}_0$.
    \item \textbf{Random-Guess (Analytical):} Randomly selects either $\mathbf{g}_0$ or $\mathbf{g}_1$ as the proxy.
\end{itemize}

\textbf{Metrics.}
We quantify the accuracy of each estimated gradient $\hat{\mathbf{g}}_t$ against the true gradient $\mathbf{g}_{\text{true}}$ using two standard metrics:
\begin{itemize}
    \item \textbf{Cosine Similarity:} Measures the alignment of the vectors' directions. A value closer to 1 indicates a more accurate estimation of the learning direction.
    \item \textbf{L2 Distance:} Measures the Euclidean distance between the vectors. A smaller value indicates a more accurate estimation of both direction and magnitude.
\end{itemize}
We report the average performance of each method over all test samples, with a specific focus on the subset of samples where the model's prediction is highly confident (i.e., $\hat{p}_t$ is close to 0 or 1), as this is the regime where our primary heuristic is designed to operate.

Detailed configurations are given in Appendix~\ref{sec:exp-set-app-grad-est}.

\begin{figure}[t]
     \centering
     \begin{subfigure}[t]{0.245\textwidth}
         \centering
         \includegraphics[width=\textwidth]{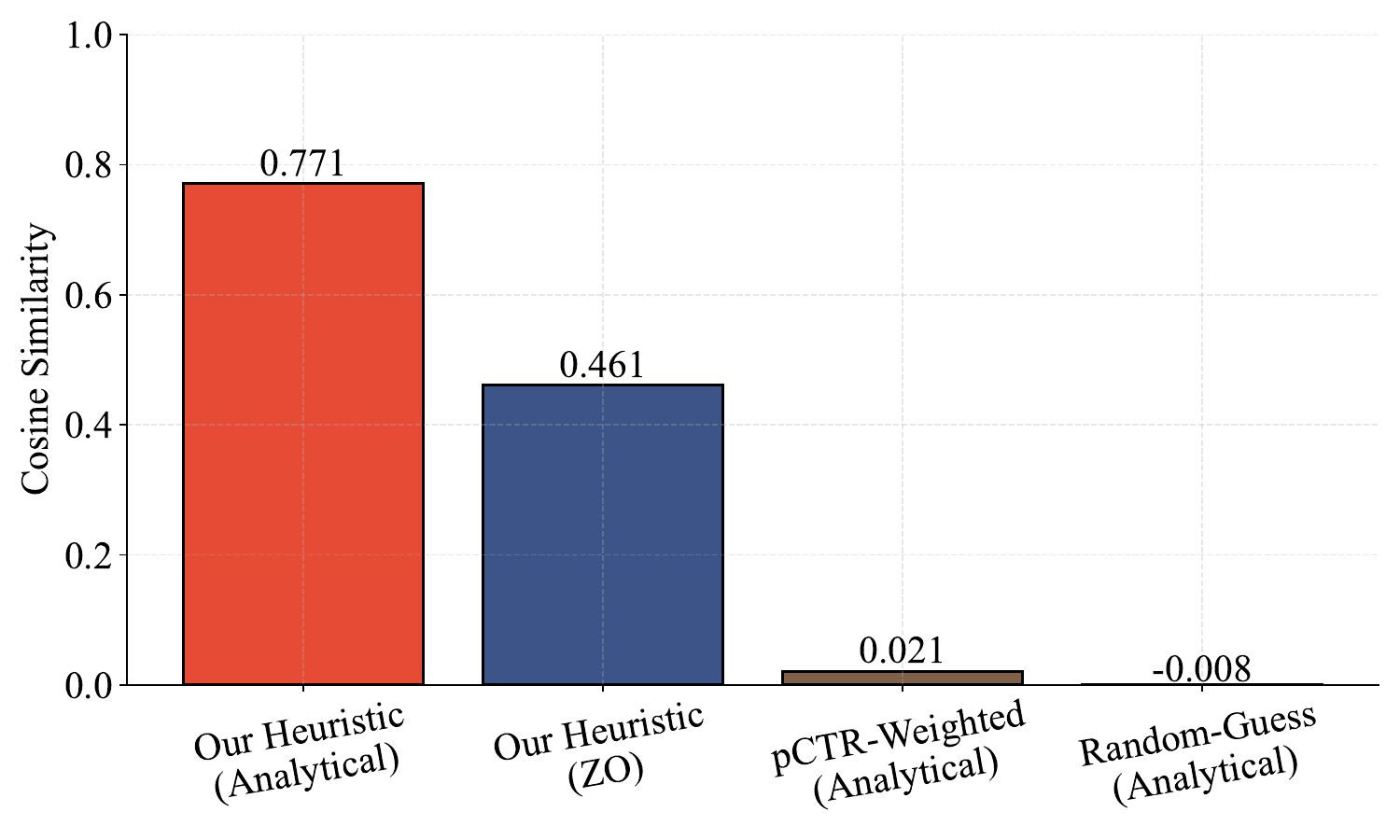}
         \caption{Cosine Similarity on all samples.}
         \label{fig:grad-est-cos-all}
     \end{subfigure}
     \hfill
     \begin{subfigure}[t]{0.245\textwidth}
         \centering
         \includegraphics[width=\textwidth]{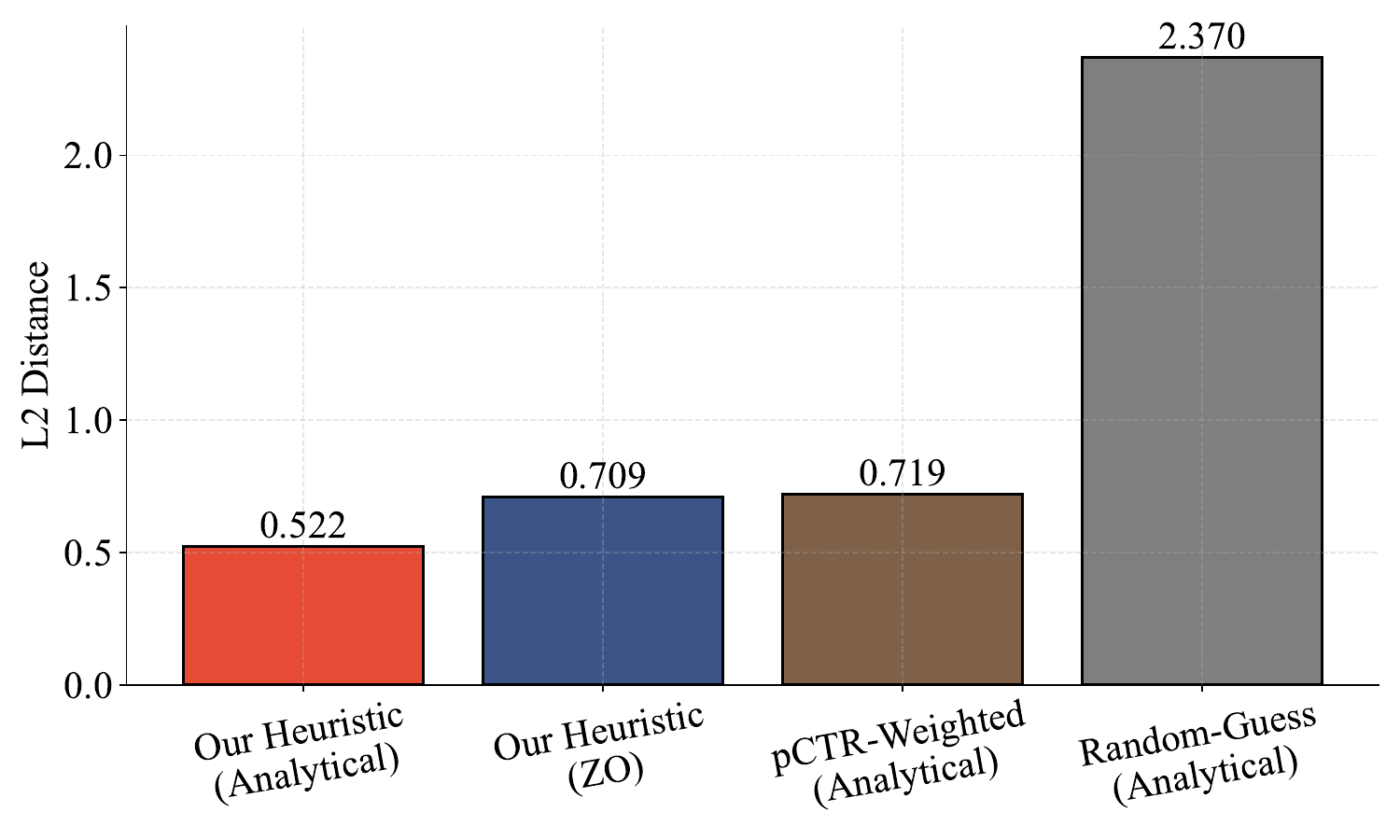}
         \caption{L-2 distance on all samples.}
         \label{fig:grad-est-l2-all}
     \end{subfigure}
     \hfill
     \begin{subfigure}[t]{0.245\textwidth}
         \centering
         \includegraphics[width=\textwidth]{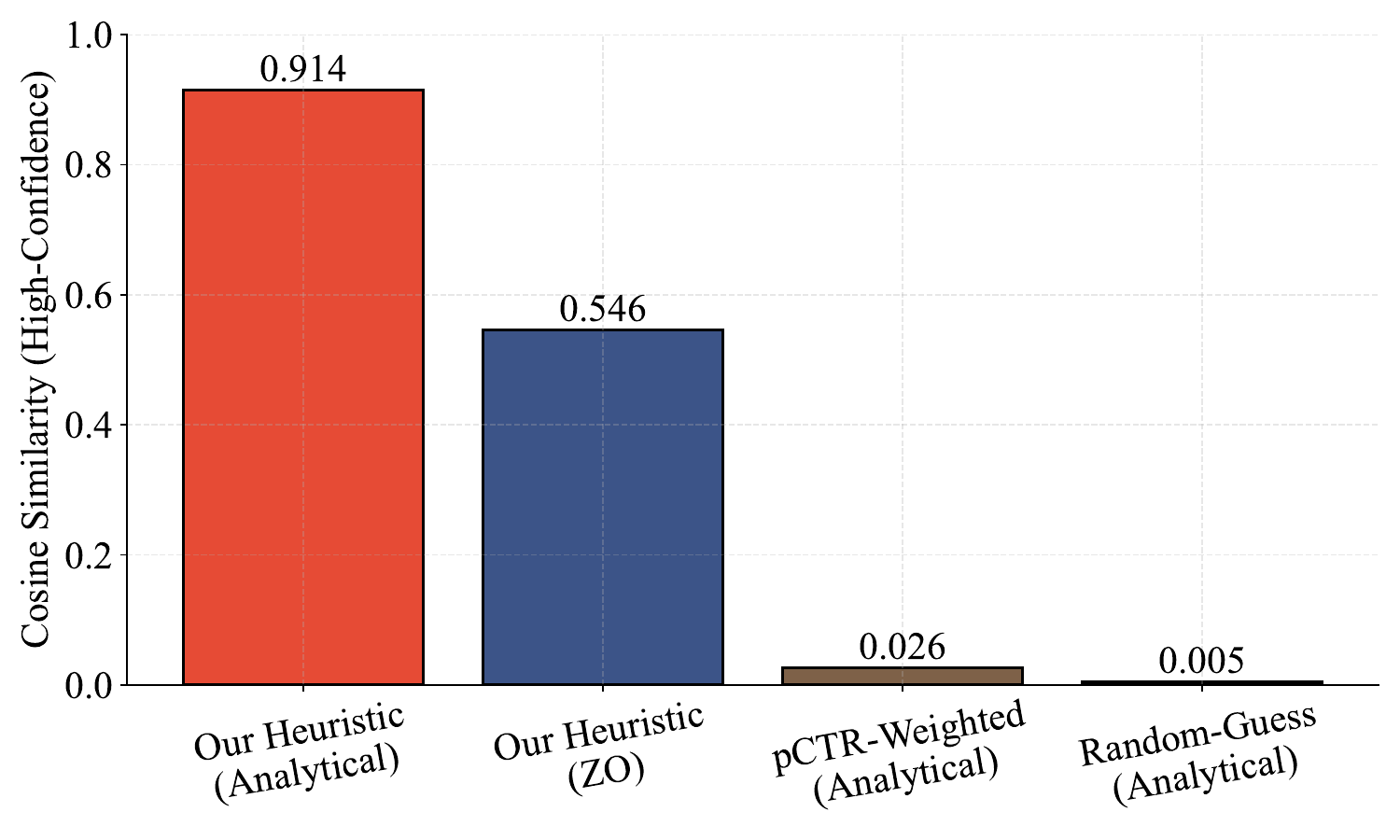}
         \caption{Cosine Similarity on high-confidence samples.}
         \label{fig:grad-est-cos-hc}
     \end{subfigure}
     \hfill
     \begin{subfigure}[t]{0.245\textwidth}
         \centering
         \includegraphics[width=\textwidth]{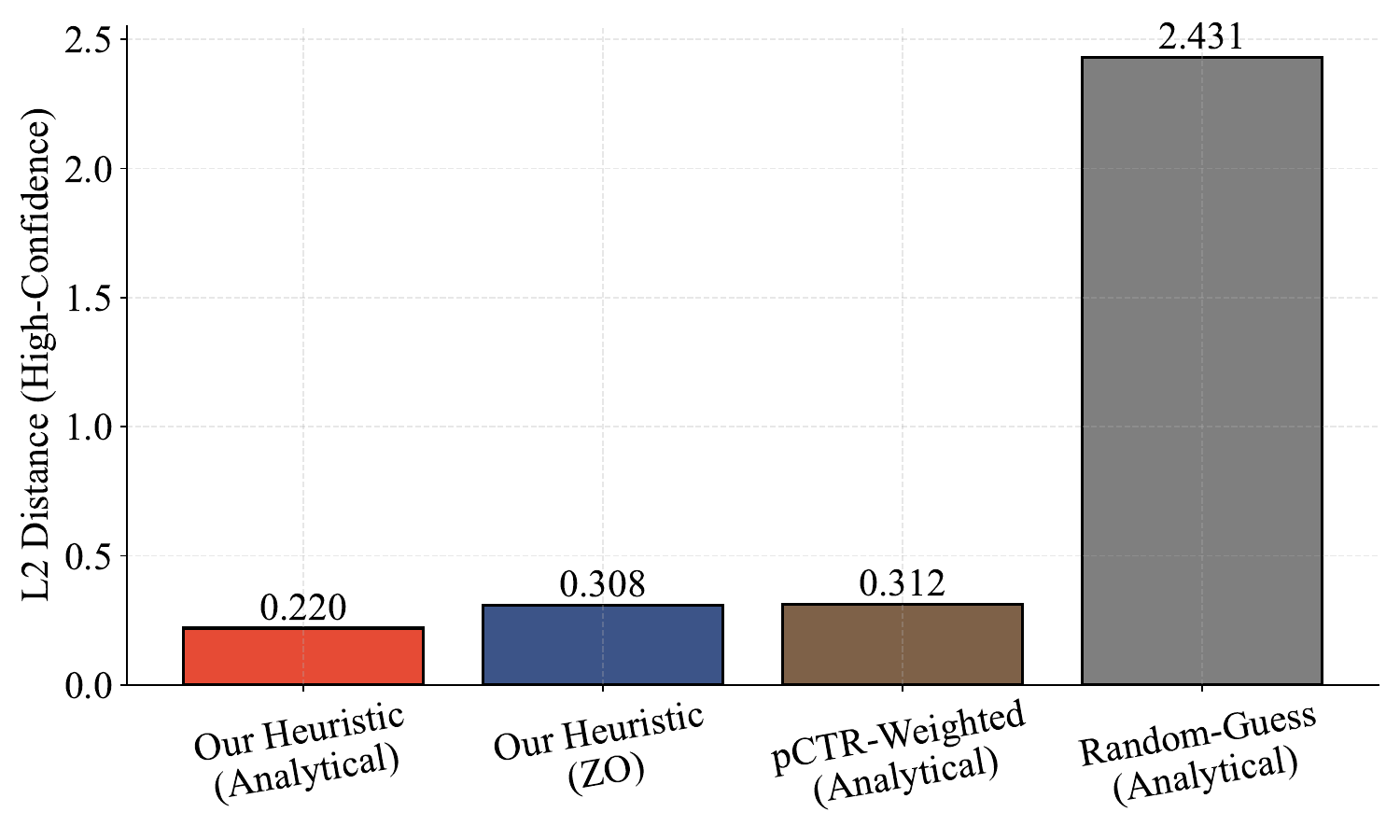}
         \caption{L-2 distance on high-confidence samples.}
         \label{fig:grad-est-l2-hc}
     \end{subfigure}
        \caption{
        Performance of gradient estimation heuristics on all and high-confidence samples. Our heuristic using analytical gradients (green) performs best. The ZO-based version (red) shows a slight, expected performance degradation but still significantly outperforms the pCTR-weighted (orange) and random (gray) baselines in both cosine similarity (higher is better) and L2 distance (lower is better).
        }
        \label{fig:grad-est}
\end{figure}

\subsubsection{Result Analysis}
The results, presented in Figure~\ref{fig:grad-est}, confirm the effectiveness of our proposed heuristic and its robustness in a black-box setting. 

First, the results clearly demonstrate the superiority of our L2-norm-based heuristic when analytical gradients are available. As shown by the green bars in Figures~\ref{fig:grad-est-cos-all} and \ref{fig:grad-est-l2-all}, our analytical heuristic achieves the highest performance, with a cosine similarity approaching 1.0 and the lowest L2 distance, respectively. This is because its core assumption, that the gradient corresponding to the correct label will have a smaller L2-norm, is most valid in this high-confidence regime. 
Our heuristic successfully exploits this difference in magnitude to identify the more likely gradient.

Second, we analyze the performance under a more challenging black-box scenario where gradients are approximated using a Zeroth-Order (ZO) estimator. As shown by the red bars, there is an expected performance gap compared to using exact analytical gradients. The cosine similarity decreases and the L2 distance increases, quantifying the inherent cost of information loss from the ZO approximation. However, what is crucial is that our heuristic, even when applied to these noisy ZO gradients, still substantially outperforms the other baselines. This demonstrates the robustness of the L2-norm selection rule itself: it remains effective at identifying the more plausible learning signal even when the input gradients are imprecise.

Finally, the baselines falter significantly in this high-confidence setting, as shown in Figure~\ref{fig:grad-est-cos-hc} and Figure~\ref{fig:grad-est-l2-hc}. The pCTR-weighted baseline (orange bars) is particularly brittle. When the model is highly confident but incorrect (e.g., predicting $\hat{p}_t=0.99$ when the true label is $y=0$), its estimate is overwhelmingly skewed towards the wrong gradient, leading to a very poor approximation precisely when the sample is most informative. The random-guess baseline (gray bars) performs poorly by design, serving as a lower bound.

In conclusion, this experiment validates that our L2-norm heuristic is an effective method for selecting a proxy gradient in high-confidence scenarios. Furthermore, its strong performance even with noisy ZO inputs confirms its practical utility for black-box models where direct gradient computation is impossible.

\subsection{Experiment 4: End-to-End Offline Case Study}\label{sec:exp-end-to-end}
\subsubsection{Setup}
To evaluate the overall performance of the complete bidding framework in a simulated environment, measuring its ability to improve model performance over time under a fixed budget.

\textbf{Datasets.} Our offline experiments will primarily utilize two datasets:
\begin{itemize}
    \item \textbf{Synthesis Dataset: } We synthesis the feature of each ad and the corresponding binary click feedback. The synthesis data is used to train a CTR model.
    \item \textbf{Criteo Dataset \cite{zhang2014real}:} A popular public benchmark for CTR prediction, containing a large volume of anonymized display advertising data. Its high-dimensional sparse features make it a suitable testbed for evaluating the performance of the pCTR model under uncertainty.
\end{itemize}
For all experiments, we split the data chronologically into training, validation, and testing sets to prevent temporal data leakage.

\textbf{Model.} We employ a standard DCM model \cite{DiemertMeynet2017} as the pCTR model, a common and effective architecture for this task. The model takes the concatenated sparse and dense features as input and outputs the predicted CTR, $\sigma(x)$.

\textbf{Methodology.}
We conduct a full simulation on the Criteo dataset. An initial pCTR model is trained on a small, early portion of the data.
We stream the subsequent data as a sequence of auctions. For each auction, we simulate competing bids from other advertisers.
Our proposed bidder and several baselines compete to win impressions subject to the same total budget $B$. The winning impressions are collected into a set $S_{\text{won}}$.
Our method is compared with the following baselines:
    \begin{itemize}
        \item \textbf{Value-Only ($\beta=1$):} Our framework configured to only maximize immediate pCTR value.
        \item \textbf{Uncertainty-Only ($\beta=0$):} Our framework configured to only maximize gradient coverage.
        \item \textbf{Uniform Bidding:} A standard baseline that bids a constant value for every impression.
        \item \textbf{pCTR-Linear Bidding:} The bid is proportional to the predicted CTR.
    \end{itemize}
After the campaign simulation is complete, for each method, we retrain a new model using the initial training data plus the set of won impressions, $S_{\text{won}}$.

Detailed setup is given in Appendix \ref{sec:exp-set-app-end-to-end}.

\textbf{Metrics and Expected Outcome.} We evaluate the final retrained models on a held-out test set, measuring performance using AUC and LogLoss. We hypothesize that our proposed method (with a tuned $\beta \in (0,1)$) will yield a final model with the best performance, outperforming both the value-only and uncertainty-only extremes. This would demonstrate that strategically balancing short-term value acquisition and long-term uncertainty reduction leads to superior model improvement.

\begin{figure}[t]
     \centering
     \begin{subfigure}[t]{0.24\textwidth}
         \centering
         \includegraphics[width=\textwidth]{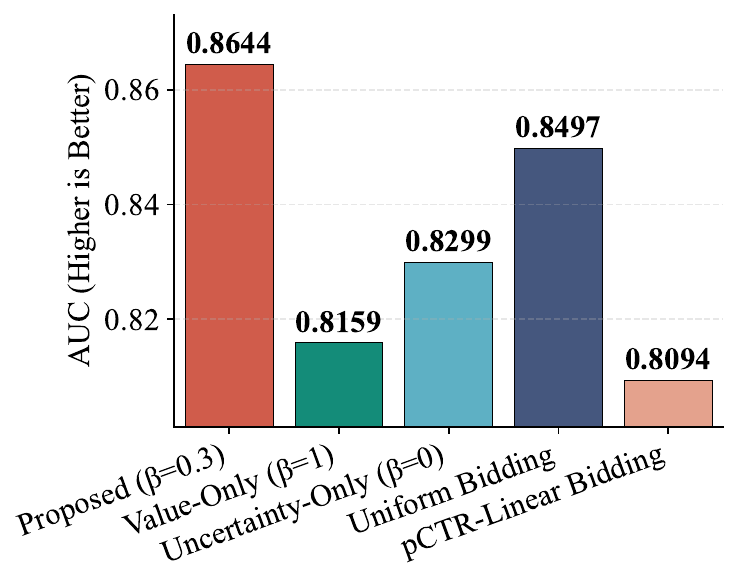}
         \caption{AUC.}
         \label{fig:syn-auc}
     \end{subfigure}
     \hfill
     \begin{subfigure}[t]{0.24\textwidth}
         \centering
         \includegraphics[width=\textwidth]{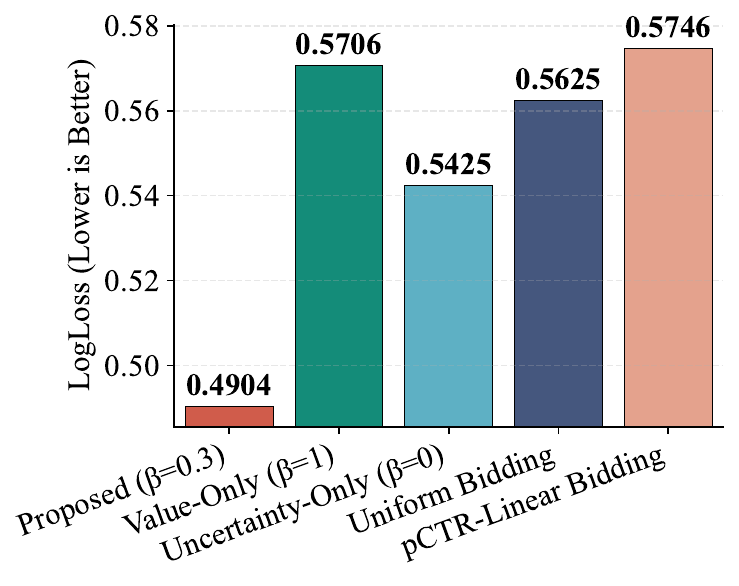}
         \caption{Log loss.}
         \label{fig:syn-logloss}
     \end{subfigure}
     \hfill
     \begin{subfigure}[t]{0.24\textwidth}
         \centering
         \includegraphics[width=\textwidth]{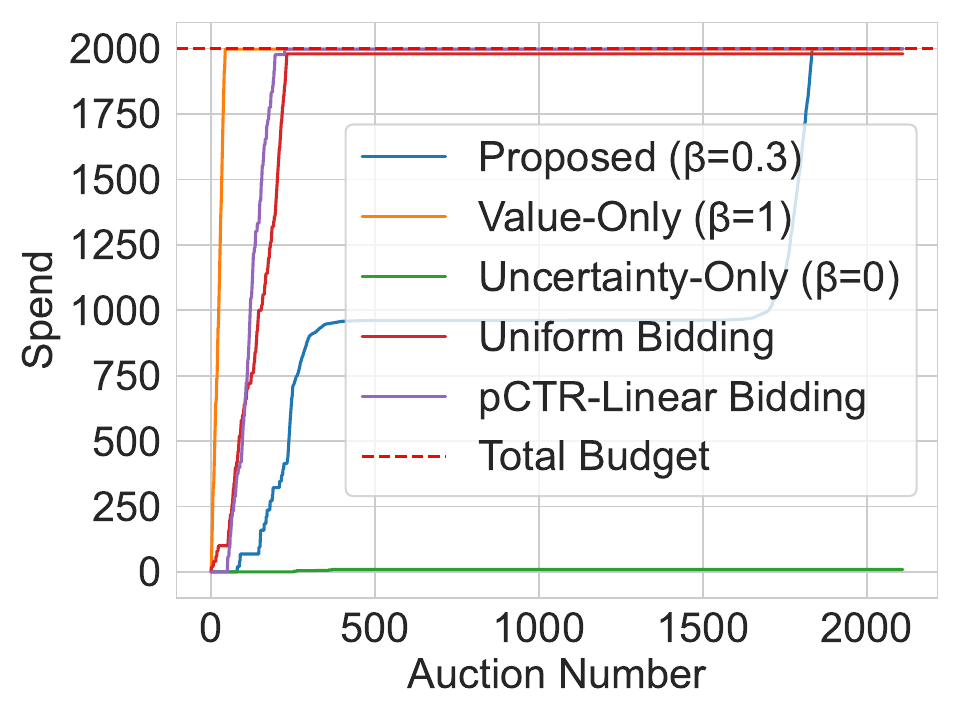}
         \caption{Spending vs. time.}
         \label{fig:syn-spend-time}
     \end{subfigure}
     \hfill
     \begin{subfigure}[t]{0.24\textwidth}
         \centering
         \includegraphics[width=\textwidth]{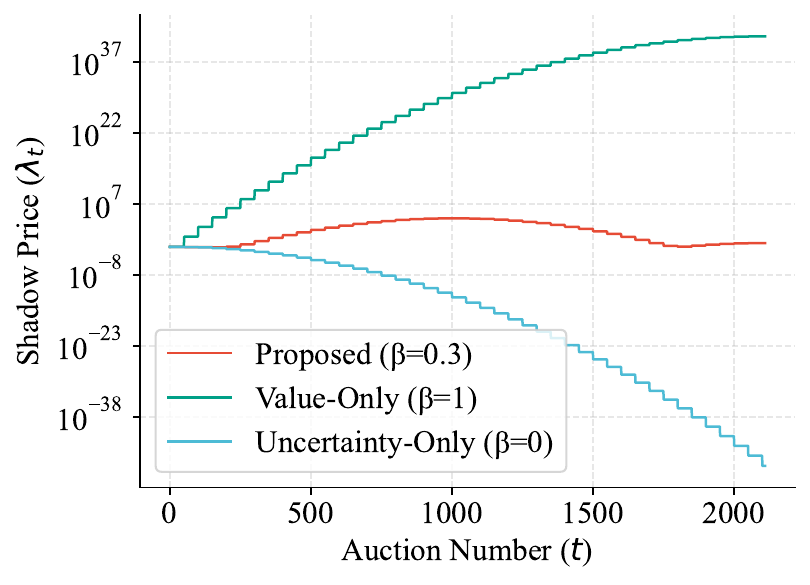}
         \caption{Dynamics of dual variable ($\lambda$).}
         \label{fig:syn-lambda-evolution}
     \end{subfigure}
        \caption{
        Offline Evaluation on Synthesis Dataset.
        }
        \label{fig:syn}
\end{figure}
\begin{figure}[t]
     \centering
     \begin{subfigure}[t]{0.24\textwidth}
         \centering
         \includegraphics[width=\textwidth]{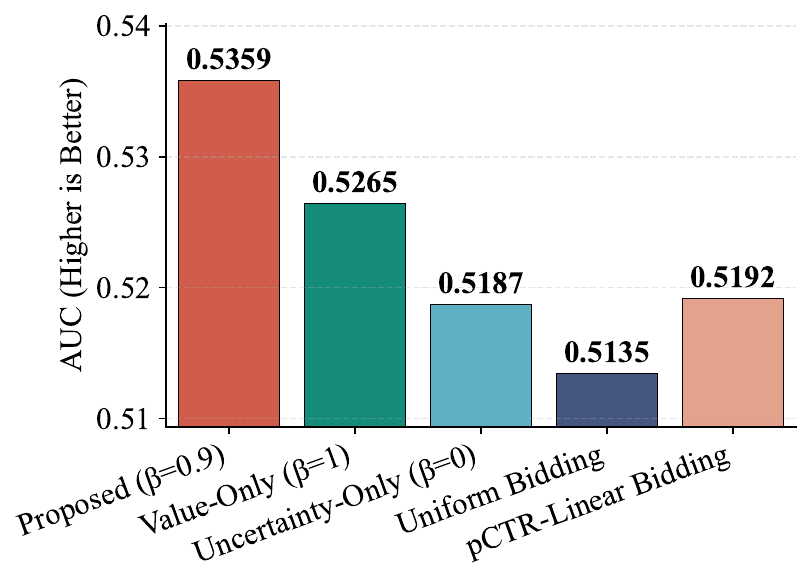}
         \caption{AUC.}
         \label{fig:critero-auc}
     \end{subfigure}
     \hfill
     \begin{subfigure}[t]{0.24\textwidth}
         \centering
         \includegraphics[width=\textwidth]{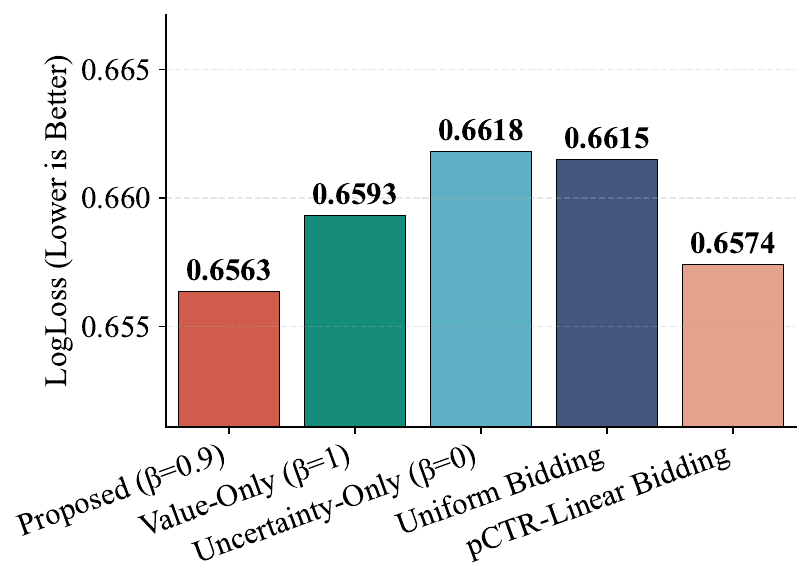}
         \caption{Log loss.}
         \label{fig:critero-logloss}
     \end{subfigure}
     \hfill
     \begin{subfigure}[t]{0.24\textwidth}
         \centering
         \includegraphics[width=\textwidth]{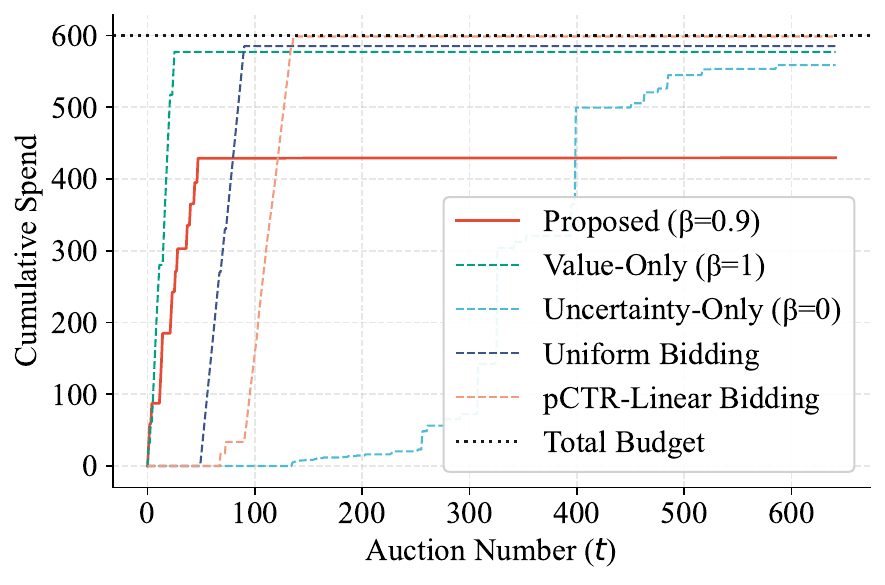}
         \caption{Spending vs. time.}
         \label{fig:critero-spend-time}
     \end{subfigure}
     \hfill
     \begin{subfigure}[t]{0.24\textwidth}
         \centering
         \includegraphics[width=\textwidth]{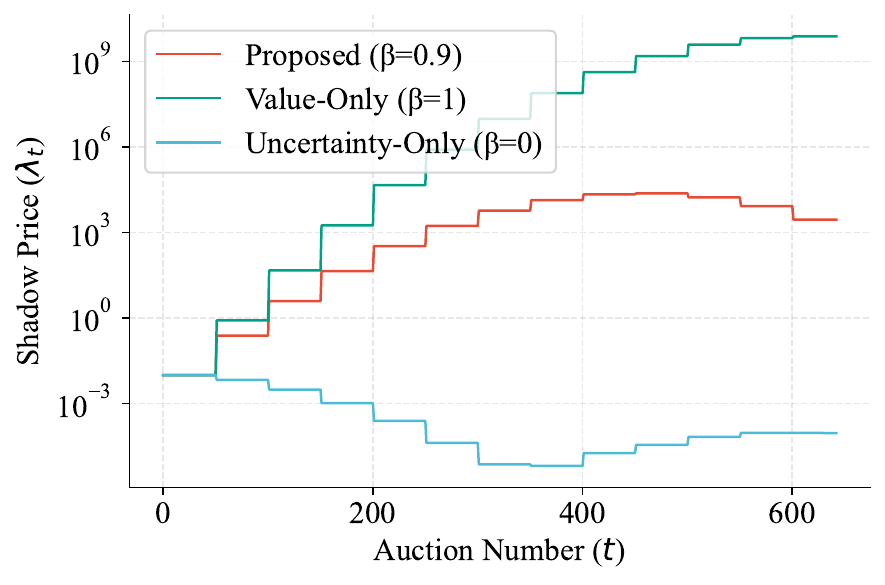}
         \caption{Dynamics of dual variable ($\lambda$).}
         \label{fig:critero-lambda-evolution}
     \end{subfigure}
        \caption{
        Offline Evaluation on Critero Dataset.
        }
        \label{fig:critero}
\end{figure}
\subsubsection{Result Analysis}

In this section, we analyze the results of the offline end-to-end simulation.
We demonstrate that our proposed bidding strategy can improve the long-term performance of the pCTR model by strategically acquiring training samples under a fixed budget\footnote{Due to the space limit, there are some statistics obscured in Figures~\ref{fig:syn} and \ref{fig:critero}. We provide a better view in the Appendix \ref{sec:exp-set-app-results}.}. Figures~\ref{fig:syn-auc} \ref{fig:critero-auc} and Figures~\ref{fig:syn-logloss} \ref{fig:critero-logloss} present the Area Under the Curve (AUC) and LogLoss of the final models, retrained on the impressions won by each bidding strategy.
The results compellingly validate our central hypothesis. The proposed method achieves the highest AUC and the lowest LogLoss, significantly outperforming all other baselines. This success stems from its unified objective function $F(\mathcal{S})$, which judiciously balances two critical goals:
\begin{enumerate}
    \item \textbf{Short-term Value Acquisition:} The $\beta \cdot V(\mathcal{S})$ term guides the bidder to acquire impressions with high predicted CTR, ensuring immediate campaign value.
    \item \textbf{Long-term Uncertainty Reduction:} The $(1-\beta) \cdot U(\mathcal{S})$ term, our gradient coverage surrogate, incentivizes the acquisition of diverse and informative samples. These samples, identified by the Zeroth-Order (ZO) gradient estimator, reside in regions of the feature space where the model is uncertain. Training on them effectively reduces model variance and improves generalization.
\end{enumerate}

Then, we show that our bidding algorithm is fiscally responsible. 
Figures~\ref{fig:syn-spend-time} and \ref{fig:critero-spend-time} illustrate the in-campaign dynamics of budget expenditure and the adaptive control of the dual variable $\lambda$.
As shown in Figures~\ref{fig:syn-spend-time} and \ref{fig:critero-spend-time}, it is possible for bidders based on our two-stage framework to exhibit smooth and controlled spending. Their cumulative expenditure curves can rise steadily and converge near the total budget limit, empirically validating the budget feasibility guarantee established in Theorem \ref{thm:budget-feas}. 
Figures~\ref{fig:syn-lambda-evolution} and \ref{fig:critero-lambda-evolution} reveal the underlying mechanism driving this control: the evolution of the dual variable $\lambda$. This variable represents the "shadow price" of the budget. We observe that $\lambda$ dynamically adjusts throughout the campaign. When a bidder's spending is ahead of the target pace, $\lambda$ increases, making bids more conservative (since bid $b_t \propto 1/\lambda_t$). Conversely, when spending is too slow, $\lambda$ decreases to encourage more aggressive bidding. This self-regulating behavior, an instance of online mirror descent, ensures that the budget is allocated intelligently across the entire campaign horizon, which is crucial for achieving the near-optimal performance guaranteed by our regret analysis (Theorem \ref{thm:regret}).

\section{Related Work}

\subsection{Fisher-information--based data selection and optimal experimental design} 
Classical optimal experimental design (OED) optimizes Fisher-information--based criteria at the set level, e.g., D-optimality (maximize $\text{det}(\mathcal{I})$) or A-optimality (minimize $\text{tr}(\mathcal{I}^{-1})$) \cite{kiefer1959,pukelsheim2006optimal}. Bayesian OED extends these ideas with priors and information-theoretic utilities \cite{ChalonerVerdinelli1995,lindley1956measure}. In generalized linear and nonlinear models (e.g., logistic), construction of optimal designs often requires iterative matrix updates and inversions \cite{ford1992use,pronzato2013design,allen2021near}, making them computationally heavy and non-decomposable in online settings.

Active learning offers related acquisition principles. Early model-based and information-theoretic approaches target entropy/information gain \cite{cohn1996active,mackay1992information}. Expected model change (often instantiated as expected gradient length) selects samples maximizing the anticipated parameter update \cite{settles2008analysis,settles2009active}, while recent deep methods embed per-sample gradients to achieve diverse, uncertainty-aware batches (e.g., BADGE) \cite{Ash2020Deep}. Submodularity has been leveraged to obtain near-optimal greedy selection in information-gain and facility-location style objectives \cite{guestrin2005near,sener2018active}.

In contrast, our work proposes a \emph{decomposable} max-kernel objective in gradient space (gradient coverage) that induces monotone submodularity and admits per-impression marginal utilities suitable for millisecond-latency auctions. 
Our total-uncertainty objective $G(\mathcal{S})$ equals an I-optimal (integrated prediction-variance) criterion on the test distribution, and it collapses to A-optimality when the test gradient covariance is isotropic (or whitened), i.e., when $J_\text{test}$ is proportional to the identity, linking our formulation directly to classical OED objectives.

\subsection{Auto-Bidding for Non-Truthful Mechanisms}

Auto-bidding methods translate advertiser goals into per-impression bids under auction and budget constraints. Foundational work on real-time bidding optimized value-centric objectives with predictive signals and pacing \cite{zhang2014optimal,aggarwal2024auto}. In non-truthful (first-price) environments, bid shading and primal--dual/online learning approaches have been proposed to handle market uncertainty and budget feasibility \cite{karlsson2021adaptive,zhou2021efficient}. Recent theory studies no-regret learning in repeated first-price auctions with budgets, including independence or structured feedback assumptions and discounted objectives \cite{han2020optimal,wang2023learning,ai2022no}.

These approaches predominantly optimize short-term campaign KPIs (clicks/engagements/conversions) and treat labels as arriving post-click, whereas our setting explicitly values impressions for their \emph{information} about the recommendation model. Methodologically, we couple a submodular, decomposable information surrogate with a dual-variable pacing scheme (shadow price $\lambda$) and prove sublinear regret and budget-feasibility guarantees in an online auction stream (Section~\ref{sec:analysis}). Practically, we address the missing-label obstacle at bid time via a confidence-gated, label-free gradient estimator (and a zeroth-order variant for black-box models), which is typically outside the scope of value-only auto-bidding. This reframes robustness from market-facing uncertainty to \emph{model-learning} robustness: avoiding the pollution of performance signals by acquiring traffic that most reduces model uncertainty, which in turn improves long-term organic outcomes.

\subsection{Zeroth-Order Gradient Estimation and Optimization}

Zeroth-Order gradient estimation computes the gradient without analytically computing the derivatives. There are two mainstream methods to of ZO gradient estimation, single-point methods and two-point methods \cite{liu2020primer}. Single-point methods \cite{dani2007price,jamieson2012query,shamir2013complexity,gasnikov2017stochastic,zhang2022new} query the function value only once at
each time step, making it suitable for online optimization and control problems. Recent advances improves the single-point method by reducing variance~\cite{chen2022improve,huang2024zeroth,chen2025regression}, reusing samples~\cite{wang2024relizo}, and estimating the Hessian~\cite{kim2025subspace}.
Two-point methods~\cite{nesterov2017random,duchi2015optimal,shamir2017optimal}, as the name suggests, query the function value twice in the same instantaneous time. Recent advances improve the two-point methods by extending them to non-convex non-smooth problems~\cite{lin2022gradient}, training deep models~\cite{pang2024stochastic,jiang2024zo}, distributed settings~\cite{yi2022zeroth}, etc.


Many classical lower bounds have been derived for zeroth-order SGD in both strongly convex and convex settings~\cite{jamieson2012query, agarwal2009information, raginsky2011information, duchi2015optimal, shamir2017optimal}, as well as in non-convex settings~\cite{wang2020zeroth}. More recently, \cite{wang2018stochastic, balasubramanian2018zeroth, cai2022zeroth} demonstrated that if the gradient has a low-dimensional structure, the query complexity scales linearly with the intrinsic dimension and logarithmically with the number of parameters. Additional techniques, such as sampling schedules~\cite{bollapragada2018adaptive} and other variance reduction methods~\cite{ji2019improved, liu2018zeroth}, can be incorporated into zeroth-order SGD. Recently, an optimizer designed for LLM, named MeZO~\cite{malladi2023fine}, adapting the classical ZO-SGD~\cite{spall1992multivariate} method to operate inplace, thereby fine-tuning LMs with the same memory footprint as inference. Recent work~\cite{gautam2024variance} modifies MeZO by applying variance reduction techniques.

\section{Conclusion} 
Paid promotion can unintentionally harm high-quality content by polluting engagement signals.
We addressed this by reframing promotion as strategic data acquisition: a dual-objective bidding framework that jointly optimizes short-term engagement value and long-term uncertainty reduction of the platform's CTR model.
Methodologically, we proposed a tractable surrogate for information gain with a provable link to optimal experimental design, enabling decomposable marginal utilities at impression granularity. 
We coupled this with a two-stage Lagrangian bidding scheme: campaign-level budget pacing via a dynamically updated shadow price, and impression-level bid optimization. 
Practically, we solved the missing-label challenge at bid time with a confidence-gated gradient heuristic and a zeroth-order estimator for black-box models.
Theoretically, we established monotone submodularity of the composite objective, a sublinear regret bound for online first-price CPM dual pacing, and an expected budget feasibility guarantee. 
Empirically, component-wise validations and end-to-end offline simulations on synthetic and real-world datasets demonstrated consistent improvements in final model AUC/LogLoss over standard baselines, stable budget adherence, and robustness when analytical gradients are unavailable.
Overall, our framework transforms paid promotion from impression-purchasing into principled, information-aware bidding that strengthens recommendation models and enhances long-term organic outcomes. 


\section*{Acknowledgments}
We thank the anonymous reviewers and our shepherd, for their valuable comments and suggestions.
This work was supported in part by China NSF grant No. 62322206, 62432007, 62132018, 62025204, U2268204, 62441236, 62272307, 62372296. The opinions, findings, conclusions, and recommendations expressed in this paper are those of the authors and do not necessarily reflect the views of the funding agencies or the government.

\bibliographystyle{apalike}
\bibliography{ref}

@article{zhang2014real,
  title={Real-time Bidding Benchmarking with iPinyou Dataset},
  author={Zhang, Weinan and Yuan, Shuai and Wang, Jun and Shen, Xuehua},
  journal={arXiv preprint arXiv:1407.7073},
  year={2014}
}

@inproceedings{DiemertMeynet2017,
    author = {{Diemert Eustache, Meynet Julien} and Galland, Pierre and Lefortier, Damien},
    title={Attribution Modeling Increases Efficiency of Bidding in Display Advertising},
    publisher = {ACM},
    booktitle = {Proceedings of the AdKDD and TargetAd Workshop},
    year = {2017}
}

@misc{Xiaohongshu,
  author = {{Xiaohongshu}},
  title = {What is Shutiao?},
  year = {2025},
  url = {https://help.reditorapp.com/yunying/liuliang/shutiao.html},
  urldate = {2025-04-11}
}

@article{maaten2008visualizing,
  title={Visualizing Data Using t-SNE},
  author={Maaten, Laurens van der and Hinton, Geoffrey},
  journal={Journal of Machine Learning Research},
  volume={9},
  number={Nov},
  pages={2579--2605},
  year={2008}
}

@article{agarwal2009information,
  title={Information-Theoretic Lower Bounds on the Oracle Complexity of Convex Optimization},
  author={Agarwal, Alekh and Wainwright, Martin J and Bartlett, Peter and Ravikumar, Pradeep},
  journal={Advances in Neural Information Processing Systems},
  volume={22},
  year={2009}
}

@article{raginsky2011information,
  title={Information-Based Complexity, Feedback and Dynamics in Convex Programming},
  author={Raginsky, Maxim and Rakhlin, Alexander},
  journal={IEEE Transactions on Information Theory},
  volume={57},
  number={10},
  pages={7036--7056},
  year={2011},
  publisher={IEEE}
}

@article{duchi2015optimal,
  title={Optimal Rates for Zero-Order Convex Optimization: The Power of Two Function Evaluations},
  author={Duchi, John C and Jordan, Michael I and Wainwright, Martin J and Wibisono, Andre},
  journal={IEEE Transactions on Information Theory},
  volume={61},
  number={5},
  pages={2788--2806},
  year={2015},
  publisher={IEEE}
}

@article{shamir2017optimal,
  title={An Optimal Algorithm for Bandit and Zero-Order Convex Optimization with Two-Point Feedback},
  author={Shamir, Ohad},
  journal={The Journal of Machine Learning Research},
  volume={18},
  number={1},
  pages={1703--1713},
  year={2017},
  publisher={JMLR. org}
}

@article{wang2020zeroth,
  title={Zeroth-Order Algorithms for Nonconvex Minimax Problems with Improved Complexities},
  author={Wang, Zhongruo and Balasubramanian, Krishnakumar and Ma, Shiqian and Razaviyayn, Meisam},
  journal={arXiv preprint arXiv:2001.07819},
  year={2020}
}

@inproceedings{wang2018stochastic,
  title={Stochastic Zeroth-Order Optimization in High Dimensions},
  author={Wang, Yining and Du, Simon and Balakrishnan, Sivaraman and Singh, Aarti},
  booktitle={Twenty-First Annual Conference on Artificial Intelligence and Statistics},
  pages={1356--1365},
  year={2018},
}

@article{balasubramanian2018zeroth,
  title={Zeroth-Order (Non)-Convex Stochastic Optimization via Conditional Gradient and Gradient Updates},
  author={Balasubramanian, Krishnakumar and Ghadimi, Saeed},
  journal={Advances in Neural Information Processing Systems},
  volume={31},
  year={2018}
}

@article{cai2022zeroth,
  title={Zeroth-Order Regularized Optimization (ZORO): Approximately Sparse Gradients and Adaptive Sampling},
  author={Cai, HanQin and Mckenzie, Daniel and Yin, Wotao and Zhang, Zhenliang},
  journal={SIAM Journal on Optimization},
  volume={32},
  number={2},
  pages={687--714},
  year={2022},
  publisher={SIAM}
}

@article{bollapragada2018adaptive,
  title={Adaptive Sampling Strategies for Stochastic Optimization},
  author={Bollapragada, Raghu and Byrd, Richard and Nocedal, Jorge},
  journal={SIAM Journal on Optimization},
  volume={28},
  number={4},
  pages={3312--3343},
  year={2018},
  publisher={SIAM}
}

@inproceedings{ji2019improved,
  title={Improved Zeroth-Order Variance Reduced Algorithms and Analysis for Nonconvex Optimization},
  author={Ji, Kaiyi and Wang, Zhe and Zhou, Yi and Liang, Yingbin},
  booktitle={International Conference on Machine Learning},
  pages={3100--3109},
  year={2019},
}

@article{liu2018zeroth,
  title={Zeroth-Order Stochastic Variance Reduction for Nonconvex Optimization},
  author={Liu, Sijia and Kailkhura, Bhavya and Chen, Pin-Yu and Ting, Paishun and Chang, Shiyu and Amini, Lisa},
  journal={Advances in Neural Information Processing Systems},
  volume={31},
  year={2018}
}

@article{malladi2023fine,
  title={Fine-Tuning Language Models with Just Forward Passes},
  author={Malladi, Sadhika and Gao, Tianyu and Nichani, Eshaan and Damian, Alex and Lee, Jason D and Chen, Danqi and Arora, Sanjeev},
  journal={arXiv preprint arXiv:2305.17333},
  year={2023}
}

@article{spall1992multivariate,
  title={Multivariate Stochastic Approximation Using a Simultaneous Perturbation Gradient Approximation},
  author={Spall, James C},
  journal={{IEEE Transactions on Automatic Control}},
  volume={37},
  number={3},
  pages={332--341},
  year={1992},
  publisher={IEEE}
}

@article{gautam2024variance,
  title={Variance-Reduced Zeroth-Order Methods for Fine-Tuning Language Models},
  author={Gautam, Tanmay and Park, Youngsuk and Zhou, Hao and Raman, Parameswaran and Ha, Wooseok},
  journal={arXiv preprint arXiv:2404.08080},
  year={2024}
}

@article{dani2007price,
  title={The Price of Bandit Information for Online Optimization},
  author={Dani, Varsha and Kakade, Sham M and Hayes, Thomas},
  journal={Advances in Neural Information Processing Systems},
  volume={20},
  year={2007}
}

@article{jamieson2012query,
  title={Query Complexity of Derivative-Free Optimization},
  author={Jamieson, Kevin G and Nowak, Robert and Recht, Ben},
  journal={Advances in Neural Information Processing Systems},
  volume={25},
  year={2012}
}

@inproceedings{shamir2013complexity,
  title={On the Complexity of Bandit and Derivative-Free Stochastic Convex Optimization},
  author={Shamir, Ohad},
  booktitle={Conference on Learning Theory},
  pages={3--24},
  year={2013},
  organization={PMLR}
}

@inproceedings{chen2022improve,
  title={Improve Single-Point Zeroth-Order Optimization using High-Pass and Low-Pass Filters},
  author={Chen, Xin and Tang, Yujie and Li, Na},
  booktitle={International Conference on Machine Learning},
  pages={3603--3620},
  year={2022},
}

@article{gasnikov2017stochastic,
  title={Stochastic Online Optimization. Single-Point and Multi-Point Non-Linear Multi-Armed Bandits. Convex and Strongly-Convex Case},
  author={Gasnikov, Alexander V and Krymova, Ekaterina A and Lagunovskaya, Anastasia A and Usmanova, Ilnura N and Fedorenko, Fedor A},
  journal={Automation and Remote Control},
  volume={78},
  number={2},
  pages={224--234},
  year={2017},
}

@article{zhang2022new,
  title={A New One-Point Residual-Feedback Oracle for Black-Box Learning and Control},
  author={Zhang, Yan and Zhou, Yi and Ji, Kaiyi and Zavlanos, Michael M},
  journal={Automatica},
  volume={136},
  pages={110006},
  year={2022},
  publisher={Elsevier}
}

@article{nesterov2017random,
  title={Random Gradient-Free Minimization of Convex Functions},
  author={Nesterov, Yurii and Spokoiny, Vladimir},
  journal={Foundations of Computational Mathematics},
  volume={17},
  number={2},
  pages={527--566},
  year={2017},
  publisher={Springer}
}

@article{liu2020primer,
  title={A Primer on Zeroth-Order Optimization in Signal Processing and Machine Learning: Principals, Recent Advances, and Applications},
  author={Liu, Sijia and Chen, Pin-Yu and Kailkhura, Bhavya and Zhang, Gaoyuan and Hero III, Alfred O and Varshney, Pramod K},
  journal={IEEE Signal Processing Magazine},
  volume={37},
  number={5},
  pages={43--54},
  year={2020},
  publisher={IEEE}
}

@article{wang2024relizo,
  title={Relizo: Sample Reusable Linear Interpolation-Based Zeroth-Order Optimization},
  author={Wang, Xiaoxing and Qin, Xiaohan and Yang, Xiaokang and Yan, Junchi},
  journal={Advances in Neural Information Processing Systems},
  volume={37},
  pages={15070--15096},
  year={2024}
}

@article{huang2024zeroth,
  title={Zeroth-Order Learning in Continuous Games via Residual Pseudogradient Estimates},
  author={Huang, Yuanhanqing and Hu, Jianghai},
  journal={IEEE Transactions on Automatic Control},
  year={2024},
  publisher={IEEE}
}

@article{kim2025subspace,
  title={Subspace-Based Approximate Hessian Method for Zeroth-Order Optimization},
  author={Kim, Dongyoon and Lee, Sungjae and Lee, Wonjin and Kim, Kwang In},
  journal={arXiv preprint arXiv:2507.06125},
  year={2025}
}

@article{chen2025regression,
  title={Regression-Based Single-Point Zeroth-Order Optimization},
  author={Chen, Xin and Ren, Zhaolin},
  journal={arXiv preprint arXiv:2507.04223},
  year={2025}
}

@article{lin2022gradient,
  title={Gradient-Free methods for Deterministic and Stochastic Nonsmooth Nonconvex Optimization},
  author={Lin, Tianyi and Zheng, Zeyu and Jordan, Michael},
  journal={Advances in Neural Information Processing Systems},
  volume={35},
  pages={26160--26175},
  year={2022}
}

@article{pang2024stochastic,
  title={Stochastic Two Points Method for Deep Model Zeroth-order Optimization},
  author={Pang, Yijiang and Zhou, Jiayu},
  journal={arXiv preprint arXiv:2402.01621},
  year={2024}
}

@article{yi2022zeroth,
  title={Zeroth-Order Algorithms for Stochastic Distributed Nonconvex Optimization},
  author={Yi, Xinlei and Zhang, Shengjun and Yang, Tao and Johansson, Karl H},
  journal={Automatica},
  volume={142},
  pages={110353},
  year={2022},
  publisher={Elsevier}
}

@inproceedings{jiang2024zo,
  title={ZO-AdamU optimizer: Adapting Perturbation by the Momentum and Uncertainty in Zeroth-Order Optimization},
  author={Jiang, Shuoran and Chen, Qingcai and Pan, Youcheng and Xiang, Yang and Lin, Yukang and Wu, Xiangping and Liu, Chuanyi and Song, Xiaobao},
  booktitle={Proceedings of the AAAI Conference on Artificial Intelligence},
  volume={38},
  number={16},
  pages={18363--18371},
  year={2024}
}

@article{aggarwal2024auto,
  title={Auto-Bidding and Auctions in Online Advertising: A Survey},
  author={Aggarwal, Gagan and Badanidiyuru, Ashwinkumar and Balseiro, Santiago R and Bhawalkar, Kshipra and Deng, Yuan and Feng, Zhe and Goel, Gagan and Liaw, Christopher and Lu, Haihao and Mahdian, Mohammad and others},
  journal={ACM SIGecom Exchanges},
  volume={22},
  number={1},
  pages={159--183},
  year={2024},
  publisher={ACM New York, NY, USA}
}

@inproceedings{wang2023learning,
  title={Learning to Bid in Repeated First-Price Auctions with Budgets},
  author={Wang, Qian and Yang, Zongjun and Deng, Xiaotie and Kong, Yuqing},
  booktitle={International Conference on Machine Learning},
  pages={36494--36513},
  year={2023},
}

@article{han2020optimal,
  title={Optimal No-Regret Learning in Repeated First-Price Auctions},
  author={Han, Yanjun and Weissman, Tsachy and Zhou, Zhengyuan},
  journal={Operations Research},
  volume={73},
  number={1},
  pages={209--238},
  year={2025},
  publisher={INFORMS}
}

@article{ai2022no,
  title={No-Regret Learning in Repeated First-Price Auctions with Budget Constraints},
  author={Ai, Rui and Wang, Chang and Li, Chenchen and Zhang, Jinshan and Huang, Wenhan and Deng, Xiaotie},
  journal={arXiv preprint arXiv:2205.14572},
  year={2022}
}

@inproceedings{zhang2014optimal,
  title={Optimal Real-Time Bidding for Display Advertising},
  author={Zhang, Weinan and Yuan, Shuai and Wang, Jun},
  booktitle={Proceedings of the 20th ACM SIGKDD International Conference on Knowledge Discovery and Data Mining},
  pages={1077--1086},
  year={2014}
}

@inproceedings{karlsson2021adaptive,
  title={Adaptive Bid Shading Optimization of First-Price Ad Inventory},
  author={Karlsson, Niklas and Sang, Qian},
  booktitle={American Control Conference},
  pages={4983--4990},
  year={2021},
  organization={IEEE}
}

@inproceedings{zhou2021efficient,
  title={An Efficient Deep Distribution Network for Bid Shading in First-Price Auctions},
  author={Zhou, Tian and He, Hao and Pan, Shengjun and Karlsson, Niklas and Shetty, Bharatbhushan and Kitts, Brendan and Gligorijevic, Djordje and Gultekin, San and Mao, Tingyu and Pan, Junwei and others},
  booktitle={Proceedings of the 27th ACM SIGKDD Conference on Knowledge Discovery \& Data Mining},
  pages={3996--4004},
  year={2021}
}

@misc{xhs, 
url={https://www.xiaohongshu.com}, 
title={Xiaohongshu}, 
author={Xiaohongshu},
year={2025},
}

@misc{tt,
    url={https://www.tiktok.com}, 
title={TikTok}, 
author={TikTok},
year={2025},
}

@misc{TikTok,
  author = {{TikTok}},
  title = {About Promote on TikTok},
  year = {2025},
  url = {https://ads.tiktok.com/help/article/about-promote-on-tiktok?lang=en},
  urldate = {2025-04-11}
}

@misc{mediastats,
  author = {{Bernard Marr \& Co.}},
  title = {How Much Data Do We Create Every Day? The Mind-Blowing Stats Everyone Should Read},
  year = {2025},
  url = {https://bernardmarr.com/how-much-data-do-we-create-every-day-the-mind-blowing-stats-everyone-should-read/},
  urldate = {2025-09-18}
}

@article{panda2022approaches,
  title={Approaches and Algorithms to Mitigate Cold Start Problems in Recommender Systems: a Systematic Literature Review},
  author={Panda, Deepak Kumar and Ray, Sanjog},
  journal={Journal of Intelligent Information Systems},
  volume={59},
  number={2},
  pages={341--366},
  year={2022},
  publisher={Springer}
}

@article{yuan2023user,
  title={User Cold Start Problem in Recommendation Systems: A Systematic Review},
  author={Yuan, Hongli and Hernandez, Alexander A},
  journal={IEEE access},
  volume={11},
  pages={136958--136977},
  year={2023},
  publisher={IEEE}
}

@inproceedings{grbovic2013large,
  title={Large Scale Ad Latency Analysis},
  author={Grbovic, Mihajlo and Malkin, Jon and Das, Hirakendu},
  booktitle={International Conference on Big Data},
  pages={762--767},
  year={2013},
  organization={IEEE}
}

@article{kiefer1959,
  title={Optimum Experimental Designs},
  author={Kiefer, Jack},
  journal={Journal of the Royal Statistical Society: Series B (Methodological)},
  volume={21},
  number={2},
  pages={272--304},
  year={1959},
  publisher={Wiley Online Library}
}

@article{ChalonerVerdinelli1995,
  title={Bayesian Experimental Design: A Review},
  author={Chaloner, Kathryn and Verdinelli, Isabella},
  journal={Statistical Science},
  pages={273--304},
  year={1995},
  publisher={JSTOR}
}

@inproceedings{guestrin2005near,
  title={Near-Optimal Sensor Placements in Gaussian Processes},
  author={Guestrin, Carlos and Krause, Andreas and Singh, Ajit Paul},
  booktitle={Proceedings of the 22nd International Conference on Machine Learning},
  pages={265--272},
  year={2005}
}

@inproceedings{
sener2018active,
title={Active Learning for Convolutional Neural Networks: A Core-Set Approach},
author={Ozan Sener and Silvio Savarese},
booktitle={International Conference on Learning Representations},
year={2018},
}

@inproceedings{settles2008analysis,
  title={An Analysis of Active Learning Strategies for Sequence Labeling Tasks},
  author={Settles, Burr and Craven, Mark},
  booktitle={Proceedings of the 2008 Conference on Empirical Methods in Natural Language Processing},
  pages={1070--1079},
  year={2008}
}

@article{settles2009active,
  title={Active Learning Literature Survey},
  author={Settles, Burr},
  year={2009},
  publisher={University of Wisconsin-Madison Department of Computer Sciences}
}

@inproceedings{
Ash2020Deep,
title={Deep Batch Active Learning by Diverse, Uncertain Gradient Lower Bounds},
author={Jordan T. Ash and Chicheng Zhang and Akshay Krishnamurthy and John Langford and Alekh Agarwal},
booktitle={International Conference on Learning Representations},
year={2020},
}

@inproceedings{freytag2014selecting,
  title={Selecting Influential Examples: Active Learning with Expected Model Output Changes},
  author={Freytag, Alexander and Rodner, Erik and Denzler, Joachim},
  booktitle={European Conference on Computer Vision},
  pages={562--577},
  year={2014},
  organization={Springer}
}

@article{lindley1956measure,
  title={On a Measure of the Information Provided by an Experiment},
  author={Lindley, Dennis V},
  journal={The Annals of Mathematical Statistics},
  volume={27},
  number={4},
  pages={986--1005},
  year={1956},
  publisher={Institute of Mathematical Statistics}
}

@article{ford1992use,
  title={The Use of a Canonical Form in the Construction of Locally Optimal Designs for Non-Linear Problems},
  author={Ford, Ian and Torsney, Bernard and Wu, CF Jeff},
  journal={Journal of the Royal Statistical Society Series B: Statistical Methodology},
  volume={54},
  number={2},
  pages={569--583},
  year={1992},
  publisher={Oxford University Press}
}

@article{pronzato2013design,
  title={Design of Experiments in Nonlinear Models},
  author={Pronzato, Luc and P{\'a}zman, Andrej},
  journal={Lecture Notes in Statistics},
  volume={212},
  number={1},
  year={2013},
  publisher={Springer}
}

@article{cohn1996active,
  title={Active Learning with Statistical Models},
  author={Cohn, David A and Ghahramani, Zoubin and Jordan, Michael I},
  journal={Journal of Artificial Intelligence Research},
  volume={4},
  pages={129--145},
  year={1996}
}

@article{mackay1992information,
  title={Information-Based Objective Functions for Active Data Selection},
  author={MacKay, David JC},
  journal={Neural computation},
  volume={4},
  number={4},
  pages={590--604},
  year={1992},
  publisher={MIT Press One Rogers Street, Cambridge, MA 02142-1209, USA journals-info~…}
}

@article{lu2024daved,
  title={DAVED: Data Acquisition via Experimental Design for Data Markets},
  author={Lu, Charles and Huang, Baihe and Karimireddy, Sai Praneeth and Vepakomma, Praneeth and Jordan, Michael and Raskar, Ramesh},
  journal={arXiv preprint arXiv:2403.13893},
  year={2024}
}

@article{allen2021near,
  title={Near-Optimal Discrete Optimization for Experimental Design: A Regret Minimization Approach},
  author={Allen-Zhu, Zeyuan and Li, Yuanzhi and Singh, Aarti and Wang, Yining},
  journal={Mathematical Programming},
  volume={186},
  number={1},
  pages={439--478},
  year={2021},
  publisher={Springer}
}

@article{dd80f5c8-fa20-3a07-bd73-ee4c02e5352c,
 ISSN = {00034851},
 author = {Jack Sherman and Winifred J. Morrison},
 journal = {The Annals of Mathematical Statistics},
 number = {1},
 pages = {124--127},
 publisher = {Institute of Mathematical Statistics},
 title = {Adjustment of an Inverse Matrix Corresponding to a Change in One Element of a Given Matrix},
 urldate = {2025-10-14},
 volume = {21},
 year = {1950}
}

@misc{Googlefirstprice,
  author = {{Google AdSense}},
  title = {Moving AdSense to a First-Price Auction},
  year = {2021},
  url = {https://blog.google/products/adsense/our-move-to-a-first-price-auction/},
  urldate = {2025-04-11}
}

@book{pukelsheim2006optimal,
  title={Optimal Design of Experiments},
  author={Pukelsheim, Friedrich},
  year={2006},
  publisher={SIAM}
}

@article{huan2024optimal,
  title={Optimal Experimental Design: Formulations and Computations},
  author={Huan, Xun and Jagalur, Jayanth and Marzouk, Youssef},
  journal={Acta Numerica},
  volume={33},
  pages={715--840},
  year={2024},
  publisher={Cambridge University Press}
}

@article{wittman2025fisher,
  title={Fisher Matrix for Beginners},
  author={Wittman, David},
  journal={arXiv preprint arXiv:2510.09683},
  year={2025}
}

@article{myerson1981optimal,
  title={Optimal Auction Design},
  author={Myerson, Roger B},
  journal={Mathematics of Operations Research},
  volume={6},
  number={1},
  pages={58--73},
  year={1981},
  publisher={INFORMS}
}

@article{owen1998strategic,
  title={Strategic Facility Location: A Review},
  author={Owen, Susan Hesse and Daskin, Mark S},
  journal={European Journal of Operational Research},
  volume={111},
  number={3},
  pages={423--447},
  year={1998},
  publisher={Elsevier}
}

@article{soland1974optimal,
  title={Optimal Facility Location with Concave Costs},
  author={Soland, Richard M},
  journal={Operations Research},
  volume={22},
  number={2},
  pages={373--382},
  year={1974},
  publisher={INFORMS}
}

\appendix
\section{Insight from a Toy Model}\label{sec:insight-toy-model}

We construct a theoretical model to demonstrate that high-variance (noisy) rewards prevent creators from improving their content, effectively trapping them in suboptimal local regions. We adopt the ``Try-Accept'' creator model from Yao et al.\ [60] but relax the unrealistic convexity assumption. We analyze the behavior in a general non-convex, $L$-smooth landscape.

\begin{definition}[Try-Accept Creator Strategy]
At step $t$, a creator with current content parameters $x_t$ generates a variation $x' = x_t + \delta_t$, where $\delta_t$ is sampled uniformly from a sphere of radius $r$, i.e., $\delta_t \sim \text{Unif}(r\mathbb{S}^{d-1})$. The creator observes a noisy reward $\tilde{R}(x)$. The strategy updates as follows:
\begin{equation*}
    x_{t+1} = \begin{cases} 
    x_t + \delta_t & \text{if } \tilde{R}(x_t + \delta_t) > \tilde{R}(x_t), \\
    x_t & \text{otherwise}.
    \end{cases}
\end{equation*}
\end{definition}

\noindent We define the assumptions for the reward landscape and the noise process.

\begin{assumption}[Non-Convex Landscape and Noise] \label{ass:nonconvex}
\hfill
\begin{enumerate}
    \item \textnormal{\textbf{L-Smoothness:}} The expected reward function $R(x)$ is differentiable and $L$-smooth. For all $x, y \in \mathbb{R}^d$, $\|\nabla R(x) - \nabla R(y)\| \le L \|x - y\|$.
    \item \textnormal{\textbf{Bounded Reward:}} The function $R(x)$ is bounded above by $R^*$.
    \item \textnormal{\textbf{Gaussian Reward Noise:}} The creator observes $\tilde{R}(x) = R(x) + \epsilon$, where $\epsilon \sim \mathcal{N}(0, \xi^2)$ is i.i.d.\ noise.
    \item \textnormal{\textbf{Small Step Size:}} The exploration radius $r$ is sufficiently small relative to the noise and smoothness, specifically $r \cdot \|\nabla R(x)\| \ll \xi$.
\end{enumerate}
\end{assumption}

\begin{theorem}[Non-Convex Convergence Rate with Noise] \label{thm:converge-try-accept}
Under Assumption~\ref{ass:nonconvex}, let $\Delta R = R^* - R(x_0)$. The expected average squared gradient norm (convergence to a stationary point) after $T$ steps satisfies:
\begin{equation*}
    \frac{1}{T} \sum_{t=0}^{T-1} \mathbb{E}\left[ \|\nabla R(x_t)\|^2 \right] \leq \mathcal{O}\left( \frac{\xi \cdot d \cdot \Delta R}{r^2 T} \right) + \mathcal{O}(L \cdot \xi \cdot d).
\end{equation*}
\end{theorem}


\begin{proof}
Let the true improvement at step $t$ be $\Delta_t = R(x_t + \delta_t) - R(x_t)$. Due to $L$-smoothness, we have the lower bound:
\begin{equation} \label{eq:smooth_bound}
    \Delta_t \ge \langle \nabla R(x_t), \delta_t \rangle - \frac{L}{2} \|\delta_t\|^2 = \langle \nabla R(x_t), \delta_t \rangle - \frac{L r^2}{2}.
\end{equation}
The update condition is $\tilde{R}(x_t+\delta_t) - \tilde{R}(x_t) > 0$. Let $\epsilon_{diff} = \epsilon' - \epsilon \sim \mathcal{N}(0, 2\xi^2)$. The step is accepted if:
\begin{equation*}
    \Delta_t + \epsilon_{diff} > 0.
\end{equation*}
The probability of acceptance, conditioned on the direction $\delta_t$, is:
\begin{equation*}
    \pi(\delta_t) = \mathbb{P}_{\epsilon}(\epsilon_{diff} > -\Delta_t) = \Phi\left( \frac{\Delta_t}{\sqrt{2}\xi} \right),
\end{equation*}
where $\Phi$ is the CDF of the standard normal distribution.
The expected reward increase at iteration $t$, taking expectation over both $\delta_t$ and noise, is:
\begin{equation*}
    \mathbb{E}[R_{t+1} - R_t] = \mathbb{E}_{\delta_t} \left[ \Delta_t \cdot \pi(\delta_t) \right] = \mathbb{E}_{\delta_t} \left[ \Delta_t \cdot \Phi\left( \frac{\Delta_t}{\sqrt{2}\xi} \right) \right].
\end{equation*}
Using Assumption~\ref{ass:nonconvex}.4 (high noise regime/small step), we approximate $\Phi(z) \approx \frac{1}{2} + \frac{z}{\sqrt{2\pi}}$ near zero.
Let $g_t = \nabla R(x_t)$. Approximating $\Delta_t \approx \langle g_t, \delta_t \rangle$ (ignoring the second-order term for the linear coefficient):
\begin{align*}
    \mathbb{E}[R_{t+1} - R_t] &\approx \mathbb{E}_{\delta_t} \left[ \Delta_t \left( \frac{1}{2} + \frac{\Delta_t}{2\sqrt{\pi}\xi} \right) \right] \\
    &= \frac{1}{2}\underbrace{\mathbb{E}_{\delta_t}[\Delta_t]}_{\text{Drift}} + \frac{1}{2\sqrt{\pi}\xi} \underbrace{\mathbb{E}_{\delta_t}[\Delta_t^2]}_{\text{Variance Signal}}.
\end{align*}
From Eq.~\eqref{eq:smooth_bound}, the first term (Drift) is bounded by $-\frac{L r^2}{2}$ (since $\mathbb{E}[\langle g, \delta \rangle]=0$).
For the second term, dominates by the first-order gradient term: $\mathbb{E}[\Delta_t^2] \approx \mathbb{E}[\langle g_t, \delta_t \rangle^2]$.
For a vector $\delta_t$ uniform on a sphere of radius $r$ in $d$ dimensions, $\mathbb{E}[\langle g_t, \delta_t \rangle^2] = \frac{r^2}{d} \|g_t\|^2$.
Combining these:
\begin{equation*}
    \mathbb{E}[R_{t+1} - R_t] \gtrsim - \frac{L r^2}{4} + \frac{1}{2\sqrt{\pi}\xi} \frac{r^2}{d} \|g_t\|^2.
\end{equation*}
Rearranging to bound the gradient norm:
\begin{equation*}
   \frac{r^2}{2\sqrt{\pi} d \xi} \|g_t\|^2 \lesssim \mathbb{E}[R_{t+1} - R_t] + \frac{L r^2}{4}.
\end{equation*}
Multiplying by $\frac{2\sqrt{\pi} d \xi}{r^2}$:
\begin{equation*}
    \|g_t\|^2 \lesssim \frac{C_1 d \xi}{r^2} \mathbb{E}[R_{t+1} - R_t] + C_2 L d \xi.
\end{equation*}
Summing over $t = 0 \dots T-1$ and dividing by $T$:
\begin{align*}
    \frac{1}{T} \sum_{t=0}^{T-1} \mathbb{E}[\|g_t\|^2] &\lesssim \frac{C_1 d \xi}{r^2 T} \sum_{t=0}^{T-1} \mathbb{E}[R_{t+1} - R_t] + C_2 L d \xi \\
    &= \frac{C_1 d \xi}{r^2 T} \mathbb{E}[R_T - R_0] + C_2 L d \xi.
\end{align*}
Since $R_T \le R^*$, the total gain is bounded by $R^* - R_0$. Thus, we arrive at the bound in the theorem statement.
\end{proof}

\noindent\textbf{Interpretation:} In non-convex optimization, convergence is measured by the gradient norm tending to zero (finding a local optimum). The bound contains two terms:
\begin{itemize}
    \item The first term decays with $T$, representing the optimization progress. Note that it scales linearly with noise $\xi$: higher noise slows down the learning rate.
    \item The second term, $\mathcal{O}(L \xi d)$, is a constant \textbf{Noise Floor}. As $T \to \infty$, the gradient norm does not vanish but hits this floor. The creator stops improving when the true gradient signal $\|\nabla R(x)\|$ becomes comparable to the noise-induced variance.
\end{itemize}
This proves that reducing platform uncertainty (decreasing $\xi$) is strictly necessary for creators to refine content beyond a coarse approximation.

We now establish the direct link between the uncertainty in the platform's CTR model and the variance of the reward signal, $\xi^2$, perceived by the creator. The creator's reward from natural traffic is proportional to the probability of their content winning the user-level auction, which depends on its pCTR, $\hat{\sigma}$. The pCTR can be modeled as the sum of the true underlying CTR, $\mu_i$, and a zero-mean noise term representing the model's uncertainty, whose variance is $\zeta_i^2$.

The creator's expected reward is a function of the probability that their pCTR is the highest among all competing items. Assuming the pCTR estimates $\{\hat{\sigma}_j\}$ for competing items follow independent Gaussian distributions, i.e., $\hat{\sigma}_j \sim \mathcal{N}(\mu_j, \zeta_j^2)$, this win probability for creator $i$ is:
\begin{equation*}
    P(\hat{\sigma}_i = \arg\max_j \hat{\sigma}_j) = \int_{-\infty}^{\infty} \left( \prod_{j \neq i} \Phi\left(\frac{y - \mu_j}{\zeta_j}\right) \right) \frac{1}{\zeta_i} \phi\left(\frac{y - \mu_i}{\zeta_i}\right) dy,
\end{equation*}
where $\phi(\cdot)$ and $\Phi(\cdot)$ are the PDF and CDF of the standard normal distribution, respectively.

This formulation reveals that the variance of the creator's reward is directly influenced by the variances $\{\zeta_j^2\}$ of the pCTR estimates. A higher $\zeta_i^2$, \textit{i.e.}, greater uncertainty in the model's prediction for the creator's own content—leads to a more volatile and unpredictable reward. Therefore, minimizing the CTR model's uncertainty (reducing $\zeta^2$) is equivalent to reducing the reward noise $\xi^2$ in Theorem~\ref{thm:converge-try-accept}, thereby facilitating more effective creator improvement.

\section{Additional Details for Section~\ref{sec:measurement}}\label{sec:measure-details}

To ensure a fair and controlled comparison in our empirical study (Figure 1), we applied a rigorous set of filtering criteria to sample both the promoted and organic posts from the platform. These criteria were designed to isolate the effect of the "Shutiao" promotion service on posts with a similar initial performance profile. The detailed sampling scope for each group is outlined below.

\subsection{Sampling Criteria for Promoted Posts (Treatment Group)}
A post was included in the treatment group if it met all of the following conditions:
\begin{itemize}
    \item \textbf{Promotion History:} The post must have been promoted \textbf{exactly once} using the "Shutiao" service. It must not have been promoted using any other competitive bidding ad products on the platform.
    \item \textbf{Campaign Budget:} The expenditure for the single "Shutiao" campaign must have been greater than or equal to 30 units.
    \item \textbf{Pre-Promotion Performance:} In the period from its publication until the day before the promotion began, the post must satisfy:
    \begin{itemize}
        \item In-feed clicks > 0 (to ensure it was not completely ignored by the organic system).
        \item Total impressions < 50,000.
        \item Total clicks < 10,000.
        \item Total engagements < 2,000.
    \end{itemize}
    \item \textbf{Account and Content Type:} The post must \textbf{not} be from an enterprise account, a product-centric post, or a post containing direct e-commerce affiliate links. This focuses the study on genuine creator content.
\end{itemize}

\subsection{Sampling Criteria for Organic Posts (Control Group)}
A post was included in the control group to serve as a comparable baseline if it met all of the following conditions:
\begin{itemize}
    \item \textbf{Promotion History:} The post must have \textbf{never} been promoted with "Shutiao" or any other competitive bidding ad products.
    \item \textbf{Pre-Sampling Performance:} To ensure the control group had a similar starting point to the treatment group, the post's performance from its publication until the day of sampling must satisfy:
    \begin{itemize}
        \item In-feed clicks > 0.
        \item Total impressions < 50,000.
        \item Total clicks < 10,000.
        \item Total engagements < 2,000.
    \end{itemize}
    \item \textbf{Account and Content Type:} Similar to the treatment group, the post must \textbf{not} be from an enterprise account, a product-centric post, or a post containing direct e-commerce affiliate links.
\end{itemize}

\section{On Parameter Alignment for Gradient Coverage}\label{sec:parameter-alignment}
In practice, mainstream content platforms continually retrain their pCTR models on a fixed cadence (for example, every few hours or once per day). 
Our bidding campaigns are scheduled to align with this cadence: a campaign’s horizon is the same as the continual-training period. 
Consequently, the gradients used by the gradient-coverage surrogate and the gradients estimated at bid time are both computed with respect to the same, current model snapshot (the ``anchor'' parameters at campaign start). 
This alignment eliminates the alleged mismatch between $\theta_0$ and $\theta_t$ during a campaign. 
Even in deployments with minor mid-campaign calibrations (such as bias correction or lightweight feature re-scaling), the Gaussian-kernel similarity in gradient space is robust to small parameter drift, and the validation gradient bank can be refreshed at the next campaign boundary. 
Empirically, we observe negligible differences between using a fixed snapshot versus recomputing within the campaign window, confirming that, under the standard industrial retraining schedule, parameter-point mismatch is not a practical concern.
\section{Adaptation to Second-Price Auctions}\label{sec:spa}

Our proposed framework is mechanism-agnostic regarding the valuation of impressions. The core components—the surrogate objective $U(\mathcal{S})$, the gradient coverage calculation, and the confidence-gated marginal utility estimation ($\Delta_t$)—quantify the intrinsic value of an impression to the model's learning process. This value exists independently of how the impression is auctioned.

While the main text focuses on the First-Price Auction (FPA) due to its prevalence in industry and the complexity of bid shading, our framework is easily adapted to Second-Price Auctions (SPA) (e.g., VCG mechanisms). In this section, we derive the optimal bidding strategy for SPA, showing that it results in a simplified linear bidding formula.

\subsection{Problem Formulation}
In a Second-Price Auction, the winner pays the market clearing price (the highest competing bid), denoted as $z_t$. The bidder does not know $z_t$ beforehand but knows its distribution or can treat it as a random variable. The goal remains to maximize the total utility (immediate value + uncertainty reduction) subject to a budget constraint $B$.

Let $x_t \in \{0, 1\}$ be the allocation variable (1 if we win, 0 otherwise). We win if our bid $b_t \ge z_t$.
The optimization problem is:
\begin{align*}
    \max_{\{b_t\}} \quad & \sum_{t=1}^T \Delta_t \cdot x_t(b_t) \\
    \text{s.t.} \quad & \sum_{t=1}^T z_t \cdot x_t(b_t) \le B
\end{align*}
where $\Delta_t$ is the estimated marginal utility derived in Section 3.2.

\subsection{Optimal Bidding Strategy}
We construct the Lagrangian with a dual variable $\lambda \ge 0$ (the shadow price of the budget):
\begin{equation*}
    \mathcal{L}(\mathbf{b}, \lambda) = \sum_{t=1}^T \Delta_t x_t - \lambda \left( \sum_{t=1}^T z_t x_t - B \right) = \lambda B + \sum_{t=1}^T (\Delta_t - \lambda z_t) x_t
\end{equation*}
To maximize the Lagrangian at step $t$, we should win the impression ($x_t=1$) if and only if the marginal contribution to the Lagrangian is non-negative:
\begin{equation*}
    \Delta_t - \lambda z_t \ge 0 \implies z_t \le \frac{\Delta_t}{\lambda}
\end{equation*}
In a Second-Price Auction, truth-telling is dominant with respect to the valuation. Here, our ``valuation'' is adjusted by the opportunity cost of the budget $\lambda$. To ensure we win exactly when the market price $z_t$ is below our threshold $\Delta_t/\lambda$, we simply submit this threshold as our bid:
\begin{equation} \label{eq:spa_bid}
    b_t^* = \frac{\Delta_t}{\lambda}
\end{equation}

\subsection{Algorithm Modifications}
Adapting the proposed two-stage framework to SPA requires two simple changes:

\begin{enumerate}
    \item \textbf{Stage 1 (Pacing):} The update rule for $\lambda$ (Equation \ref{eq:dual-update}) remains structurally the same. The dual variable $\lambda$ still increases if the budget is consumed too fast and decreases if consumed too slowly. However, the cost term in the update logic changes from ``our bid'' to ``the second price'' (market price).
    
    \item \textbf{Stage 2 (Bidding):} The complex inverse-shading optimization used in First-Price scenarios (finding $b_t$ to maximize surplus) is replaced by the closed-form linear scaling in Equation (\ref{eq:spa_bid}).
\end{enumerate}

\paragraph{Comparison.}
In the FPA setting defined in Section \ref{sec:two-stage-bidding}, the bidder must shade their bid $b_t < \Delta_t / \lambda$ to generate surplus, requiring estimation of the win probability curve $W_t(b)$. In the SPA setting, the strategy simplifies to bidding the marginal utility deflated by the shadow price. This confirms that our core contribution—accurately estimating $\Delta_t$ via Gradient Coverage—is robust and transferable across different auction mechanisms.
\section{Deferred Proof}

\subsection{Proof of Theorem~\ref{thm:surrogate-relation}}\label{sec:proof-surrogate-relation}
\begin{proof}
Fix $x\in\mathcal{D}_{\mathrm{val}}$ and let $z_x\in S$ be any nearest neighbor in gradient space, i.e.,
\[
z_x \in \arg\min_{z\in S} \|g(x)-g(z)\|^2,\qquad d_x(S):=\|g(x)-g(z_x)\|^2.
\]
By PSD order, $\mathcal{I}_\gamma(S) \succeq \gamma I_d + g(z_x)g(z_x)^\top$, hence 
$\mathcal{I}_\gamma(S)^{-1} \preceq \big(\gamma I_d + g(z_x)g(z_x)^\top\big)^{-1}$. 
By Sherman--Morrison \cite{dd80f5c8-fa20-3a07-bd73-ee4c02e5352c},
\[
\big(\gamma I_d + u u^\top\big)^{-1}
= \frac{1}{\gamma}I_d - \frac{1}{\gamma^2}\,\frac{u u^\top}{1+\frac{1}{\gamma}\|u\|^2},
\]
so with $u=g(z_x)$ we get
\begin{align*}
g(x)^\top \mathcal{I}_\gamma(S)^{-1} g(x)
&\le 
g(x)^\top \big(\gamma I_d + g(z_x)g(z_x)^\top\big)^{-1} g(x)\\
&= \frac{1}{\gamma}\|g(x)\|^2 - \frac{1}{\gamma^2}\,\frac{\langle g(x),g(z_x)\rangle^2}{1+\frac{1}{\gamma}\|g(z_x)\|^2}.
\end{align*}
Using the identity 
$\|g(x)-g(z_x)\|^2=\|g(x)\|^2+\|g(z_x)\|^2-2\langle g(x),g(z_x)\rangle$
and rearranging,
\[
\langle g(x),g(z_x)\rangle 
= \frac{\|g(x)\|^2+\|g(z_x)\|^2 - d_x(S)}{2}.
\]
Therefore, on the subset 
\[
A_\tau \;:=\; \big\{x\in\mathcal{D}_{\mathrm{val}}:\; d_x(S)\le \tau\big\},
\]
and using the bounds $\|g(x)\|\le L$ and $\|g(z_x)\|\ge m$, we obtain for $x\in A_\tau$,
\[
\langle g(x),g(z_x)\rangle^2 
\;\ge\; \left(\frac{m^2+m^2-\tau}{2}\right)^2 
= \frac{(2m^2-\tau)^2}{4}.
\]
Also, $1+\tfrac{1}{\gamma}\|g(z_x)\|^2 \le 1+\tfrac{L^2}{\gamma}$. Summing the per-$x$ bound over 
$\mathcal{D}_{\mathrm{val}}$ yields
\begin{align*}
G_\gamma(S)
&=\sum_{x} g(x)^\top \mathcal{I}_\gamma(S)^{-1} g(x)\\
&\le \frac{1}{\gamma}\sum_{x}\|g(x)\|^2 \;-\; 
\frac{1}{\gamma^2}\sum_{x\in A_\tau}\frac{\langle g(x),g(z_x)\rangle^2}{1+\tfrac{1}{\gamma}\|g(z_x)\|^2}\\
&\le \frac{k L^2}{\gamma} \;-\; 
\frac{(2m^2-\tau)^2}{4\,\gamma^2\,\big(1+\tfrac{L^2}{\gamma}\big)}\cdot |A_\tau|.
\end{align*}
It remains to relate $|A_\tau|$ to $U_\lambda(S)$. For any $x\in A_\tau$, we have 
$\exp(-\lambda d_x(S))\ge \exp(-\lambda\tau)$, and for $x\notin A_\tau$, 
$\exp(-\lambda d_x(S))\le \exp(-\lambda\tau)$. Hence
\[
U_\lambda(S)
= \sum_{x} \exp\big(-\lambda d_x(S)\big)
\le |A_\tau|\cdot 1 \;+\; (k-|A_\tau|)\cdot e^{-\lambda\tau}.
\]
Rearranging gives
\[
|A_\tau| \;\ge\; \frac{U_\lambda(S) - k\,e^{-\lambda\tau}}{1 - e^{-\lambda\tau}}.
\]
Substituting this lower bound on $|A_\tau|$ into the previous inequality completes the proof.
\end{proof}

\subsection{Proof of Theorem~\ref{thm:submodular}}\label{sec:proof-submodular}
\begin{proof}
Noticing that $V(S)$ is an additive function, to prove the submodularity of $F(S)$, we only need to prove that $U(S)$ is submodular.

\textbf{Step 1: Decompose $U(S)$:}  
The function is defined as:  
\[
U(S) = \sum_{x \in \mathcal{D}_{\text{val}}} \max_{z \in S}  \exp\left(-\lambda \left\| g_{\theta_0}(x) - g_{\theta_0}(z) \right\|^2 \right).
\]  
Define a \textit{similarity kernel} \(k(x, z) = \exp\left(-\lambda \left\| g_{\theta_0}(x) - g_{\theta_0}(z) \right\|^2 \right)\). Then, for each validation point \(x\), define:  
\[
f_x(S) = \max_{z \in S}  k(x, z).
\]  
Thus, \(U(S) = \sum_{x \in \mathcal{D}_{\text{val}}} f_x(S)\). Since a sum of submodular functions is submodular, it suffices to prove that each \(f_x(S)\) is submodular for fixed \(x\).  

\textbf{Step 2: Prove $f_x(S)$ is Submodular:}
Fix a validation point \(x \in \mathcal{D}_{\text{val}}\). We show \(f_x(S)\) is submodular. For any \(A \subseteq B \subseteq \mathcal{D}_{\text{train}}\) and \(v \in \mathcal{D}_{\text{train}} \setminus B\):  
\[
f_x(A \cup \{v\}) - f_x(A) \geq f_x(B \cup \{v\}) - f_x(B).
\]  

Case 1: \(k(x, v) \leq \max_{z \in B} k(x, z)\)  
Since \(A \subseteq B\), we have \(\max_{z \in A} k(x, z) \leq \max_{z \in B} k(x, z)\).  
Right-hand side:
  \[
  f_x(B \cup \{v\}) - f_x(B) = \max\left\{ \max_{z \in B} k(x, z),  k(x, v) \right\} - \max_{z \in B} k(x, z) = 0,
  \]  
  because \(k(x, v) \leq \max_{z \in B} k(x, z)\).  
Left-hand side:
  \[
  f_x(A \cup \{v\}) - f_x(A) = \max\left\{ \max_{z \in A} k(x, z),  k(x, v) \right\} - \max_{z \in A} k(x, z) \geq 0,
  \]  
  since the \(\max\) can only increase or stay the same.  
Thus, \(0 \geq 0\) holds.  

Case 2: \(k(x, v) > \max_{z \in B} k(x, z)\)  
Since \(A \subseteq B\), \(\max_{z \in A} k(x, z) \leq \max_{z \in B} k(x, z) < k(x, v)\).  
Right-hand side:
  \[
  f_x(B \cup \{v\}) - f_x(B) = k(x, v) - \max_{z \in B} k(x, z).
  \]  
Left-hand side:
  \[
  f_x(A \cup \{v\}) - f_x(A) = k(x, v) - \max_{z \in A} k(x, z).
  \]  
Since \(\max_{z \in A} k(x, z) \leq \max_{z \in B} k(x, z)\), we have:  
\[
k(x, v) - \max_{z \in A} k(x, z) \geq k(x, v) - \max_{z \in B} k(x, z),
\]  
so the inequality holds.  

Case 3: \(k(x, v) > \max_{z \in A} k(x, z)\) but \(k(x, v) \leq \max_{z \in B} k(x, z)\)  
Right-hand side:
  \[
  f_x(B \cup \{v\}) - f_x(B) = 0 \quad (\text{since } k(x, v) \leq \max_{z \in B} k(x, z)).
  \]  
Left-hand side:
  \[
  f_x(A \cup \{v\}) - f_x(A) = k(x, v) - \max_{z \in A} k(x, z) > 0.
  \]  
Thus, \(> 0 \geq 0\) holds.  

In all cases, \(f_x(A \cup \{v\}) - f_x(A) \geq f_x(B \cup \{v\}) - f_x(B)\).  

\end{proof}

\subsection{Proof of Theorem~\ref{thm:regret}}\label{sec:proof-regret}
\begin{proof}
Define the convex per-round loss $\ell_t(\lambda):=-f_t(\lambda)$.
Since $f_t(\lambda)$ is the supremum of affine functions in $\lambda$, $\ell_t$ is convex, and a valid subgradient at $\lambda_{t-1}$ is 
\[
\partial \ell_t(\lambda_{t-1})\;\ni\;\frac{h_t}{C},\qquad 
\text{where } h_t=W_t(b_t)\,b_t\in[0,b_{\max}]\subseteq[0,C].
\]
With the negative-entropy mirror map on $\lambda>0$, the multiplicative-weights update is
$\lambda_t=\lambda_{t-1}\exp\!\big(\eta\,h_t/C\big)$.
Standard online mirror descent (OMD) analysis for scalar $\lambda$ with negative entropy (see, e.g., MWU regret bounds) yields
\[
\sum_{t=1}^T \Big(\ell_t(\lambda_{t-1}) - \ell_t(\lambda^*)\Big)
\;\le\; 
\frac{\log(\lambda_{\max}/\lambda_0)}{\eta}
\;+\;
\frac{\eta}{2}\sum_{t=1}^T\Big(\tfrac{h_t}{C}\Big)^2
\;\le\;
\frac{\log(\lambda_{\max}/\lambda_0)}{\eta}
\;+\;
\frac{\eta T}{2},
\]
where we used $h_t\le C$.
Rearranging and using $\ell_t=-f_t$ gives
\begin{equation}
\label{eq:ft-gap}
\sum_{t=1}^T f_t(\lambda_{t-1})
\;\ge\;
\sum_{t=1}^T f_t(\lambda^*)
\;-\;
\frac{\log(\lambda_{\max}/\lambda_0)}{\eta}
\;-\;
\frac{\eta T}{2}.
\end{equation}
By the envelope theorem, $b_t\in\arg\max_b W_t(b)\,(\Delta_t-\lambda_{t-1}b)$ implies
\[
f_t(\lambda_{t-1}) \;=\; W_t(b_t)\,\big(\Delta_t - \lambda_{t-1}\,b_t\big).
\]
Therefore,
\begin{equation}
\label{eq:alg-decomp}
\sum_{t=1}^T W_t(b_t)\,\Delta_t 
\;=\; \sum_{t=1}^T f_t(\lambda_{t-1}) \;+\; \sum_{t=1}^T \lambda_{t-1}\,h_t.
\end{equation}
A standard OMD inequality with negative entropy yields
\begin{equation}
\label{eq:lambda-cross}
\sum_{t=1}^T h_t\,(\lambda_{t-1}-\lambda^*)
\;\le\;
\frac{\log(\lambda_{\max}/\lambda_0)}{\eta}
\;+\;
\frac{\eta}{2}\sum_{t=1}^T\Big(\tfrac{h_t}{C}\Big)^2 C^2
\;\le\;
\frac{\log(\lambda_{\max}/\lambda_0)}{\eta}
\;+\;
\frac{\eta T\,C^2}{2}.
\end{equation}
Using \eqref{eq:lambda-cross}, we obtain
\[
\sum_{t=1}^T \lambda_{t-1}\,h_t
\;\ge\;
\lambda^*\sum_{t=1}^T h_t
\;-\;
\frac{\log(\lambda_{\max}/\lambda_0)}{\eta}
\;-\;
\frac{\eta T\,C^2}{2}.
\]
Combining this with \eqref{eq:ft-gap} and \eqref{eq:alg-decomp} yields
\[
\sum_{t=1}^T W_t(b_t)\,\Delta_t 
\;\ge\;
\sum_{t=1}^T f_t(\lambda^*)
\;+\;
\lambda^*\sum_{t=1}^T h_t
\;-\;
\frac{2\log(\lambda_{\max}/\lambda_0)}{\eta}
\;-\;
\eta T\Big(\frac{1}{2}+\frac{C^2}{2}\Big).
\]
By weak duality for the budget-constrained offline optimum,
\[
\mathrm{OPT}
\;\le\;
\sum_{t=1}^T f_t(\lambda^*) \;+\; \lambda^* B.
\]
Therefore,
\[
\sum_{t=1}^T W_t(b_t)\,\Delta_t 
\;\ge\;
\mathrm{OPT}
\;-\;
\lambda^*\Big(B-\textstyle\sum_{t=1}^T h_t\Big)
\;-\;
\frac{2\log(\lambda_{\max}/\lambda_0)}{\eta}
\;-\;
\eta T\Big(\frac{1}{2}+\frac{C^2}{2}\Big).
\]
Finally, with $\eta=\sqrt{\frac{\log(\lambda_{\max}/\lambda_0)}{T}}\frac{1}{C}$, the error terms are $O\!\big(C\,\sqrt{T\,\log(\lambda_{\max}/\lambda_0)}\big)$, giving the stated bound after taking expectations.
\end{proof}

\subsection{Proof of Theorem~\ref{thm:budget-feas}}\label{sec:proof-budget-feas}
\begin{proof}
First, we consider the update of the dual variable.
$$ \lambda_t = \lambda_{t-1} \exp\left(\eta \frac{W_a(b_t) b_t}{B}\right) \implies \log \lambda_t = \log \lambda_{t-1} + \eta \frac{W_a(b_t) b_t}{B}.$$
Telescoping from $1$ to $T$, 
$$\log \lambda_T - \log \lambda_0 = \frac{\eta}{B} \sum_{t=1}^T W_a(b_t) b_t. $$

Then, we bound the expenditure. Since $\lambda_T\leq \lambda_{\max}$ and $\mathbb{E}[\mathbf{1}_{\text{win}_t}|b_t]=W_a(b_t)$:
$$\sum_t W_a(b_t) b_t = \frac{B}{\eta} \log \frac{\lambda_T}{\lambda_0} \leq \frac{B}{\eta} \log \frac{\lambda_{\max}}{\lambda_0}. $$
Taking expectation:
$$\mathbb{E}\left[\sum_t b_t \mathbf{1}_{\text{win}t}\right] = \sum_t \mathbb{E}[W_a(b_t) b_t] \leq \frac{B}{\eta} \log \frac{\lambda_{\max}}{\lambda_0}. $$

\end{proof}

\section{Details in Experiments}\label{sec:exp-set-app}
\paragraph{Experimental Setup of Section~\ref{sec:exp-surrogate-relationship}}
\label{sec:exp-set-app-surrogate-relationship}

We randomly generate $2000$ samples with $20$ dimensions for binary classification using \texttt{sklearn} as the original dataset. Then, randomly choose $500$ samples for the initial training dataset, another disjoint $500$ samples as the test dataset, and another disjoint $500$ samples as the valid set. The left samples are the candidate dataset to be chosen by the algorithms. Each algorithm chooses $50$ sampels from the candidates.

\paragraph{Experimental Setup for Section~\ref{sec:exp-budget-feas}}
\label{sec:exp-set-app-budget-feas}

We simulate a repeated auction for $5000$ times. In each auction, there is assumed to be one competitor whose bid is sampled uniformly from $[0,1]$. The value of the bidder is assumed to be $1.5$, and the auction mechanism follows the first-price auction.

\paragraph{Experimental Setup for Section~\ref{sec:exp-grad-est}}
\label{sec:exp-set-app-grad-est}
We randomly generate $500$ samples for model training and another $1500$ i.i.d. samples as the test set. For each time of ZO gradient estimation, we set $\mu=0.01$ and uses $10$ data samples as a batch.

\paragraph{Experimental Setup for Section~\ref{sec:exp-end-to-end}}
\label{sec:exp-set-app-end-to-end}
To evaluate the end-to-end performance of our framework in a realistic, budget-constrained setting, we conduct a comprehensive offline simulation. The experiment utilizes a dataset (either the public Criteo dataset or a synthetic one) that is chronologically partitioned into four disjoint sets: an initial training set ($\mathcal{D}_{\text{init}}$) to establish a baseline model, a fixed validation set ($\mathcal{D}_{\text{val}}$) used for computing the uncertainty surrogate $U(\mathcal{S})$, a large auction stream ($\mathcal{D}_{\text{auc}}$) for the bidding simulation, and a final held-out test set ($\mathcal{D}_{\text{test}}$) for evaluation. First, an initial pCTR model, a standard multi-layer perceptron (MLP), is trained on $\mathcal{D}_{\text{init}}$. We then simulate a sequence of first-price auctions by streaming impressions one-by-one from $\mathcal{D}_{\text{auc}}$. In each auction, all bidding agents compete against each other and a simulated market price to win the impression, subject to an identical total budget, $B$. We evaluate five distinct strategies: (i) our full \textbf{Proposed} method with a balanced objective ($\beta=0.5$); (ii-iii) two ablative variants, \textbf{Value-Only} ($\beta=1$) and \textbf{Uncertainty-Only} ($\beta=0$), to isolate the effects of the utility components; and (iv-v) two standard industry baselines, \textbf{pCTR-Linear} and \textbf{Uniform Bidding}. Crucially, to simulate a practical black-box scenario where direct model gradients are inaccessible, all methods leveraging our framework utilize a Zeroth-Order (ZO) estimator to approximate the necessary loss gradients for the utility calculation. Upon completion of the auction stream, for each bidder, we create an augmented dataset by combining its set of won impressions ($\mathcal{S}_{\text{won}}$) with the initial training set $\mathcal{D}_{\text{init}}$. A new model is then retrained from scratch on this augmented data. The ultimate effectiveness of each strategy is measured by the AUC and LogLoss of its corresponding retrained model on the final, unseen test set $\mathcal{D}_{\text{test}}$, which represents future data.

For reproducibility, we specify the precise hyperparameter configurations used throughout our simulation. The pCTR model is an MLP with two hidden layers of sizes 128 and 64, respectively, each followed by a ReLU activation and a dropout layer with a rate of 0.3. All models are trained using the Adam optimizer with a learning rate of $10^{-3}$ for 5 epochs and a batch size of 1024. 
For the synthesis dataset, the simulation runs over an auction stream of 6,00 impressions, with each bidding agent allocated an identical total budget of $B=6,00$, while for the Critero dataset, the number of impressions is $2000$ and the budget is $2000$.
The pacing controller for our framework-based methods updates the dual variable $\lambda$ every 100 auctions. For our main \textbf{Proposed} method, the objective's balancing hyperparameter is set to $\beta=0.5$. The budget pacing mechanism is configured with an initial dual variable $\lambda_0=0.01$ and a learning rate $\eta=0.1$ for its multiplicative updates. The Gaussian kernel in the uncertainty surrogate $U(\mathcal{S})$ uses a bandwidth parameter of $\lambda_{\text{kernel}}=0.1$. The Zeroth-Order (ZO) gradient estimator, which is critical for our black-box setting, is configured with a smoothing parameter $\mu=0.01$ and averages its estimate over 5 random direction vectors per computation. The baseline methods are configured as follows: the \textbf{Uniform Bidding} agent places a constant bid of 20.0, and the \textbf{pCTR-Linear} agent uses a base multiplier of 45.0 for its bids. All experiments were conducted using the PyTorch framework on a GPU-accelerated machine.

\subsection{Additional Experimental Results}\label{sec:exp-set-app-results}
In this section, we supplement the end-to-end offline case study presented in Section~5.4 by providing enlarged visualizations of the experimental results.
Due to space constraints in the main text, some details in Figure~7 and Figure~8 were obscured. 
Figure~\ref{fig:syn-app} and Figure~\ref{fig:critero-app} present the detailed performance metrics for the Synthesis and Criteo datasets, respectively. 
These figures offer a clearer view of the final model performance (AUC and LogLoss), the cumulative spending curves demonstrating budget feasibility, and the dynamic evolution of the dual variable $\lambda$ (shadow price) over the course of the auction stream
\begin{figure}
    \centering
    \includegraphics[width=0.8\textwidth]{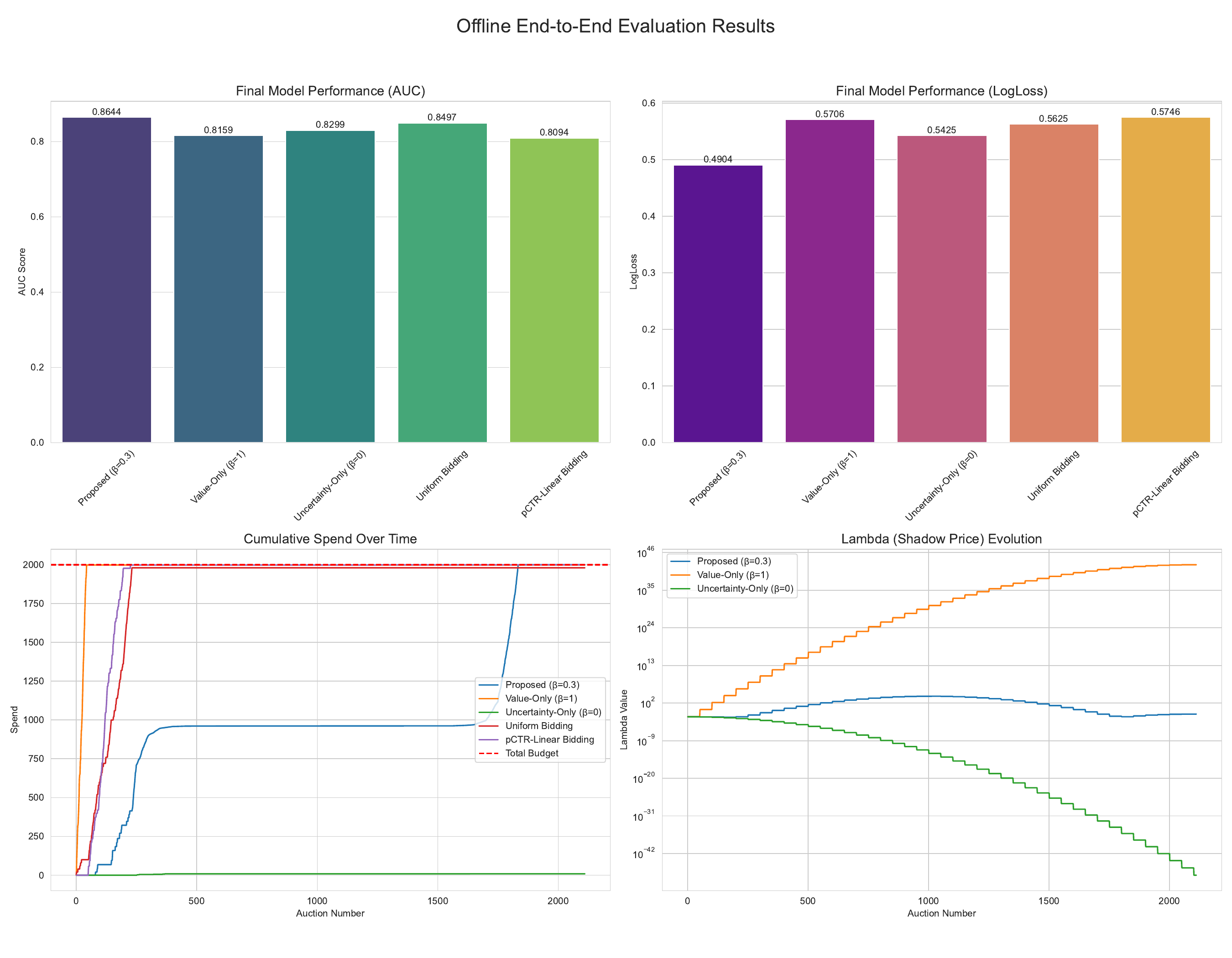}
        \caption{
        Offline Evaluation on Synthesis Dataset.
        }
        \label{fig:syn-app}
\end{figure}
\begin{figure}
    \centering
    \includegraphics[width=0.8\textwidth]{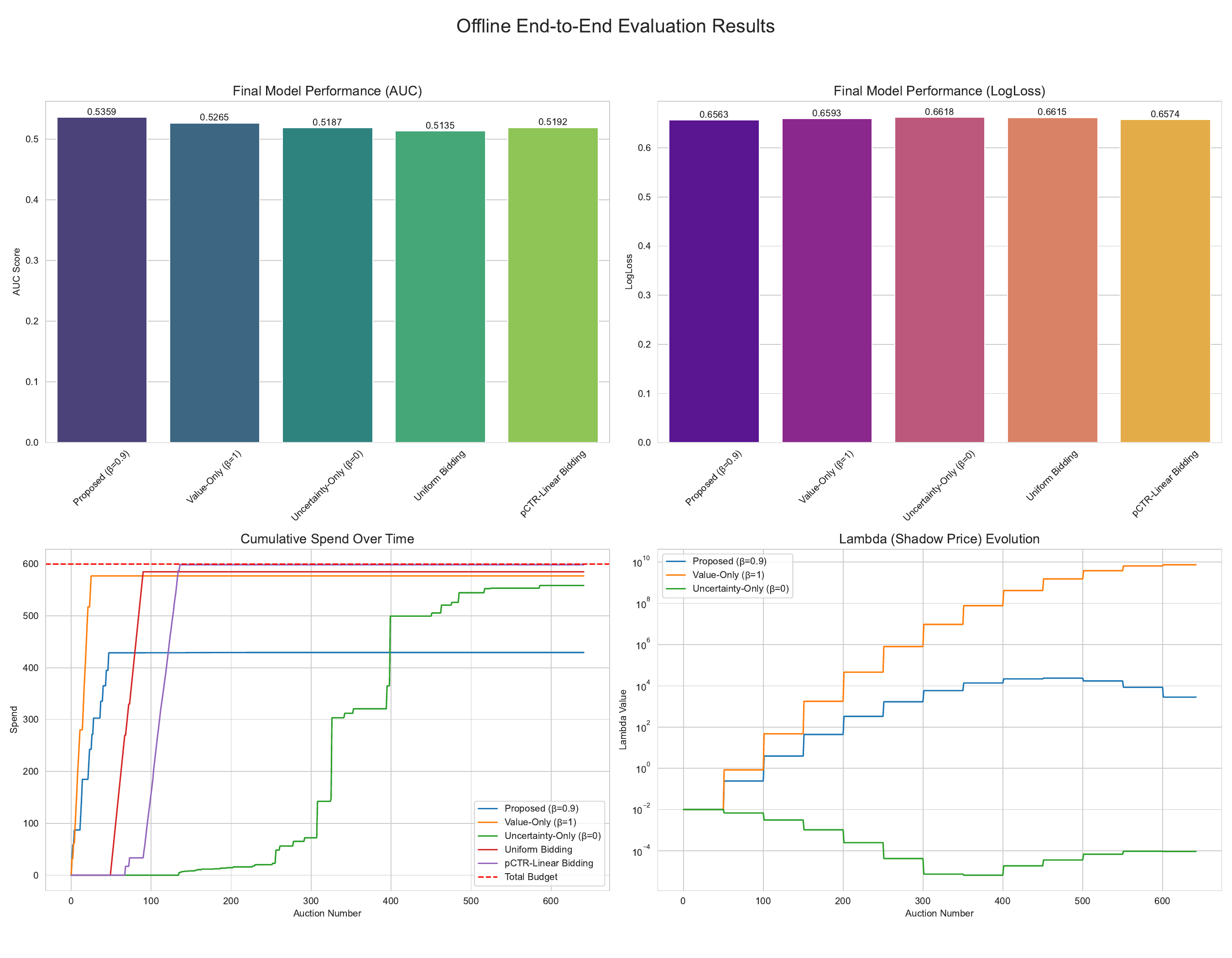}
        \caption{
        Offline Evaluation on Critero Dataset.
        }
        \label{fig:critero-app}
\end{figure}



\end{document}